\documentclass{article}
\usepackage[utf8]{inputenc}
\usepackage{booktabs}
\usepackage{threeparttable}
\usepackage{amsmath, amssymb, amsfonts,amsthm}
\usepackage{bm}
\usepackage{mathrsfs}
\usepackage{xr}
\usepackage{tablefootnote}
\usepackage{tabularx}
\usepackage{graphicx}
\usepackage{float}
\usepackage{color}
\usepackage{booktabs}
\usepackage{natbib}
\usepackage{color}
\usepackage{fullpage}
\usepackage{algorithm}
\usepackage{algpseudocode}
\usepackage{multirow}
\usepackage[colorlinks,citecolor=blue,linkcolor=blue]{hyperref}
 \usepackage{algpseudocode}
\newtheorem{remark}{Remark}
\newtheorem{lemma}{Lemma}
\newtheorem{assumption}{Assumption}
\newtheorem{theorem}{Theorem}

\title{Conditional Rank-Rank Regression via Deep Conditional Transformation Models}
\author{Xiaoyi Wang, Long Feng and Zhaojun Wang\\
Nankai University}
\date{\today}

\begin{document}

\maketitle
\begin{abstract}
Intergenerational mobility quantifies the transmission of socio-economic outcomes from parents to children. While rank-rank regression (RRR) is standard, adding covariates directly (RRRX) often yields parameters with unclear interpretation. Conditional rank-rank regression (CRRR) resolves this by using covariate-adjusted (conditional) ranks to measure within-group mobility. We improve and extend CRRR by estimating conditional ranks with a deep conditional transformation model (DCTM) and cross-fitting, enabling end-to-end conditional distribution learning with structural constraints and strong performance under nonlinearity, high-order interactions, and discrete ordered outcomes where the distributional regression used in traditional CRRR may be cumbersome or prone to misconfiguration. We further extend CRRR to discrete outcomes via an $\omega$-indexed conditional-rank definition and study sensitivity to $\omega$. For continuous outcomes, we establish an asymptotic theory for the proposed estimators and verify the validity of exchangeable bootstrap inference. Simulations across simple/complex continuous and discrete ordered designs show clear accuracy gains in challenging settings. Finally, we apply our method to two empirical studies, revealing substantial within-group persistence in U.S. income and pronounced gender differences in educational mobility in India.

\end{abstract}

\section{Introduction}
Intergenerational mobility is a central topic in economics and sociology. It quantifies how socio-economic status---such as income, education, occupation, and health---is transmitted (or reshuffled) from parents to children, and thus provides a window into the mechanisms underlying inequality, opportunity, and social development. A leading empirical tool is the rank-rank regression (RRR) \citep{dahl2008association,chetty2014land,chetty2014united}. In its canonical form, one transforms the parental and child outcomes into ranks and regresses the child rank on the parent rank; the slope coefficient serves as a popular measure of intergenerational persistence (and hence mobility). RRR has been widely used in empirical work on education \citep{kotera2017educational,chetty2017mobility}, labor and occupations \citep{song2020long,boar2021occupational}, income and wealth \citep{chetty2014land,chetty2014united,fagereng2021wealthy}, health \citep{halliday2021intergenerational,fletcher2021intergenerational}, and child development \citep{chetty2016effects}, and has become a benchmark for comparing mobility across regions, cohorts, and social groups, informing policy discussions on inequality and opportunity \citep{mogstad2023family}.

RRR is attractive for both statistical and substantive reasons. For continuous outcomes, the RRR slope is closely connected to rank dependence: it equals the Spearman's rank correlation in the population \citep{chetverikov2023inference}, a standard dependence measure that is invariant to monotone transformations \citep{spearman1961proof,kendall1948rank}. Moreover, by operating on ranks, RRR is typically more robust to extreme outcomes than classic mobility measures based on levels, such as the intergenerational elasticity (IGE) \citep{becker1986human,atkinson1980intergenerational,mazumder2018intergenerational,olivetti2015name,beller2006intergenerational}.

In many applications, however, researchers seek to account for observed covariates $X$ (e.g., region, race, parental education) and to distinguish within-group persistence from between-group persistence. A common practice is to incorporate $X$ directly in the rank regression, yielding RRR with covariates (RRRX). Yet the coefficient from RRRX is difficult to interpret: it generally no longer corresponds to a rank correlation measure and may even fall outside the natural range $[-1,1]$ \citep{chetverikov2023inference}. This interpretability concern motivates the conditional rank-rank regression (CRRR) recently proposed by \cite{chernozhukov2024conditional}. CRRR replaces marginal ranks with conditional ranks computed within covariate-defined group, and its slope admits a transparent interpretation as an average within-group persistence: it can be represented as the average Spearman's rank correlation conditional on $X$, averaged over the distribution of $X$ \citep{chernozhukov2024conditional}. Consequently, CRRR enables a decomposition of overall persistence into within- and between-group components, facilitating more granular mobility analysis.

Implementing CRRR requires estimating conditional ranks, i.e., conditional distribution functions $F_{Y|X}$ and $F_{W|X}$. When $X$ is low-dimensional and discrete with sufficient within-group sample sizes, empirical conditional CDFs are feasible. In modern applications, $X$ is often continuous and/or high-dimensional, making group-based estimation impractical. \cite{chernozhukov2024conditional} propose estimating conditional ranks via distribution regression (DR) \citep{foresi1995conditional,chernozhukov2013inference}, which approximates the conditional CDF by fitting a collection of threshold-specific binary regressions (often with logit/probit link functions) and then interpolating across thresholds. While DR provides an operational pathway, it faces several limitations in complex data environments: (i) link-function specifications impose implicit shape restrictions, which can lead to misspecification under heavy tails, multimodality, strong nonlinearity, or high-order interactions; (ii) accurately approximating the full distribution may require many thresholds, increasing computational and engineering burden; and (iii) independently fitted thresholds do not automatically enforce the probability axioms for a CDF (e.g., monotonicity), often necessitating post-processing. In addition, existing CRRR theory is primarily developed for continuously distributed outcomes, whereas many variables of interest (e.g., education categories, occupational ranks) are discrete and ordered, inducing pervasive ties. For discrete outcomes, ranks are not uniquely defined without explicit tie-handling, and \cite{chernozhukov2024conditional} explicitly note that discrete extensions are left for future work.

This paper advances CRRR along two dimensions. First, we develop a flexible estimator of conditional ranks based on a Deep Conditional Transformation Model (DCTM), combined with cross-fitting to mitigate overfitting bias. DCTM learns the conditional CDF end-to-end and, through architectural constraints, directly enforces validity of the estimated distribution, avoiding repeated threshold-by-threshold fits as in DR. Second, we provide a first systematic investigation of CRRR for discrete ordered outcomes. We introduce an $\omega$-indexed definition of conditional ranks that parameterizes tie-handling and study how the CRRR slope depends on $\omega$, highlighting that mobility conclusions for discrete outcomes are intrinsically sensitive to rank definitions and must be reported accordingly.

Our main contributions are summarized as follows.
\begin{enumerate}
    \item {Methodology.} We propose a ``DCTM + cross-fitting'' procedure for estimating conditional ranks in CRRR, which is computationally efficient and robust in settings with nonlinearity, high-order interactions, and discrete ordered outcomes, improving upon the DR-based implementation \citep{foresi1995conditional,chernozhukov2013inference,chernozhukov2024conditional}.
    \item {Theory (continuous outcomes).} Under a fixed-complexity framework, we establish consistency and asymptotic normality of the proposed estimator and prove validity of exchangeable bootstrap inference \citep{chernozhukov2024conditional,chetverikov2023inference}.
    \item {Discrete CRRR.} We develop a parameterized conditional-rank definition for discrete outcomes and quantify sensitivity of the target parameter to tie-handling, filling an important gap highlighted in \cite{chernozhukov2024conditional}.
    \item {Empirics and simulations.} Extensive simulations demonstrate accuracy gains under complex continuous and discrete ordered designs. Empirical applications to PSID-SHELF and IHDS illustrate how CRRR decomposes persistence into within- and between-group components and uncovers gender heterogeneity in mobility patterns.
\end{enumerate}

The remainder of the paper is organized as follows. Section 2 reviews the key theoretical background, including rank-rank regression, conditional rank-rank regression, and deep conditional transformation models (DCTM). Section 3 presents the proposed framework and model construction, replacing the traditional distribution regression approach for conditional-rank estimation with DCTM and adopting a cross-fitting strategy. Section 4 establishes consistency and asymptotic normality of the proposed estimator under a fixed model-complexity regime. Section 5 reports simulation studies to highlight its advantages relative to the traditional distribution-regression-based CRRR. Section 6 presents empirical applications to the PSID-SHELF and IHDS datasets. Finally, Section 7 concludes and outlines directions for further improvements. The Appendix collects proofs of the lemmas and theorems in Section 4.

\section{Basic Concepts}\label{chpt:method}

To facilitate subsequent exposition, we adopt the following unified notation throughout this section:
\begin{itemize}
    \item $Y$ denotes the child's socio-economic outcome (e.g., income, educational attainment, occupational status), and $W$ denotes the parent's outcome.
    \item $X \in \mathcal{X} \subset \mathbb{R}^p$ is a covariate vector (e.g., race, region, family characteristics).
    \item $F_Y(y)=\mathbb{P}(Y \leq y)$ and $F_W(w)=\mathbb{P}(W \leq w)$ are the marginal distribution functions of $Y$ and $W$, respectively.
    \item $F_{Y \mid X}(y \mid x)=\mathbb{P}(Y \leq y \mid X=x)$ and $F_{W \mid X}(w \mid x)=\mathbb{P}(W \leq w \mid X=x)$ are the conditional distribution functions of $Y\mid X$ and $W\mid X$, respectively.
\end{itemize}
Based on these definitions, we define the marginal ranks and conditional ranks:
\[
R_Y:=F_Y(Y), \quad R_W:=F_W(W), \quad U:=F_{Y\mid X}(Y \mid X), \quad V:=F_{W \mid X}(W \mid X),
\]
where, in the continuous-outcome case, $R_Y$, $R_W$, $U$, and $V$ are all distributed as $\mathrm{Unif}(0,1)$.

\subsection{Rank-Rank Regresssion}

Rank-rank regression is a rank-based regression approach for studying the dependence between two variables and has been widely used in economics, especially in intergenerational mobility research. The basic idea is to transform the child outcome $Y$ and the parent outcome $W$ into their population ranks $R_Y$ and $R_W$, and then use a linear model to describe the dependence of $R_Y$ on $R_W$. The slope $\rho$ is the parameter of interest:
\[
R_{Y,i}=\alpha+\rho R_{W,i}+\varepsilon_i.
\]

When $Y$ and $W$ are continuous, $R_Y \sim \mathrm{Unif}(0,1)$ and $R_W \sim \mathrm{Unif}(0,1)$, and
\[
\mathbb{E}[R_Y]=\mathbb{E}[R_W]=\frac{1}{2},\qquad 
\operatorname{Var}(R_Y)=\operatorname{Var}(R_W)=\frac{1}{12}.
\]
Therefore, the RRR coefficient $\rho$ is both the regression slope of $R_Y$ on $R_W$ and the Spearman's rank correlation between $Y$ and $W$ (equivalently, the Pearson correlation between $R_Y$ and $R_W$):
\[
\rho
=\frac{\operatorname{Cov}(R_Y,R_W)}{\operatorname{Var}(R_W)}
=\frac{\operatorname{Cov}(R_Y,R_W)}{\operatorname{Var}(R_Y)}
=\operatorname{Cor}(R_Y,R_W)
=12\,\mathbb{E}\!\left[\left(R_Y-\tfrac{1}{2}\right)\left(R_W-\tfrac{1}{2}\right)\right].
\]

In intergenerational mobility studies, the rank-rank regression slope $\rho$ is typically viewed as a measure of intergenerational persistence. A larger $\rho$ indicates a stronger association between parental rank and child rank, i.e., higher persistence and lower mobility. Because $\rho$ has an intuitive economic interpretation and is relatively robust to extreme values compared with alternative measures, rank-rank regression is widely used to quantify intergenerational dependence and to compare mobility across groups and regions.

It is important to note that the rank-rank regression model involves only the rank relationship between two variables. It captures aggregate intergenerational mobility, namely the Spearman's rank correlation between $Y$ and $W$ without controlling for other factors (sometimes referred to as the ``unconditional'' rank correlation). In many applications, however, researchers wish to study the association after controlling for additional covariates, such as within-group mobility. A conventional practice is to include covariates $X$ additively or non-additively in the rank regression, referred to as RRR with covariates (RRRX), for example,
\[
R_Y=\beta_0+\rho R_W+\beta_2^\top X+\epsilon,\qquad 
\mathbb{E}[\epsilon]=0,\ \mathbb{E}[R_W\epsilon]=0,\ \mathbb{E}[X\epsilon]=0.
\]
However, \cite{chernozhukov2024conditional} show that the RRRX coefficient $\rho$ is difficult to interpret both economically and statistically: it may fall outside the natural range $[-1,1]$ and no longer coincides with Spearman's rank correlation. This undermines the intuitive role of rank correlation as a dependence measure. Motivated by this issue, \cite{chernozhukov2024conditional} propose conditional rank-rank regression, which incorporates covariates while preserving a rank-correlation interpretation.

\subsection{Conditional Rank-Rank Regression}
To address the interpretability difficulty of rank-rank regression with covariates, \cite{chernozhukov2024conditional} propose conditional rank-rank regression (CRRR) and develop asymptotic theory for the continuous-outcome case.

The key idea of CRRR is to replace the marginal ranks $R_Y$ and $R_W$ with the conditional ranks $U:=F_{Y\mid X}(Y\mid X)$ and $V:=F_{W\mid X}(W\mid X)$, thereby incorporating covariate effects. Informally, CRRR computes within-group ranks of $Y$ and $W$ conditional on $X$, and then regresses the child's conditional rank on the parent's conditional rank:
\begin{equation*}
U_i=\alpha+\rho_C V_i+\varepsilon_i.
\end{equation*}

When $Y$ and $W$ are continuous, one has $U\mid X \sim \mathrm{Unif}(0,1)$ and $V\mid X \sim \mathrm{Unif}(0,1)$, and hence $U,V \sim \mathrm{Unif}(0,1)$. In particular,
\begin{align*}
\mathbb{E}[V]=\mathbb{E}[U]=\mathbb{E}[V\mid X]=\mathbb{E}[U\mid X]=\frac{1}{2},\qquad
\operatorname{Var}(V)=\operatorname{Var}(U)=\operatorname{Var}(V\mid X)=\operatorname{Var}(U\mid X)=\frac{1}{12}.
\end{align*}
Therefore, the CRRR coefficient $\rho_C$ is both the regression slope of $U$ on $V$ and the (unconditional) correlation between $U$ and $V$:
\begin{equation}\label{eq:spearman_rho_C}
\rho_C
=\frac{\operatorname{Cov}(U,V)}{\operatorname{Var}(V)}
=\frac{\operatorname{Cov}(U,V)}{\operatorname{Var}(U)}
=\operatorname{Cor}(U,V).
\end{equation}
Since the mean and variance are known constants, $\rho_C$ can also be written as
\begin{equation}\label{eq:half}
\rho_C
=12\,\mathbb{E}\!\left[\left(U-\tfrac{1}{2}\right)\left(V-\tfrac{1}{2}\right)\right].
\end{equation}
Moreover, because $\operatorname{Cov}(\mathbb{E}[U\mid X],\mathbb{E}[V\mid X])=0$, the law of total covariance implies
\[
\operatorname{Cov}(U,V)=\mathbb{E}\!\left[\operatorname{Cov}(U,V\mid X)\right].
\]
Since $U\mid X$ and $V\mid X$ have constant means and variances, it follows that $\rho_C$ is the $X$-average of the conditional Spearman's rank correlation between $Y$ and $W$:
\[
\rho_C=\mathbb{E}\!\left[\rho_{Y,W\mid X}\right],\qquad 
\rho_{Y,W\mid X}:=\operatorname{Cor}(U,V\mid X).
\]

Consequently, $\rho_C$ always lies in $[-1,1]$ and retains an interpretation analogous to the unconditional Spearman rank correlation, namely \emph{average within-group intergenerational persistence}. By contrast, the RRRX slope generally does not correspond to a rank correlation and suffers from interpretability issues, whereas CRRR preserves the rank-correlation interpretive framework by construction.

CRRR also provides a more refined lens for mobility analysis. The RRR slope $\rho$ measures overall persistence, while the CRRR slope $\rho_C$ measures average within-group persistence. Their difference can be interpreted as \emph{between-group persistence}, capturing the extent to which cross-group heterogeneity contributes to overall dependence. For example, in intergenerational income studies, $\rho_C$ reflects the association between parent and child income ranks within group defined by gender or family background, while $\rho$ reflects the overall association; their gap can be attributed to differences across such group.

To estimate $\rho_C$, one needs to estimate the conditional ranks $U$ and $V$, i.e., the conditional distribution functions $F_{Y\mid X}$ and $F_{W\mid X}$, and then plug the estimates into \eqref{eq:spearman_rho_C} or \eqref{eq:half} to obtain a sample estimator $\widehat{\rho}_C$. \cite{chernozhukov2024conditional} propose estimating conditional ranks via distribution regression (DR). Specifically, for a grid of thresholds $y$, construct binary responses $\mathbb{I}\{Y\le y\}$ and fit $\mathbb{P}(Y\le y\mid X)$ using logit or probit links, thereby approximating the entire conditional CDF by threshold-by-threshold fitting. Concretely, to estimate the conditional ranks for $Y$, choose quantile-based thresholds $\{y^{(m)}\}_{m=1}^M$ (e.g., an equally spaced grid from the 2\% to 98\% quantiles), and for each $y^{(m)}$ fit a logit/probit model for $\mathbb{P}(Y\le y^{(m)}\mid X)$ with response $\mathbb{I}(Y\le y^{(m)})$ and covariates $X$. Conditional ranks for values of $y$ between adjacent grid points are obtained by linear interpolation. To ensure reasonable tail behavior, local extrapolation is applied outside the grid range: if $y$ is below the minimum threshold, extrapolate using the fitted model at the smallest threshold and map the result back to $[0,1]$ through the link function; similarly extrapolate above the maximum threshold using the model at the largest threshold. Applying the same procedure to $W$ yields $\widehat V$. Substituting $\widehat U$ and $\widehat V$ into \eqref{eq:spearman_rho_C} or \eqref{eq:half} yields the following three sample estimators:
\begin{align}
\label{eq:ols}
\widehat \rho_C
&:=\frac{\sum_{i=1}^n\left(\widehat{U}_i-\overline{\widehat{U}}\right)\left(\widehat{V}_i-\overline{\widehat{V}}\right)}
{\sum_{i=1}^n\left(\widehat{V}_i-\overline{\widehat{V}}\right)^2},\\
\label{eq:corr}
\widetilde \rho_C
&:=\frac{\sum_{i=1}^n\left(\widehat{U}_i-\overline{\widehat{U}}\right)\left(\widehat{V}_i-\overline{\widehat{V}}\right)}
{\sqrt{\sum_{i=1}^n\left(\widehat{U}_i-\overline{\widehat{U}}\right)^2\,
\sum_{i=1}^n\left(\widehat{V}_i-\overline{\widehat{V}}\right)^2}},\\
\label{eq:corr_1_2}
\breve{\rho}_C
&:=\frac{12}{n}\sum_{i=1}^n\left(\widehat{U}_i-\tfrac{1}{2}\right)\left(\widehat{V}_i-\tfrac{1}{2}\right).
\end{align}
Here \eqref{eq:ols} is the OLS-slope form, while \eqref{eq:corr} and \eqref{eq:corr_1_2} are correlation-type forms. The three estimators are equivalent if and only if $Y$ and $W$ are continuous. When $Y$ and $W$ are discrete, we take the OLS-slope estimator $\widehat\rho_C$ in \eqref{eq:ols} as the default.

However, because DR estimates the CDF by threshold-by-threshold fitting, it implicitly assumes correct model specification at each threshold. While DR can deliver consistent distribution estimates under favorable conditions and sufficient data, misspecification risk can arise in complex settings with high-dimensional $X$, nonlinear interactions, and discrete ordered outcomes; this is also confirmed in our simulation studies. Therefore, we need alternative conditional-rank estimation methods that are more flexible, robust, accurate, and broadly applicable.

\subsection{Deep Conditional Transformation Model} \label{sec:dctm}
The deep conditional transformation model (DCTM) is a general method for directly modeling conditional distributions by combining classical transformation models with modern deep neural networks \citep{baumann2021deep,kook2024estimating}. It provides a flexible framework for estimating the conditional CDF $F_{Y\mid X}(y\mid x)$ without imposing stringent parametric assumptions on the distributional shape, while allowing architectural design to enforce the required monotonicity and boundedness properties of a valid distribution function.

Transformation models form a general framework that includes many common statistical models as special cases, such as linear regression and logistic regression. Historically, the development of transformation models can be traced from the Box-Cox transformation paradigm \citep{box1964analysis}, to the most likely transformation (MLT) framework \citep{hothorn2018most}, to conditional transformation models (CTM) for modeling conditional distributions \citep{hothorn2014conditional}, and finally to DCTM by integrating CTM with modern deep learning \citep{baumann2021deep,kook2024estimating}. In particular, the conditional transformation model proposed by \cite{hothorn2014conditional} is built on the idea that there exists a transformation function $h(y;x)$ that is nondecreasing in $y$ and maps the response $Y$ (given covariates $X=x$) to a latent variable $Z$ with a known baseline distribution:
\[
h(Y;x)=Z,\qquad Z\sim F_0,
\]
where $F_0$ is a chosen baseline CDF, e.g., $Z\sim\mathcal{N}(0,1)$. This assumption implies
\begin{equation*}
\begin{aligned}
\mathbb{P}(Y\le y\mid X=x)
&=\mathbb{P}\!\bigl(h(Y;x)\le h(y;x)\mid X=x\bigr)\\
&=\mathbb{P}\!\bigl(Z\le h(y;x)\bigr)
=F_0\!\bigl(h(y;x)\bigr).
\end{aligned}
\end{equation*}
In other words, instead of modeling the CDF of $Y\mid X$ directly, one models the transformation function $h(y;x)$, i.e., how $Y$ is transformed to follow the baseline distribution. It is crucial to impose appropriate constraints so that, for each fixed $x$, $h(y;x)$ is nondecreasing in $y$. The function $h(y;x)$ can be estimated, for example, via maximum likelihood, and the estimator of $F_{Y\mid X}$ follows from the identity above.

Classical conditional transformation models nonetheless have practical boundaries. When the covariate dimension is high, the relationship between variables exhibits strong nonlinearity and high-order interactions, or covariates include unstructured information such as images and text, the estimators relying on strong structural assumptions may be inadequate. The DCTM proposed by \cite{baumann2021deep} is designed for conditional distribution modeling in such complex scenarios by combining the interpretable structure of transformation models with the flexibility of deep learning. The DCTM can be summarized as
\[
h(y;x)=T\bigl(y;f(x)\bigr),
\]
where $T$ is a transformation function that is nondecreasing in $y$, and its form is governed by the neural network output $f(x)$. In practice, $T$ is typically modeled as either the sum or the product of a $y$-function and an $x$-dependent shift component. For instance, in the no-interaction case one may use $h(y;x)=\psi(y)-m(x)$, where $\psi(y)$ is a baseline transformation depending only on $y$, and $m(x)$ captures covariate effects. More generally, DCTM allows for interactions, meaning that covariates can affect not only location but also the shape of the transformation:
\[
h(y;x)=h_1+h_2
=\boldsymbol{a}(y)^{\top}\boldsymbol{\vartheta}(x)+\beta(x),
\]
where $\boldsymbol{a}(y)$ is a pre-specified basis function $\boldsymbol{a}:\Xi\mapsto\mathbb{R}^{M+1}$ with $\Xi$ the sample space, and $\boldsymbol{\vartheta}:\chi_{\vartheta}\mapsto\mathbb{R}^{M+1}$ is a parameter function defined on $\chi_{\vartheta}\subseteq \chi$. The function $\boldsymbol{\vartheta}$ governs the shape of the $y$-transformation, while $\beta(x)$ is a covariate-dependent distributional shift; both are learned and output by neural networks.

To implement this idea, DCTM typically uses two network modules: a \emph{shift} prediction module that predicts the transformation shift $\beta(x)$, and an \emph{interaction} prediction module that outputs $\boldsymbol{\vartheta}(x)$, which is multiplied by the basis $\boldsymbol{a}(y)$ to adjust the transformation shape. Benefiting from deep learning’s representation power, the network can automatically learn feature selection and cross-feature interactions in $X$. A schematic DCTM architecture is shown in Figure~\ref{fig:dctm_fig}.
\begin{figure}[htbp]
    \centering
    \includegraphics[width=\textwidth]{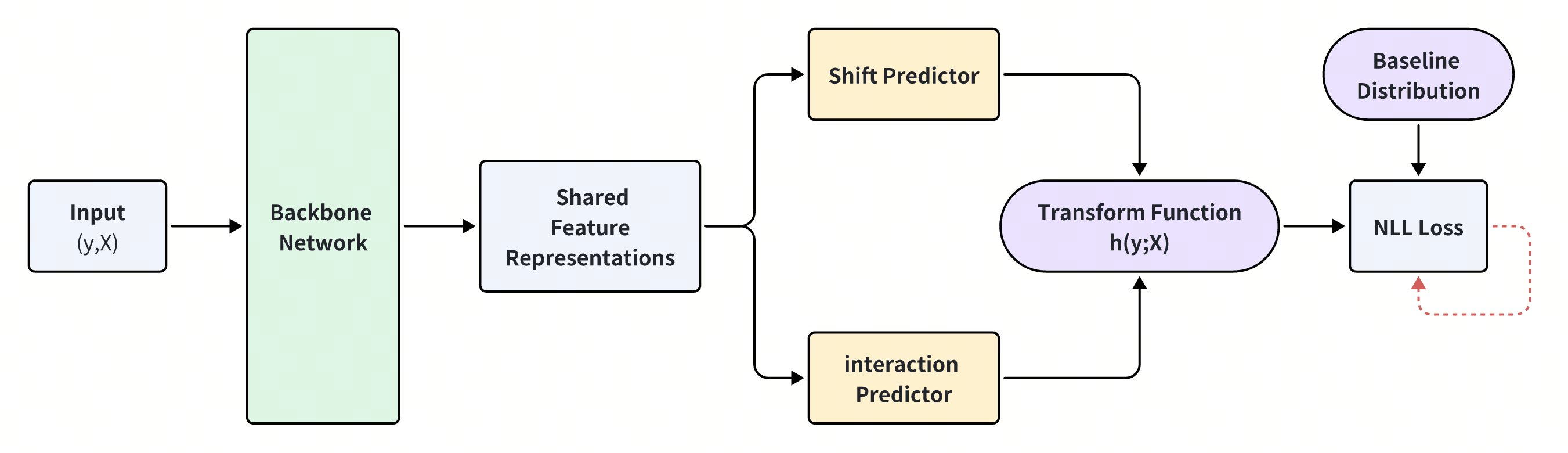}
    \caption[]{Schematic architecture of the DCTM network.}
    \label{fig:dctm_fig}
\end{figure}

Moreover, to ensure that for each fixed $x$ the transformation function $h(y;x)$ is nondecreasing in $y$, one needs to impose constraints on the network architecture or parameters. For example, one may enforce ordering constraints on output-layer weights or nonnegativity constraints on certain parameters. Such constraints guarantee structurally that the estimated $\widehat F_{Y\mid X}$ is always a valid distribution function satisfying the basic probability axioms.

DCTM is trained by maximizing the log-likelihood or, equivalently, minimizing the negative log-likelihood loss in \eqref{eq:nll}, where $\boldsymbol{\omega}$ denotes learnable network parameters and $\ell\bigl(\boldsymbol{\omega};y_i,\boldsymbol{x}_i\bigr)$ denotes the single-observation log-likelihood function \citep{baumann2021deep,kook2024estimating}. Since $F_{Y\mid X}(y\mid x)=F_0(h(y;x))$, the log-likelihood is straightforward to write down, and the model can be trained using standard optimizers (e.g., SGD \citep{bottou2012stochastic}, Adam \citep{kingma2014adam}, AdamW \citep{loshchilov2017decoupled}):
\begin{equation}\label{eq:nll}
\mathrm{NLL}\bigl(\boldsymbol{\omega};y_i,x_i\bigr)
:=-\frac{1}{n}\sum_{i=1}^{n}\ell\bigl(\boldsymbol{\omega};y_i,x_i\bigr).
\end{equation}

\section{Methodology and Model Construction}

\subsection{Estimating conditional ranks via DCTM}
A key step of conditional rank-rank regression (CRRR) is the estimation of \emph{conditional ranks}.
Accurate conditional-rank estimation directly improves the estimation of the target parameter $\rho_C$.
The classical CRRR approach relies on \emph{distribution regression} (DR): one selects a grid
$\{y^{(m)}\}_{m=1}^M$ and fits $P(Y\le y^{(m)}\mid X)$ separately for each $m$, and then obtains
conditional-rank estimates for all observations via interpolation and local extrapolation.

DR suffers from several limitations:
(i) fitting many pointwise binary regressions is computationally expensive and does not guarantee that the
resulting estimated distribution satisfies basic probabilistic constraints (e.g., monotonicity in $y$) or
global coherence;
(ii) boundary handling is complicated and often requires ad local extrapolation;
(iii) parametric link functions may be misspecified and can fail to adapt to nonlinearities, strong interactions,
and discrete/heavy-tailed outcomes;
(iv) extensive manual feature engineering is often required, which is costly and inefficient.

We replace DR by \emph{deep conditional transformation models} (DCTM) to estimate conditional ranks.
DCTM resolves the above issues because: a single end-to-end model outputs the entire conditional distribution
without repeated pointwise fits and is structurally constrained to yield a valid CDF; boundary points do not
require special treatment; neural networks provide strong representation power for high-dimensional and complex
data types without manual feature engineering. Moreover, DCTM is a unified framework: depending on the response
type, one may flexibly design the transformation structure, neural architecture, and monotonicity constraints.
Below we describe our DCTM designs for continuous and discrete outcomes; the network structures used in our
numerical experiments follow this design.

\subsubsection{Continuous outcomes}
The core idea of DCTM is to learn, via a neural network, a monotone nondecreasing transformation
$T(y,x)$ such that the transformed variable $Z=T(Y,X)\mid X$ approximately follows a chosen baseline distribution.
In the continuous-outcome simulations, we set the baseline distribution to $\mathcal{N}(0,1)$. For the child outcome $Y$, we model
\begin{equation}\label{eq:dctm_continuous_cdf}
F_{Y\mid X}(y\mid x)=\Phi\!\big(T(y,x)\big),\qquad
T(y,x)=\sum_{j=0}^{J}\beta_j(x)\,B_{j,J}\!\big(y_{01}\big),
\end{equation}
where $\Phi(\cdot)$ is the CDF of $\mathcal{N}(0,1)$, $\{\beta_j(\cdot)\}_{j=0}^J$ are learned functions of $x$, $\{B_{j,J}(\cdot)\}_{j=0}^J$ are Bernstein basis functions, and $y_{01}$ is a normalized version of $y$ that maps $\mathbb{R}$ to $[0,1]$:
\begin{equation}\label{eq:y01_def}
y_{01}=\sigma\!\left(\frac{y-m}{s}\right)=\frac{1}{1+\exp\!\left(-\frac{y-m}{s}\right)},
\end{equation}
with fixed constants $m$ and $s>0$, and $\sigma(\cdot)$ the sigmoid function.

\paragraph{Monotonicity via a two-head network.}
To ensure that, for each fixed $x$, the map $y\mapsto T(y,x)$ is nondecreasing, we design two output heads:
one head predicts a baseline value $c(x)=\beta_0(x)$, and the other head predicts increments
$\{\Delta_k(x)\}_{k=1}^J$. We then set
\begin{equation}\label{eq:beta_softplus}
\beta_j(x)=c(x)+\sum_{k=1}^{j}\mathrm{softplus} \big(\Delta_k(x)\big),
\qquad
\mathrm{softplus} (a)=\log(1+e^a),
\end{equation}
which implies the ordered constraint
\begin{equation}\label{eq:beta_increase}
\beta_0(x)\le \beta_1(x)\le \cdots \le \beta_J(x)\qquad \text{for all }x.
\end{equation}

\paragraph{Why \eqref{eq:beta_increase} implies monotonicity of $T(\cdot,x)$.}
Recall the Bernstein basis on $[0,1]$:
\[
B_{k,J}(u)=\binom{J}{k}u^k(1-u)^{J-k},\qquad u\in[0,1],\quad k=0,\ldots,J.
\]
A standard identity is
\begin{equation}\label{eq:bern_derivative_identity}
\frac{d}{du}B_{k,J}(u)=J\Big(B_{k-1,J-1}(u)-B_{k,J-1}(u)\Big),
\end{equation}
with the conventions $B_{-1,J-1}\equiv 0$ and $B_{J,J-1}\equiv 0$.
Using \eqref{eq:bern_derivative_identity}, for $u\in[0,1]$,
\[
\begin{aligned}
\frac{\partial}{\partial u}T(u;x)
&=\sum_{k=0}^J \beta_k(x)\frac{d}{du}B_{k,J}(u) \\
&=J\sum_{k=0}^J \beta_k(x)\Big(B_{k-1,J-1}(u)-B_{k,J-1}(u)\Big) \\
&=J\sum_{k=0}^{J-1}\Big(\beta_{k+1}(x)-\beta_k(x)\Big)B_{k,J-1}(u),
\end{aligned}
\]
where the last step follows by index shifting and the endpoint conventions.
Let $\widetilde{\Delta}_k(x):=\beta_{k+1}(x)-\beta_k(x)\ge 0$ by \eqref{eq:beta_increase}.
Since $B_{k,J-1}(u)\ge 0$ and $\sum_{k=0}^{J-1}B_{k,J-1}(u)=1$ for all $u\in[0,1]$, we obtain
\[
\frac{\partial}{\partial u}T(u;x)=J\sum_{k=0}^{J-1}\widetilde{\Delta}_k(x)\,B_{k,J-1}(u)\ge 0.
\]
Hence $u\mapsto T(u;x)$ is nondecreasing. Because $y\mapsto y_{01}$ in \eqref{eq:y01_def} is strictly increasing,
it follows that $y\mapsto T(y,x)=T(y_{01}(y);x)$ is also nondecreasing. Therefore, the architecture enforces
monotonicity structurally, ensuring that $\widehat F_{Y\mid X}(y\mid x)=\Phi(T(y,x))$ is a valid CDF.

\paragraph{Training objective.}
DCTM is trained by minimizing the negative log-likelihood (NLL). For the standard normal baseline,
the (population) objective can be written as
\begin{equation}\label{eq:dctm_nll_continuous}
\mathcal{L}
=
-\mathbb{E}\!\left[
\log \phi\!\big(T(Y,X)\big)
+
\log\!\left(\frac{\partial T(Y,X)}{\partial y}\right)
\right],
\end{equation}
where $\phi(\cdot)$ is the $\mathcal{N}(0,1)$ density. Minimizing \eqref{eq:dctm_nll_continuous} encourages the
model-implied CDF $\Phi(T(y,x))$ to approximate the true conditional distribution of $Y\mid X$.

Similarly, for the parent outcome $W$, we independently train a DCTM of the same form to estimate $F_{W\mid X}$
(and the conditional rank $V:=F_{W\mid X}(W\mid X)$ in the continuous case). Compared with DR (multiple pointwise
fits), DCTM fits the entire conditional distribution end-to-end and guarantees distributional validity by
construction, without any post-hoc monotonicity corrections.

\subsubsection{Discrete outcomes}\label{sec:discrete}
The DCTM framework also applies to ordinal discrete responses (denoted by dDCTM). We introduce the notation:
let $X\in\mathbb{R}^p$ be covariates, and let the child and parent ordinal outcomes be
\[
Y\in\{0,1,\ldots,K_Y-1\},\qquad
W\in\{0,1,\ldots,K_W-1\},
\]
representing ordered categories (e.g., occupational status, education level). Define the conditional CDFs
\begin{align}
F_k^{Y}(x) &:= \mathbb{P}(Y\le k\mid X=x),\quad k=0,\ldots,K_Y-1, \label{eq:cdf_defY}\\
F_\ell^{W}(x) &:= \mathbb{P}(W\le \ell\mid X=x),\quad \ell=0,\ldots,K_W-1, \label{eq:cdf_defW}
\end{align}
where $F_{K_Y-1}^{Y}(x)=F_{K_W-1}^{W}(x)=1$.

Unlike the continuous DCTM, the discrete dDCTM takes only $X$ as input and outputs the entire vector of
conditional CDF values across categories. We choose the logistic distribution as the baseline and use a
simple cumulative monotone construction. For $Y$, define
\begin{equation}\label{eq:ddctm_monotone_scores}
s_k^{Y}(x)=c^{Y}(x)+\sum_{j=0}^{k}\widetilde{\Delta}_j^{Y}(x),
\qquad
\widetilde{\Delta}_j^{Y}(x)=\mathrm{softplus} \big(\Delta_j^{Y}(x)\big)\ge 0,
\end{equation}
and
\begin{equation}\label{eq:ddctm_cdf}
F_k^{Y}(x)=\sigma\!\big(s_k^{Y}(x)\big),\qquad k=0,\ldots,K_Y-2,
\end{equation}
where $\sigma(t)=1/(1+e^{-t})$ is the logistic CDF, and $c^Y(\cdot)$ and $\{\Delta_j^Y(\cdot)\}$ are network outputs.

\paragraph{Network structure.}
Following the general DCTM architecture in Figure~\ref{fig:dctm_fig}, we use a backbone network
$h_\theta:\mathbb{R}^p\to\mathbb{R}^d$ to extract features from $X$, and then feed them into two heads that output
$c^Y$ and $\{\Delta_j^Y\}$. The cumulative construction and $\mathrm{softplus} (\cdot)$ ensure
\[
s_0^{Y}(x)\le s_1^{Y}(x)\le \cdots \le s_{K_Y-2}^{Y}(x),
\]
and since $\sigma(\cdot)$ is increasing, we obtain
\[
F_0^{Y}(x)\le F_1^{Y}(x)\le \cdots \le F_{K_Y-2}^{Y}(x),
\]
i.e., monotonicity across categories holds by structural design, without post-hoc rearrangement.

\paragraph{From CDF to category probabilities.}
Given $(F_0^{Y}(x),\ldots,F_{K_Y-2}^{Y}(x))$, the conditional category probabilities are obtained by differences:
\[
\begin{aligned}
p_0^{Y}(x)&=F_0^{Y}(x),\\
p_y^{Y}(x)&=F_y^{Y}(x)-F_{y-1}^{Y}(x),\qquad 1\le y\le K_Y-2,\\
p_{K_Y-1}^{Y}(x)&=1-F_{K_Y-2}^{Y}(x).
\end{aligned}
\]
We train dDCTM by minimizing the negative log-likelihood
\begin{equation}\label{eq:ddctm_nll}
\mathcal{L}_Y(\theta)
=
-\mathbb{E}\!\left[\log p_Y^{Y}(X)\right]
=
-\frac{1}{n}\sum_{i=1}^{n}\log p_{y_i}^{Y}(x_i).
\end{equation}
The construction for $W$ is analogous. One may train the two networks separately or use a multi-task architecture
(e.g., MMoE \citep{Ma2018MMoE}) to learn them jointly. To address class imbalance, one can also introduce class-weighted NLL (e.g.,
inverse-frequency weights for tail categories).

\paragraph{Comparison with DR in the discrete case.}
Under DR, one typically fits $K_Y-1$ separate binary regressions for $\mathbb{I}\{Y\le k\}\mid X$ to obtain
$\widehat F_k^{Y}(x)$ and then enforces monotonicity across $k$. This approach (i) does not share information
across thresholds, (ii) is prone to misspecification and crossing CDF curves when $k$ and $x$ interact strongly,
and (iii) can be particularly fragile in the tails. In contrast, dDCTM shares the backbone representation and
uses a structured monotone increment design, enabling joint fitting across all thresholds with fewer degrees of
freedom and guaranteed global coherence.

Compared with the classical DR approach, using DCTM (and dDCTM) to estimate conditional CDFs has the following
advantages:
\begin{enumerate}
\item \textbf{Broad applicability.} DCTM can handle a wide range of response types, including continuous,
discrete/ordinal, and survival outcomes. By selecting an appropriate baseline distribution and designing the
corresponding transformation, one can cover multiple output types within a unified framework. While the concrete
form of the transformation varies across tasks, the training objective, optimization pipeline, and downstream
inference steps can remain largely consistent, enabling a ``one workflow, multiple outputs'' implementation.

\item \textbf{Distribution-free with interpretability.} DCTM is often described as \emph{distribution-free} in
the sense that it does not impose a parametric family for $Y\mid X$; instead, it uses flexible transformations
to map $Y\mid X$ toward a baseline distribution.

\item \textbf{Monotonicity and global coherence.} By enforcing monotonicity structurally through the network,
DCTM produces CDF estimates that satisfy basic probabilistic properties. Unlike DR's pointwise fitting, DCTM
fits the entire conditional distribution end-to-end, thereby improving global coherence.

\item \textbf{Scalability and extensibility with deep learning.} With neural networks as function approximators,
DCTM can automatically learn rich feature representations (e.g., higher-order nonlinear interactions), improving
pattern discovery and scalability to high dimensions \citep{sick2021deep,kook2022deep}. Related methods have been
implemented in the \textsf{\textbf{deeptrafo}} package \citep{kook2024estimating}, which allows flexible choices of
network architecture, loss functions, optimization strategies, hyperparameter tuning, and diagnostics.
\end{enumerate}

In summary, DCTM provides a powerful and flexible tool for conditional distribution estimation. In the CRRR
setting, we use DCTM to estimate $F_{Y\mid X}$ and $F_{W\mid X}$, and then compute conditional ranks. This yields
reliable conditional-rank estimates (bounded and monotone), improves accuracy, and adapts to complex data
distributions.

\subsection{Cross-Fitting}\label{sec:crossfitting}
To improve robustness and avoid in-sample overfitting bias (training and prediction on the same data),
we adopt a cross-fitting strategy. Specifically, we split the sample into $K$ folds. For each fold, we train DCTM
on the other $K-1$ folds and then compute out-of-fold (OOF) conditional ranks on the held-out fold. We then pool
the OOF ranks across all folds and compute the sample estimator of $\rho_C$. The procedure is summarized in
Algorithm~1.

\medskip 
\noindent\textbf{Algorithm 1 (Cross-Fitting Estimator of $\rho_C$).} \begin{enumerate} 
\item Randomly split the index set $\{1,\ldots,n\}$ into $K$ disjoint subsets $\mathcal{I}_1,\ldots,\mathcal{I}_K$. Let $\mathcal{I}_k^c=\{1,\ldots,n\}\setminus \mathcal{I}_k$ denote the complement of the $k$th fold. 
\item For $k=1,\ldots,K$, do: 
\begin{enumerate} 
\item \textbf{Fit fold-specific conditional CDFs.} Using the training sample $\{(Y_i,X_i)\}_{i\in\mathcal{I}_k^c}$, fit a DCTM and obtain $\widehat F_{Y\mid X}^{(-k)}(\cdot\mid\cdot)$. Similarly, using $\{(W_i,X_i)\}_{i\in\mathcal{I}_k^c}$, fit a DCTM and obtain $\widehat F_{W\mid X}^{(-k)}(\cdot\mid\cdot)$. 
\item \textbf{Compute OOF conditional ranks.} For each $i\in\mathcal{I}_k$, compute \[ \widehat U_i=\widehat F_{Y\mid X}^{(-k)}(Y_i\mid X_i), \qquad \widehat V_i=\widehat F_{W\mid X}^{(-k)}(W_i\mid X_i). \] 
If $Y$ and/or $W$ are discrete, conditional ranks require the tie parameter $\omega$; see Section~\ref{sc:discrete}. 
\end{enumerate} 
\item \textbf{Compute the OLS-slope estimator of $\rho_C$.} \[ \widehat\rho_C = \frac{\sum_{i=1}^{n}(\widehat U_i-\overline U)(\widehat V_i-\overline V)} {\sum_{i=1}^{n}(\widehat V_i-\overline V)^2}, \qquad \overline U=\frac{1}{n}\sum_{i=1}^n \widehat U_i,\quad \overline V=\frac{1}{n}\sum_{i=1}^n \widehat V_i. \] If $Y$ and $W$ are continuous, two equivalent forms (defined earlier as \eqref{eq:corr} and \eqref{eq:corr_1_2}) are \[ \widetilde \rho_C := \frac{\sum_{i=1}^{n}(\widehat U_i-\overline{\widehat U})(\widehat V_i-\overline{\widehat V})} {\sqrt{\sum_{i=1}^{n}(\widehat U_i-\overline{\widehat U})^2\ \sum_{i=1}^{n}(\widehat V_i-\overline{\widehat V})^2}}, \qquad \breve\rho_C := \frac{12}{n}\sum_{i=1}^{n}\left(\widehat U_i-\frac{1}{2}\right)\left(\widehat V_i-\frac{1}{2}\right). \] 
\end{enumerate}
\subsection{CRRR with Discrete Outcomes}\label{sc:discrete}
In many empirical applications, the outcomes are not continuous but discrete/ordinal, such as occupational status
\citep{ward2022internal}, education level \citep{asher2024intergenerational}, and human capital
\citep{croix2024nepotism}. Discreteness induces many ties. Importantly, without changing the overall ordering,
tied observations can be assigned any rank value within an interval. Specifically, if
$\mathbb{P}(Y=y_0\mid X=x)=p_0>0$, then
$F_{Y\mid X}(y_0\mid x)-F_{Y\mid X}^{-}(y_0\mid x)=p_0$, and observations with $Y=y_0$ may be assigned ranks
anywhere in
\[
\big[F_{Y\mid X}^{-}(y_0\mid x),\ F_{Y\mid X}(y_0\mid x)\big]
\]
without affecting the induced ordering. Since $\rho_C$ depends on the ranks, we must handle ties carefully to
ensure identification and interpretability. Note also that, in the discrete case, the three estimators
\eqref{eq:ols}--\eqref{eq:corr_1_2} are no longer equivalent; we take the OLS-slope estimator \eqref{eq:ols} as
the default.

In the context of rank-rank regression, \citet{chetverikov2023inference} emphasize three key differences in the
discrete case: (i) the meaning of the target parameter depends on whether ties exist; (ii) different tie-breaking
rank definitions imply different target parameters; (iii) the sampling properties of the estimator also depend
on the rank definition. Hence, the discrete setting requires a separate treatment of parameter definition and
inference \citep{chetverikov2023inference}. Similarly, CRRR faces analogous challenges when $Y$ and/or $W$ are
discrete. Existing CRRR work focuses on the continuous case and does not address discreteness
\citep{chernozhukov2024conditional}. We therefore propose a parametric family of conditional-rank definitions for
the discrete case and investigate how $\widehat\rho_C$ changes with the rank definition.

\subsubsection{Conditional-rank definition with a tie parameter}
Let $F_{Y\mid X}(y\mid X)=\mathbb{P}(Y\le y\mid X)$ and $F_{Y\mid X}^{-}(y\mid X)=\mathbb{P}(Y<y\mid X)$; define
$F_{W\mid X}(w\mid X)$ and $F_{W\mid X}^{-}(w\mid X)$ analogously. For any $\omega\in[0,1]$ and fixed $X=x$, define
the conditional ranks by
\begin{align}
R_{Y\mid X=x}(y)
&:=
\omega F_{Y\mid X}(y\mid x) + (1-\omega)F_{Y\mid X}^{-}(y\mid x), \label{eq:rankY_omega_en}\\
R_{W\mid X=x}(w)
&:=
\omega F_{W\mid X}(w\mid x) + (1-\omega)F_{W\mid X}^{-}(w\mid x). \label{eq:rankW_omega_en}
\end{align}
For comparability, we recommend using the same $\omega$ for both $Y$ and $W$. Fixing $\omega$ yields well-defined
conditional ranks $U_\omega=R_{Y\mid X}(Y)$ and $V_\omega=R_{W\mid X}(W)$. Intuitively, $\omega$ controls where the
rank of a tied group falls within the interval between the smallest possible rank and the largest possible rank:
$\omega=1$ assigns the largest rank $F_{Y\mid X}(y\mid X)$; $\omega=0$ assigns the smallest rank
$F_{Y\mid X}^{-}(y\mid X)$; and $\omega=0.5$ assigns the midpoint, i.e., the \emph{mid-rank} (see, e.g.,
\citealp{hoeffding1992class}). Table~\ref{tab:ranks-three-defs_en} illustrates the three choices.

\setlength{\heavyrulewidth}{1.2pt}
\setlength{\lightrulewidth}{0.6pt}
\setlength{\belowcaptionskip}{4pt}
\renewcommand{\arraystretch}{1.15}
\setlength{\tabcolsep}{4pt}
\begin{table}[htbp]
\centering
\caption{An example of three tie-breaking rank definitions (adapted from \cite{chetverikov2023inference}).}
\label{tab:ranks-three-defs_en}
\begin{tabular}{l*{10}{c}}
\toprule
$i$   & 1 & 2 & 3 & 4 & 5 & 6 & 7 & 8 & 9 & 10 \\
$Y_i$ & 3 & 4 & \textbf{7} & \textbf{7} & 10 & 11 & \textbf{15} & \textbf{15} & \textbf{15} & \textbf{15} \\
\midrule
Smallest rank ($\omega=0$)   & 0.1 & 0.2 & \textbf{0.3}  & \textbf{0.3}  & 0.5 & 0.6 & \textbf{0.7}  & \textbf{0.7}  & \textbf{0.7}  & \textbf{0.7}  \\
Mid-rank ($\omega=0.5$)      & 0.1 & 0.2 & \textbf{0.35} & \textbf{0.35} & 0.5 & 0.6 & \textbf{0.85} & \textbf{0.85} & \textbf{0.85} & \textbf{0.85} \\
Largest rank ($\omega=1$)    & 0.1 & 0.2 & \textbf{0.4}  & \textbf{0.4}  & 0.5 & 0.6 & \textbf{1}    & \textbf{1}    & \textbf{1}    & \textbf{1}    \\
\bottomrule
\end{tabular}
\end{table}

In the continuous case, since $U,V\sim \mathrm{Unif}(0,1)$, the CRRR slope $\rho_C$ equals the conditional
Spearman's rank correlation $\rho_S$. In the discrete case, however, $U_\omega$ and $V_\omega$ are generally not
uniform, and $\rho_C$ is no longer equal to $\rho_S$. Instead,
\[
\rho_C
=
\frac{\operatorname{Cov}(U,V)}{\operatorname{Var}(V)}
=
\frac{\operatorname{Cov}(U,V)}{\sqrt{\operatorname{Var}(U)\operatorname{Var}(V)}}
\times
\frac{\sqrt{\operatorname{Var}(U)}}{\sqrt{\operatorname{Var}(V)}}
=
\rho_S \times \frac{\sqrt{\operatorname{Var}(U)}}{\sqrt{\operatorname{Var}(V)}}.
\]
In the continuous case, $\operatorname{Var}(U)=\operatorname{Var}(V)=1/12$, hence $\rho_C=\rho_S$. In the discrete
case, the standard-deviation ratio may deviate from $1$, so $\rho_C$ can even fall outside $[-1,1]$ and should
not be interpreted directly as a correlation coefficient. In empirical work, it is therefore important to
distinguish whether one targets the Spearman-type rank correlation $\rho_S$ or the regression slope $\rho_C$; if
a correlation interpretation is desired, one may rescale $\rho_C$ by the standard-deviation ratio.

\subsubsection{Impact of $\omega$ on $\rho_C$}
With the parametric rank definition above, the estimator becomes $\widehat\rho_C(\omega)$ and varies with the
tie parameter $\omega$. In the discrete ordinal simulations in Section~\ref{ch:num_exps}, we consider 101 grid
points in $\omega\in[0,1]$ and use large-scale Monte Carlo simulation to approximate the population quantities
$\rho_C(\omega)$, $\rho_S(\omega)$, and $\mathrm{sd}(U_\omega)/\mathrm{sd}(V_\omega)$. The resulting curves are
reported in Figure~\ref{fig:rho_omega}. The figure shows that the choice of $\omega$ (i.e., the rank definition)
is crucial for intergenerational mobility measurement with discrete outcomes: different $\omega$ values may lead
to materially different conclusions, sometimes even reversing the qualitative direction. Therefore, when
outcomes are discrete, one must pre-specify and report the rank definition; otherwise, mobility conclusions are invalid.

\begin{figure}[htbp]
    \centering
    \includegraphics[width=\textwidth]{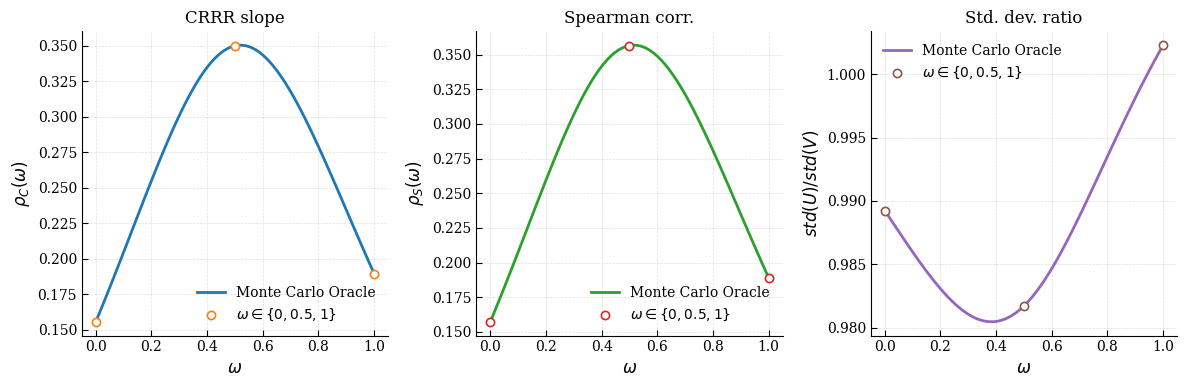}
    \caption{$\rho_C(\omega)$, $\rho_S(\omega)$, and $\mathrm{sd}(U_\omega)/\mathrm{sd}(V_\omega)$ as functions of $\omega$.}
    \label{fig:rho_omega}
\end{figure}

\subsection{Bootstrap Inference}\label{sec:bootstrap}
In Section~\ref{ch:theory}, we establish the consistency and asymptotic normality of
$\widehat\rho_C$, $\widetilde\rho_C$, and $\breve\rho_C$ in the continuous case. Although one can construct
plug-in variance estimators from their asymptotic linear representations, the implementation is often
nontrivial. We therefore follow \citet{chernozhukov2024conditional} and use an \emph{exchangeable bootstrap} to
conduct inference, obtaining standard errors and confidence intervals for
$\widehat\rho_C$, $\widetilde\rho_C$, and $\breve\rho_C$. The validity of the bootstrap procedure is proved in
Section~\ref{ch:theory}.

The exchangeable bootstrap is a general resampling-by-reweighting scheme: by appropriately choosing the
distribution of the weight vector $(\omega_{n1},\ldots,\omega_{nn})$, it can cover many common bootstrap methods,
including the empirical bootstrap, weighted bootstrap, and wild bootstrap; see \citet{van1996weak}. For example,
the empirical bootstrap corresponds to $(\omega_{n1},\ldots,\omega_{nn})$ being multinomial with parameters
$n$ and cell probabilities $(1/n,\ldots,1/n)$.

We now describe how to use the exchangeable bootstrap to compute standard errors and asymptotic confidence
intervals for $\rho_C$. Let $B$ be the number of bootstrap repetitions and $\alpha$ the significance level for
a $(1-\alpha)$ confidence interval (e.g., $B=500$ and $\alpha=0.05$).

\medskip 
\noindent\textbf{Algorithm 2 (Exchangeable Bootstrap Inference for $\rho_C$).} 
\begin{enumerate} 
\item \textbf{Main estimates.} Apply Algorithm~1 with $K$ folds to obtain the main estimates $\widehat\rho_C$, $\widetilde\rho_C$, and $\breve\rho_C$. 
\item \textbf{Generate and normalize weights.} For each bootstrap repetition $b=1,\ldots,B$, generate a weight vector $(\omega_1^{(b)},\ldots,\omega_n^{(b)})$ satisfying Assumption~\ref{ass:bootstrap}, and normalize: \[ \widetilde{\omega}_i^{(b)} := \frac{\omega_i^{(b)}}{\frac{1}{n}\sum_{j=1}^n \omega_j^{(b)}}, \qquad i=1,\ldots,n. \] 
\item \textbf{Weighted DCTM fitting (fold by fold).} For each $R\in\{Y,W\}$ and each fold $k=1,\ldots,K$, fit the DCTM on the training set $\mathcal{I}_k^c$ via weighted maximum likelihood. Let $n_k^c:=|\mathcal{I}_k^c|$. 
Define the weighted NLL: 
\[ 
\mathrm{NLL}^{\mathrm{w}}_{R} = -\frac{1}{n_k^c}\sum_{i\in\mathcal{I}_k^c}\widetilde{\omega}_i^{(b)}\, \ell_R\!\left(\theta;R_i,X_i\right). 
\] 
This yields fold-specific bootstrap CDF estimators $\widehat F_{Y\mid X}^{*(b,-k)}$ and $\widehat F_{W\mid X}^{*(b,-k)}$. For each $i\in\mathcal{I}_k$, compute the fold-$k$ out-of-fold conditional ranks 
\[ 
\widehat U_i^{*(b)}=\widehat F_{Y\mid X}^{*(b,-k)}(Y_i\mid X_i), \qquad \widehat V_i^{*(b)}=\widehat F_{W\mid X}^{*(b,-k)}(W_i\mid X_i). 
\] 
If $Y$ and/or $W$ are discrete, define ranks via \eqref{eq:rankY_omega_en}--\eqref{eq:rankW_omega_en} with the chosen tie parameter $\omega$. \item \textbf{Bootstrap estimates via weighted moments.} Using the normalized weights $\{\widetilde{\omega}_i^{(b)}\}_{i=1}^n$, compute $\widehat\rho_C^{*(b)}$, $\widetilde\rho_C^{*(b)}$, and $\breve\rho_C^{*(b)}$ by replacing all sample moments (means, variances, covariances) in the original estimators with their weighted analogues. For example, \[ \breve\rho_C^{*(b)} = \frac{12}{n}\sum_{i=1}^{n}\widetilde{\omega}_i^{(b)} \left(\widehat U_i^{*(b)}-\frac{1}{2}\right) \left(\widehat V_i^{*(b)}-\frac{1}{2}\right). \] \item \textbf{Repeat.} Repeat Steps~2--4 for $b=1,\ldots,B$ to obtain bootstrap draws $\{\widehat\rho_C^{*(b)}\}_{b=1}^B$, $\{\widetilde\rho_C^{*(b)}\}_{b=1}^B$, and $\{\breve\rho_C^{*(b)}\}_{b=1}^B$. 
\item \textbf{Standard error and confidence interval (example: $\widehat\rho_C$).} Define \[ Z_b^{*}=\sqrt{n}\left(\widehat\rho_C^{*(b)}-\widehat\rho_C\right), \qquad b=1,\ldots,B. \] Let $q_{0.75}$ and $q_{0.25}$ be the empirical $75\%$ and $25\%$ quantiles of $\{Z_b^{*}\}_{b=1}^B$. 
Set \[ \widehat{\sigma}_{\rho} = \frac{q_{0.75}-q_{0.25}}{z_{0.75}-z_{0.25}}, \qquad z_p=\Phi^{-1}(p), \] and estimate the standard error by 
\[ 
\widehat{\mathrm{SE}}(\widehat\rho_C)=\frac{\widehat{\sigma}_{\rho}}{\sqrt{n}}. 
\] 
Next define the studentized statistics \[ T_b=\frac{|Z_b^{*}|}{\widehat{\sigma}_{\rho}}, \qquad b=1,\ldots,B, \] and let $\widehat t_{1-\alpha}$ be the empirical $(1-\alpha)$ quantile of $\{T_b\}_{b=1}^B$. 
Then an asymptotic $(1-\alpha)$ confidence interval is \[ \mathrm{ACI}_{1-\alpha} = \left[ \widehat\rho_C-\widehat t_{1-\alpha}\frac{\widehat{\sigma}_{\rho}}{\sqrt{n}}, \; \widehat\rho_C+\widehat t_{1-\alpha}\frac{\widehat{\sigma}_{\rho}}{\sqrt{n}} \right]. \] The same construction applies to $\widetilde\rho_C$ and $\breve\rho_C$. 
\end{enumerate}

\section{Asymptotic Theory}\label{ch:theory}

This section develops the asymptotic theory for the CRRR slope estimator $\rho_C$ in the \emph{continuous}
outcome case where both $Y$ and $W$ are continuous. In the discrete case, the conditional ranks are step
functions and hence non-differentiable; the resulting inference theory becomes substantially more involved.
We therefore do not pursue the discrete case here and leave it for future research.

As discussed earlier, when $Y$ and $W$ are continuous, $\rho_C$ admits three equivalent representations
(cf.\ \eqref{eq:spearman_rho_C} and \eqref{eq:half}), namely the OLS-slope form $\rho_C^{\mathrm{ols}}$,
the rank-correlation form $\rho_C^{\mathrm{corr}}$, and the covariance form $\rho_C^{\mathrm{cov}}$:
\begin{align}
\rho_C^{\mathrm{ols}}
&=\frac{\mathrm{Cov}(U,V)}{\mathrm{Var}(V)}, \label{eq:ols_1}\\
\rho_C^{\mathrm{corr}}
&=\frac{\mathrm{Cov}(U,V)}{\sqrt{\mathrm{Var}(U)\mathrm{Var}(V)}}, \label{eq:corr_2}\\
    \rho_C^{cov}&=12 \mathbb{E}[(U-\tfrac{1}{2})(V-\tfrac{1}{2})]. \label{eq:cov_3}
\end{align}
Their sample counterparts are $\widehat\rho_C$, $\widetilde\rho_C$, and $\breve\rho_C$ defined in
\eqref{eq:ols}--\eqref{eq:corr_1_2}.

\medskip

For each $R\in\{Y,W\}$, as described in subsection~\ref{sec:dctm}, we first apply a
normalization:
\[
g(r)=\sigma\!\left(\frac{r-m}{s}\right),
\qquad
\sigma(a)=\frac{1}{1+e^{-a}},
\]
where $m$ and $s>0$ are fixed scale constants chosen so that most observations fall into the ``nearly linear''
region of the sigmoid curve, preventing excessive mass from being compressed near $0$ or $1$. Since $g(\cdot)$ is
strictly increasing, this normalization does not change the overall rank ordering nor the resulting $\rho_C$.
For notational convenience, we will still write $R\in\{Y,W\}$ for the normalized variables below.

\medskip

For continuous outcomes, DCTM uses the standard normal distribution as the baseline. The transformation function
is expanded using Bernstein basis functions $\{B_{j,J}\}_{j=0}^{J}$ and is constrained to be monotone
nondecreasing through a structured nondecreasing coefficient sequence $\{\beta_j\}_{j=0}^{J}$:
\[
F_{R\mid X}(r\mid x;\theta_R)=\Phi\!\Big(T(r,x;\theta_R)\Big),
\qquad
T(r,x;\theta_R)=\sum_{j=0}^{J}\beta_j(x;\theta_R)\,B_{j,J}(r),
\]
where $\theta_R$ is the network parameter vector, $\Phi(\cdot)$ is the CDF of $\mathcal{N}(0,1)$, and
$T(\cdot,x;\theta_R)$ is monotone nondecreasing in $r$. Training is carried out by maximum likelihood, i.e.,
minimizing the negative log-likelihood (NLL).

\medskip

In the cross-fitting stage, we randomly partition the sample indices $\{1,\ldots,n\}$ into $K$ disjoint folds
$\{\mathcal I_k\}_{k=1}^K$. For each fold $k$, let $n_k:=|\mathcal I_k|$ and denote the training set by
$\mathcal I_k^c$ with size $n_{ck}:=|\mathcal I_k^c|$. We fit DCTM on $\mathcal I_k^c$ to obtain fold-specific
conditional CDF estimators $\widehat F_{Y\mid X}^{(-k)}$ and $\widehat F_{W\mid X}^{(-k)}$, and then compute
out-of-fold (OOF) conditional ranks for $i\in\mathcal I_k$:
\[
\widehat U_i:=\widehat F_{Y\mid X}^{(-k)}(y_i\mid x_i),
\qquad
\widehat V_i:=\widehat F_{W\mid X}^{(-k)}(w_i\mid x_i).
\]
Finally, substituting $\{(\widehat U_i,\widehat V_i)\}_{i=1}^n$ into \eqref{eq:ols}--\eqref{eq:corr_1_2}
yields $\widehat\rho_C$, $\widetilde\rho_C$, and $\breve\rho_C$.

\medskip

Throughout this section we consider a \emph{fixed-complexity} asymptotic regime: the DCTM architecture (depth,
width, activation functions), the Bernstein order $J$, training hyperparameters, and the number of folds $K$ do
not vary with $n$. Moreover, the $K$ folds are approximately balanced so that $n_k/n\to 1/K$ and
$n_{ck}/n\to 1-1/K$. For each $R\in\{Y,W\}$, let $\Theta_R$ be the parameter space and define the induced class of
conditional CDFs
\[
\mathcal F_R
:=
\Big\{
F(\cdot\mid\cdot)=F_{R\mid X}(\cdot\mid\cdot;\theta_R):\ \theta_R\in\Theta_R
\Big\}.
\]
Because neural networks often have permutation symmetries and scale invariances, $\theta_R$ is typically not
identifiable in a strict statistical sense: different $\theta_R$ may correspond to the same conditional CDF
$F_{R\mid X}$. Therefore, we formulate the asymptotic theory directly at the functional level in terms of
$F_{R\mid X}$, avoiding identifiability issues for $\theta_R$. Proofs of the lemmas and theorems are deferred to
the appendix. We begin with assumptions.

\medskip

\begin{assumption}[Continuous case]\label{ass:data_fun}
The observed sample $\{Z_i\}_{i=1}^n=\{(Y_i,W_i,X_i)\}_{i=1}^n$ is i.i.d., where $Y,W\in[0,1]$ are the
(normalized) continuous outcomes and $X\in\mathcal X\subset\mathbb R^{p}$ is the covariate vector with $\mathcal
X$ compact. For almost every $x\in\mathcal X$, the conditional CDFs $F_{Y\mid X}(\cdot\mid x)$ and
$F_{W\mid X}(\cdot\mid x)$ are continuous and strictly increasing, and admit conditional densities
$f_{Y\mid X}$ and $f_{W\mid X}$ that are continuous on their supports and are locally bounded and strictly
positive almost everywhere.
\end{assumption}

\begin{assumption}[Correct specification of DCTM]\label{ass:spec_fun}
For each $R\in\{Y,W\}$, the true conditional CDF satisfies $F_{R\mid X}\in\mathcal F_R$. Let
$\ell_R(F;R,X)$ denote the single-observation negative log-likelihood loss for $F\in\mathcal F_R$, and define the
population risk
\[
\mathcal Q_R(F):=\mathbb{E} \big[\ell_R(F;R,X)\big],
\qquad F\in\mathcal F_R.
\]
Assume that the minimizer of $\mathcal Q_R(\cdot)$ over $\mathcal F_R$ is unique in the function sense and is
attained by the truth $F_{R\mid X}$: if $\mathcal Q_R(\widetilde F)=\inf_{G\in\mathcal F_R}\mathcal Q_R(G)$, then
$\widetilde F=F_{R\mid X}$ almost everywhere.
\end{assumption}

\begin{assumption}[Empirical-process regularity and functional asymptotic linearity]\label{ass:emp_process}
For each $R\in\{Y,W\}$ and each fold $k$, define the training-sample empirical risk
\[
\mathcal Q^{(-k)}_{R,n}(F)
:=
\frac1{n_{ck}}\sum_{i\in\mathcal I_k^c}\ell_R(F;R_i,X_i),
\qquad F\in\mathcal F_R,
\]
and let the fold-specific estimator be
\[
\widehat F_{R\mid X}^{(-k)}\in\arg\min_{F\in\mathcal F_R}\mathcal Q^{(-k)}_{R,n}(F).
\]
Assume:
\smallskip

\noindent (i) For each fixed $k$,
\[
\sup_{F\in\mathcal F_R}\big|\mathcal Q^{(-k)}_{R,n}(F)-\mathcal Q_R(F)\big|\xrightarrow{p}0.
\]

\smallskip
\noindent (ii) There exists a measurable map
$\varphi_R:[0,1]\times\mathcal X\times\mathcal Z\to\mathbb R$ such that, for each fixed $k$,
\[
\sup_{(r,x)\in[0,1]\times\mathcal X}
\left|
\sqrt{n_{ck}}\Big(\widehat F_{R\mid X}^{(-k)}(r\mid x)-F_{R\mid X}(r\mid x)\Big)
-
\frac1{\sqrt{n_{ck}}}\sum_{j\in \mathcal I_k^c}\varphi_R(r,x;Z_j)
\right|\xrightarrow{p}0,
\]
and $E[\varphi_R(r,x;Z)]=0$ for all $(r,x)$. Moreover, there exists some $\delta>0$ such that
\[
\mathbb{E} \Big[\sup_{(r,x)}|\varphi_R(r,x;Z)|^{2+\delta}\Big]<\infty.
\]
\end{assumption}

\begin{remark}
Assumption~\ref{ass:emp_process} imposes high-level regularity conditions for M-estimation / empirical risk
minimization (ERM) over the function class $\mathcal F_R$. Assumption~\ref{ass:emp_process}(ii) provides a
$\sqrt{n}$-asymptotic linear expansion for $\widehat F_{R\mid X}^{(-k)}$. Under fixed complexity and standard
smoothness/identifiability conditions, the parameter estimator $\widehat\theta_R$ admits a $\sqrt n$ asymptotic
linear representation, and the Delta method transfers this linearization to $F_{R\mid X}(r\mid x)$; see, e.g.,
\citep{newey1994large,van1996weak,van2000asymptotic}.
\end{remark}

\begin{lemma}[Consistency of fold-specific conditional CDF estimators]\label{lem:dctm_consistency_fun}
Under Assumptions~\ref{ass:data_fun}--\ref{ass:emp_process}, for each $R\in\{Y,W\}$ and each fixed fold $k$,
\[
\sup_{(r,x)\in[0,1]\times\mathcal X}
\Big|\widehat F_{R\mid X}^{(-k)}(r\mid x)-F_{R\mid X}(r\mid x)\Big|
\xrightarrow{p}0.
\]
\end{lemma}

\begin{theorem}[Consistency of $\rho_C$ estimators]\label{thm:consistency_rho_fun}
Under Assumptions~\ref{ass:data_fun}--\ref{ass:emp_process}, the three (equivalent) estimators
$\widehat\rho_C$, $\widetilde\rho_C$, and $\breve\rho_C$ satisfy
\[
\widehat\rho_C\xrightarrow{p}\rho_C,
\qquad
\widetilde\rho_C\xrightarrow{p}\rho_C,
\qquad
\breve\rho_C\xrightarrow{p}\rho_C.
\]
\end{theorem}

\medskip

Define the following sample moments:
\[
A_n:=\mathbb P_n\big[(\widehat U-\tfrac12)(\widehat V-\tfrac12)\big],
\qquad
B_{n,U}:=\mathbb P_n\big[(\widehat U-\tfrac12)^2\big],
\qquad
B_{n,V}:=\mathbb P_n\big[(\widehat V-\tfrac12)^2\big],
\]
and their population counterparts:
\[
A:=\mathbb{E} \big[(U-\tfrac12)(V-\tfrac12)\big]=\frac{\rho_C}{12},
\qquad
B_U:=\mathbb{E} \big[(U-\tfrac12)^2\big]=\frac{1}{12},
\qquad
B_V:=\mathbb{E} \big[(V-\tfrac12)^2\big]=\frac{1}{12}.
\]
Under Assumption~\ref{ass:emp_process} (ii), let $\widetilde Z=(\widetilde Y,\widetilde W,\widetilde X)$ be an
independent copy of $Z$, and set $\widetilde U=F_{Y\mid X}(\widetilde Y\mid \widetilde X)$ and
$\widetilde V=F_{W\mid X}(\widetilde W\mid \widetilde X)$. For $R\in\{Y,W\}$, define the following conditional
expectation functions (as functions of $Z$):
\begin{align*}
\gamma_Y(Z)
&:=\mathbb{E} \Big[(\widetilde V-\tfrac12)\,\varphi_Y(\widetilde Y,\widetilde X;Z)\ \Big|\ Z\Big],\\
\gamma_W(Z)
&:=\mathbb{E} \Big[(\widetilde U-\tfrac12)\,\varphi_W(\widetilde W,\widetilde X;Z)\ \Big|\ Z\Big],\\
\delta_Y(Z)
&:=2\,\mathbb{E} \Big[(\widetilde U-\tfrac12)\,\varphi_Y(\widetilde Y,\widetilde X;Z)\ \Big|\ Z\Big],\\
\delta_W(Z)
&:=2\,\mathbb{E} \Big[(\widetilde V-\tfrac12)\,\varphi_W(\widetilde W,\widetilde X;Z)\ \Big|\ Z\Big].
\end{align*}
Since $\mathbb{E} [\varphi_R(r,x;Z)]=0$ for all $(r,x)$, it follows that
$\mathbb{E} [\gamma_Y(Z)]=\mathbb{E} [\gamma_W(Z)]=\mathbb{E} [\delta_Y(Z)]=\mathbb{E} [\delta_W(Z)]=0$.
Moreover, Assumption~\ref{ass:emp_process} (ii) implies these functions have finite second moments.

\medskip

\begin{lemma}[Joint asymptotic linearity of $(A_n,B_{n,U},B_{n,V})$]\label{lem:ABlin_fun}
Under Assumptions~\ref{ass:data_fun}--\ref{ass:emp_process}, there exists a mean-zero, finite-variance influence
function vector $\Psi(Z)=(\psi_A(Z),\psi_U(Z),\psi_V(Z))^\top$ such that
\[
\sqrt n\Big((A_n-A,\ B_{n,U}-B_U,\ B_{n,V}-B_V)^\top\Big)
=
\frac{1}{\sqrt n}\sum_{i=1}^n \Psi(Z_i)+o_p(1),
\]
where one may take
\begin{align*}
\psi_A(Z)
&=\Big((U-\tfrac12)(V-\tfrac12)-A\Big)+\gamma_Y(Z)+\gamma_W(Z),\\
\psi_U(Z)
&=\Big((U-\tfrac12)^2-B_U\Big)+\delta_Y(Z),\\
\psi_V(Z)
&=\Big((V-\tfrac12)^2-B_V\Big)+\delta_W(Z).
\end{align*}
\end{lemma}

\begin{theorem}[Asymptotic normality of $\widehat\rho_C$, $\widetilde\rho_C$, and $\breve\rho_C$]
\label{thm:asy_normal_fun}
Under the conditions of Theorem~\ref{thm:consistency_rho_fun} and Lemma~\ref{lem:ABlin_fun},
\[
\sqrt n(\breve\rho_C-\rho_C)\Rightarrow \mathcal{N}(0,\sigma_{\mathrm{cov}}^2),
\qquad
\sqrt n(\widehat\rho_C-\rho_C)\Rightarrow \mathcal{N}(0,\sigma_{\mathrm{ols}}^2),
\qquad
\sqrt n(\widetilde\rho_C-\rho_C)\Rightarrow \mathcal{N}(0,\sigma_{\mathrm{corr}}^2),
\]
where the corresponding influence functions are
\begin{align*}
\psi_{\mathrm{cov}}(Z)
&=12\,\psi_A(Z),\\
\psi_{\mathrm{ols}}(Z)
&=12\Big\{\psi_A(Z)-\rho_C\,\psi_V(Z)\Big\},\\
\psi_{\mathrm{corr}}(Z)
&=12\Big\{\psi_A(Z)-\frac{\rho_C}{2}\big(\psi_U(Z)+\psi_V(Z)\big)\Big\},
\end{align*}
and
\[
\sigma_{\mathrm{cov}}^2=\mathrm{Var} \big(\psi_{\mathrm{cov}}(Z)\big),
\qquad
\sigma_{\mathrm{ols}}^2=\mathrm{Var} \big(\psi_{\mathrm{ols}}(Z)\big),
\qquad
\sigma_{\mathrm{corr}}^2=\mathrm{Var} \big(\psi_{\mathrm{corr}}(Z)\big).
\]
\end{theorem}

The asymptotic variance expressions above are relatively complicated, and implementing their plug-in estimators
can be computationally demanding. We therefore follow \citet{chernozhukov2024conditional} and use the
\emph{exchangeable bootstrap} for inference. We now show that, under the weight conditions below, the bootstrap
statistics in Algorithm~2 have the same asymptotic distribution as the corresponding CRRR estimators.

\medskip

\begin{assumption}\label{ass:bootstrap}
For each $n$, let the weight vector $(\omega_{n1},\ldots,\omega_{nn})$ consist of nonnegative
\emph{exchangeable}\footnote{%
A sequence of random variables $X_1,\ldots,X_n$ is called exchangeable if and only if for every finite
permutation $\sigma$ of $\{1,\ldots,n\}$, the permuted sequence $(X_{\sigma(1)},\ldots,X_{\sigma(n)})$ has the
same joint distribution as $(X_1,\ldots,X_n)$.}
random variables, independent of the observed sample. Assume that there exists $\varepsilon>0$ such that
\begin{equation}\label{eq:exch-bootstrap}
\sup_{n}\mathbb{E}\big[\omega_{n1}^{2+\varepsilon}\big] < \infty,
\qquad
n^{-1}\sum_{i=1}^{n}(\omega_{ni}-\bar{\omega}_n)^2 \xrightarrow{p} 1,
\qquad
\bar{\omega}_n \xrightarrow{p} 1,
\end{equation}
where $\bar{\omega}_n = n^{-1}\sum_{i=1}^{n} \omega_{ni}$.
\end{assumption}

\medskip

To establish bootstrap validity, we adopt the notation and definitions in \citet{van1996weak}. Let
$D_n=(Z_1,\ldots,Z_n)$ be the data vector with $Z_i=(Y_i,W_i,X_i)$, and let
$M_n=(\omega_{n1},\ldots,\omega_{nn})$ be the random bootstrap-weight vector generated independently of $D_n$.
Consider a random element $\mathbb Z_n^*=\mathbb Z_n(D_n,M_n)$ taking values in a normed space $\mathbb E$.
If there exists a tight random element $\mathbb Z$ such that
\[
\sup_{h\in \mathrm{BL}_1(\mathbb E)}
\Big|
\mathbb E_{M_n}\!\left[h\!\left(\mathbb Z_n^*\right)\right]-\mathbb E\!\left[h(\mathbb Z)\right]
\Big|
\xrightarrow{p}0,
\]
where $\mathrm{BL}_1(\mathbb E)$ is the set of functions with Lipschitz norm at most $1$ and $\mathbb E_{M_n}$
denotes conditional expectation with respect to $M_n$ given $D_n$, then we say that the bootstrap distribution of
$\mathbb Z_n^*$ consistently estimates the distribution of $\mathbb Z$, and write
$\mathbb Z_n^* \rightsquigarrow_P \mathbb Z$.

\medskip

\begin{lemma}[Exchangeable bootstrap CLT; \citealp{van1996weak}]\label{lem:bootclt}
Let $\{\xi_i\}_{i=1}^n$ be i.i.d.\ random variables with $\mathbb{E} [\xi_1]=0$ and
$\mathbb{E} [\xi_1^2]=\sigma^2\in(0,\infty)$, and assume there exists $\delta>0$ such that
$\mathbb{E} |\xi_1|^{2+\delta}<\infty$. Under Assumption~\ref{ass:bootstrap}, define
$\widetilde\omega_{ni}:=\omega_{ni}/\bar{\omega}_n$ and
\[
S_n^*=\frac1{\sqrt n}\sum_{i=1}^n(\widetilde\omega_{ni}-1)\xi_i.
\]
Then the bootstrap distribution of $S_n^*$ converges to the normal limit:
\[
S_n^* \rightsquigarrow_{P} \mathcal{N}(0,\sigma^2).
\]
\end{lemma}

\begin{theorem}[Bootstrap consistency]\label{thm:boot_valid}
Under Assumptions~\ref{ass:data_fun}--\ref{ass:emp_process}, Theorem~\ref{thm:asy_normal_fun}, and
Assumption~\ref{ass:bootstrap}, suppose that for any estimator form
$\ddot\rho_C\in\{\widehat\rho_C,\widetilde\rho_C,\breve\rho_C\}$ we have
\[
\sqrt n(\ddot\rho_C-\rho_C)\Rightarrow \mathcal{N}(0,\sigma_\rho^2).
\]
Then the corresponding bootstrap statistic $\ddot\rho_C^{*}$ constructed in Algorithm~2 satisfies
\[
\sqrt n(\ddot\rho_C^{*}-\ddot\rho_C)\ \rightsquigarrow_P\ Z_\rho,
\]
where $Z_\rho\sim \mathcal{N}(0,\sigma_\rho^2)$. Consequently, the standard error and confidence interval produced by
Algorithm~2 are asymptotically valid:
\[
\widehat{\sigma}_\rho \xrightarrow{p} \sigma_\rho,
\qquad
\mathbb P\big\{\rho_C \in \mathrm{ACI}_{1-\alpha}(\rho_C)\big\}\to 1-\alpha.
\]
\end{theorem}

\medskip

Finally, we emphasize that the asymptotic theory in this section is developed only for the continuous case and
under the fixed-complexity regime. In a more general nonparametric regime, network complexity and the Bernstein
order $J$ should increase with $n$ to eliminate approximation bias. However, in this case, achieving an
$o(n^{-1/2})$ convergence rate for neural networks typically requires rather stringent conditions, so an important future direction
is to incorporate the Neyman-orthogonality ideas from double/debiased machine learning (DDML) to relax the rate
requirements; see \cite{chernozhukov2018double}. However, since the main goal of this paper is to propose a new
estimation workflow and demonstrate its empirical effectiveness, we leave the nonparametric theoretical
development to future work.

\section{Simulation Studies}\label{ch:num_exps}

This section conducts simulation experiments under four settings --- simple continuous, complex continuous,
simple discrete ordinal, and complex discrete ordinal --- to comprehensively compare our proposed workflow
(DCTM/dDCTM for conditional-rank estimation with cross-fitting) against the traditional CRRR approach
(distribution regression, DR, for conditional-rank estimation). Throughout this section, we focus on the
OLS-slope estimator as the representative target:
\[
\widehat{\rho}_C
:=
\frac{\sum_{i=1}^n\left(\widehat{U}_i-\overline{\widehat{U}}\right)\left(\widehat{V}_i-\overline{\widehat{V}}\right)}
{\sum_{i=1}^n\left(\widehat{V}_i-\overline{\widehat{V}}\right)^2}.
\]

\medskip

\noindent\textbf{Implementation details.}
For different data-generating scenarios, the model design involves several tuning parameters that must be chosen
and adjusted based on the experimental setting.

\smallskip
\noindent\emph{DR.}
For DR, we choose $M$ threshold grid points and fit a logit distribution regression model pointwise. The
thresholds uniformly cover the $[0.01,0.99]$ sample quantiles of $Y$ (or $W$), with step size $0.98/M$, so that
the grid is dense enough to capture the distributional shape. We then obtain conditional-rank estimates
$\widehat U_i$ (or $\widehat V_i$) for all observations via linear interpolation and local extrapolation. See
Algorithm~1 in \citet{chernozhukov2024conditional} for details.

\smallskip
\noindent\emph{DCTM/dDCTM.}
For DCTM (and dDCTM), we train two separate networks for $Y$ and $W$. The architecture follows
subsection~\ref{sec:dctm}. Hyperparameters such as hidden-layer width/depth, learning
rate, and batch size are tuned for each setting. We use Adam for optimization and early stopping based on the
validation NLL to prevent overfitting. Note the architecture differs between continuous and discrete cases.
We set the number of cross-fitting folds to $K=3$; see Algorithm~1 in the previous section for details.

\subsection{Simple Continuous Setting}\label{sec:simple_continuous}
We first consider a simple continuous setting, using the same simulation design as in
\citet{chernozhukov2024conditional}. Let $Y$ be the child's height, $W$ be the father's height, and $X$ indicate
the country/region (Netherlands: $X=0$; Ireland: $X=1$). Conditional on $X$, $(Y,W)$ follows a bivariate normal
distribution whose mean may depend on $X$, while the covariance matrix is constant:
\[
\binom{Y}{W}\,\Big|\,X=x
\sim
N_2\!\left(
\binom{165}{180-\delta x},
\;
4^2
\begin{pmatrix}
1 & 0.6\\
0.6 & 1
\end{pmatrix}
\right).
\]
We set $\mathbb{P}(X=0)=\mathbb{P}(X=1)=1/2$. The parameter $\delta$ represents a negative shock affecting
Ireland. We consider $\delta\in\{0,12\}$. For each case, we generate $n=500{,}000$ observations.

Note that the conditional Pearson correlation is constant, $\rho_P=0.6$. Up to negligible numerical error, the
conditional Spearman rank correlation $\rho_S$ and Pearson correlation are related by \citep{cramer1999mathematical}
\begin{equation}\label{eq:rhoS_rhoP}
\rho_S=\frac{6}{\pi}\arcsin\!\left(\frac{\rho_P}{2}\right).
\end{equation}
Hence the true CRRR slope is
\[
\rho_C=\rho_S=\frac{6}{\pi}\arcsin\!\left(\frac{\rho_P}{2}\right)\approx 0.58192.
\]

In this simple setting --- especially under normality --- most models perform well. We thus expect DR and DCTM to
achieve similarly good conditional-rank estimation. For DR, we use $M=100$ grid quantiles in
$[0.01,0.99]$. For DCTM, we use a Bernstein order $J=32$; other hyperparameters (learning rate, hidden sizes,
batch size) are tuned based on training diagnostics.

\subsubsection{Conditional CDF fitting}
We first evaluate conditional CDF estimation (the first-stage task in CRRR). Figures
\ref{fig:simple_normal_delta0_cdf} and \ref{fig:simple_normal_delta12_cdf} plot the estimated conditional CDF
curves under $\delta=0$ and $\delta=12$. As expected, DR and DCTM produce nearly identical estimates that overlap
almost perfectly with the truth, providing a strong foundation for estimating $\rho_C$.
\begin{figure}[htbp]
\centering
\includegraphics[width=0.65\textwidth]{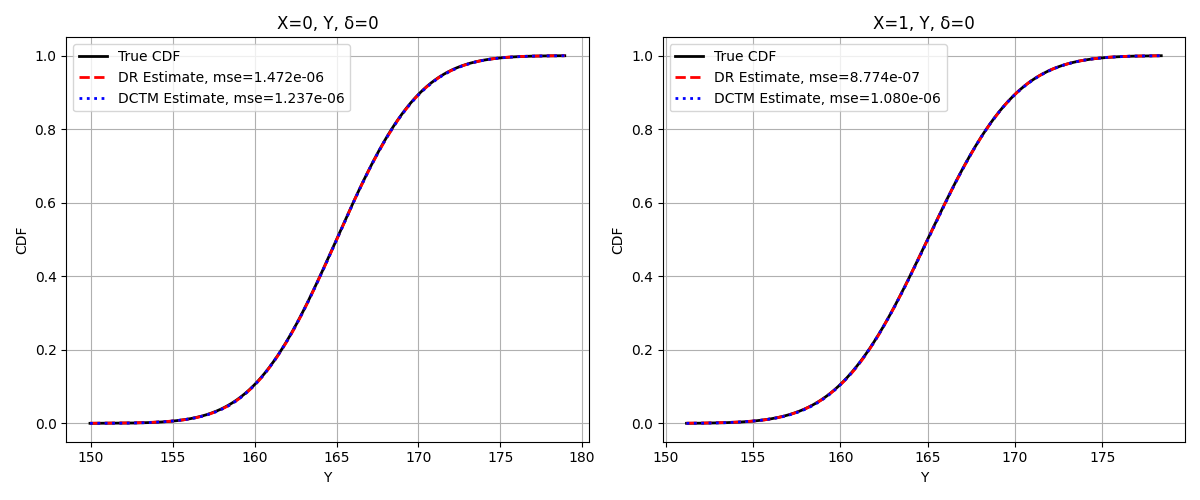}
\includegraphics[width=0.65\textwidth]{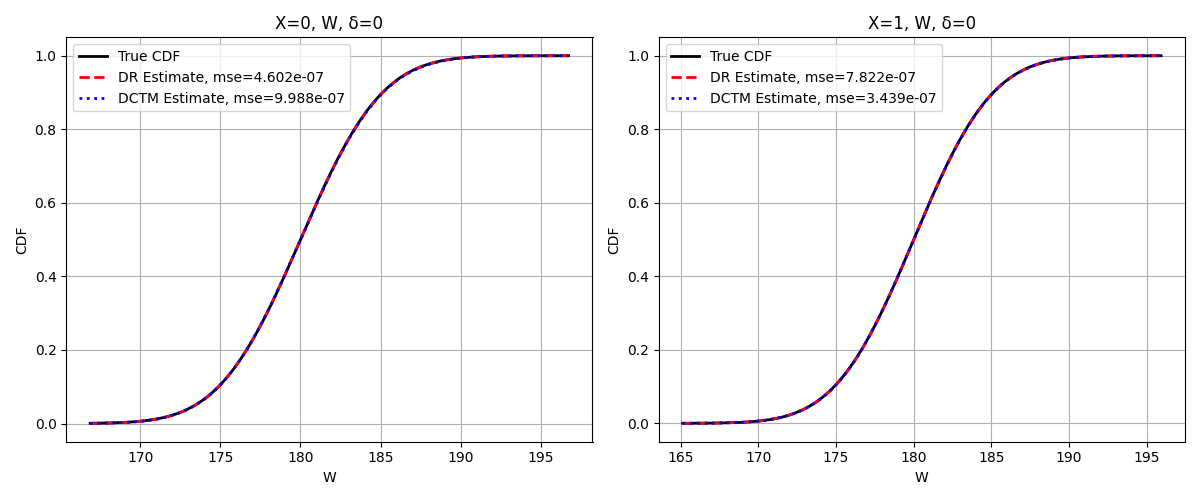}
\caption{Simple continuous setting with $\delta=0$: conditional CDFs estimated by DR and DCTM versus the truth.}
\label{fig:simple_normal_delta0_cdf}
\end{figure}
\begin{figure}[htbp]
\centering
\includegraphics[width=0.65\textwidth]{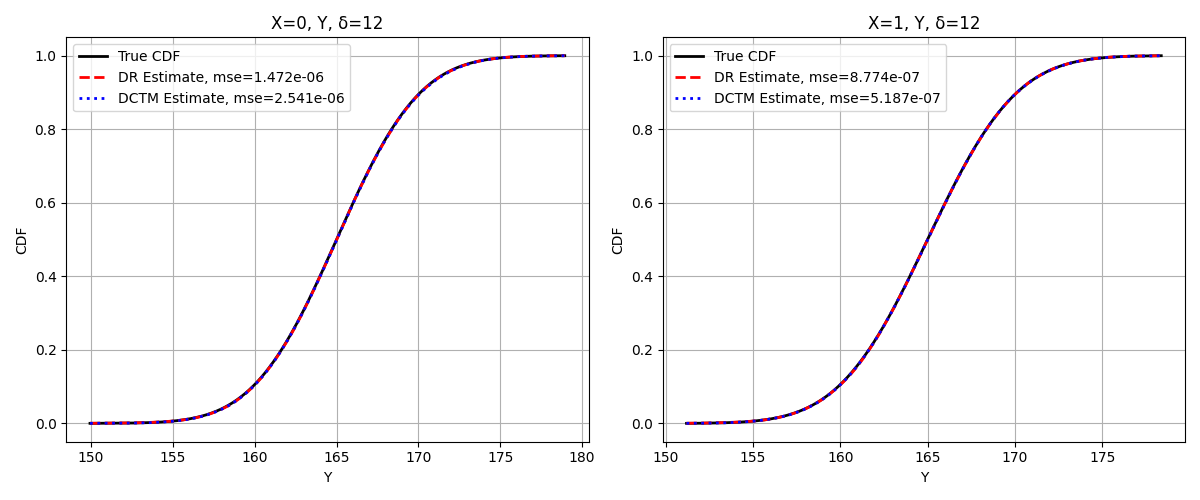}
\includegraphics[width=0.65\textwidth]{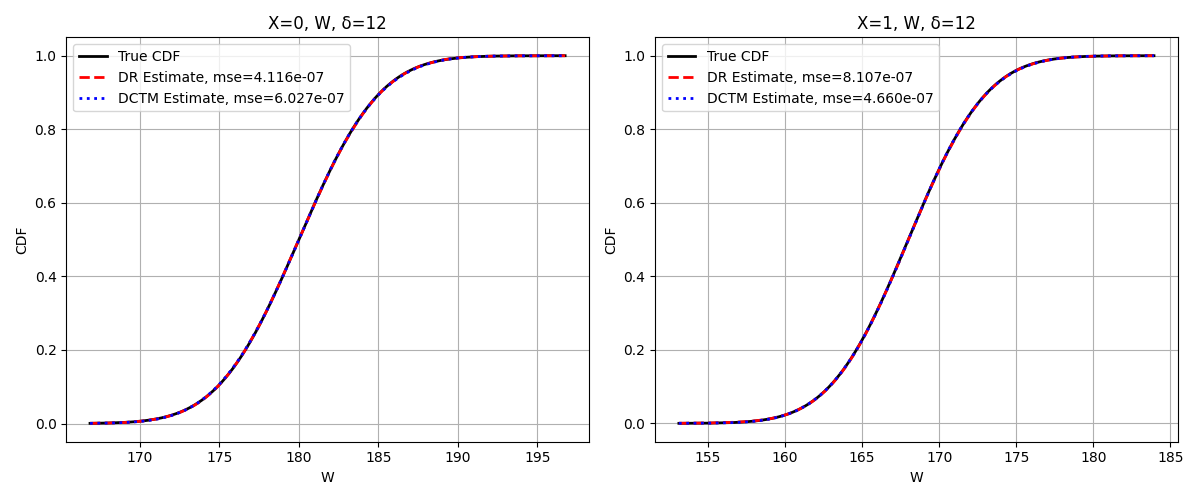}
\caption{Simple continuous setting with $\delta=12$: conditional CDFs estimated by DR and DCTM versus the truth.}
\label{fig:simple_normal_delta12_cdf}
\end{figure}

\subsubsection{Calibration via PIT}
To further assess calibration, we perform probability integral transform (PIT) diagnostics for each combination
of (target variable $Y$ or $W$), ($X=0$ or $1$), and ($\delta=0$ or $12$). If the estimated conditional CDF is
correct, then $\{\widehat U_i\}$ and $\{\widehat V_i\}$ should be approximately i.i.d.\ $\mathrm{Unif}(0,1)$.
Therefore, the PIT distribution reveals systematic bias or over-/under-confidence in conditional distribution
estimation.

Figures \ref{fig:pit_dr} and \ref{fig:pit_dctm} display PIT histograms (with the ideal uniform reference line)
and QQ plots. The PIT histograms for both DR and DCTM are fairly flat, without prominent U-shaped or hump-shaped
patterns. The QQ plots closely follow the 45-degree line, with only mild deviations in extreme quantiles.
Overall, both methods exhibit good calibration with no obvious systematic bias.
\begin{figure}[htbp]
\centering
\includegraphics[scale=0.22]{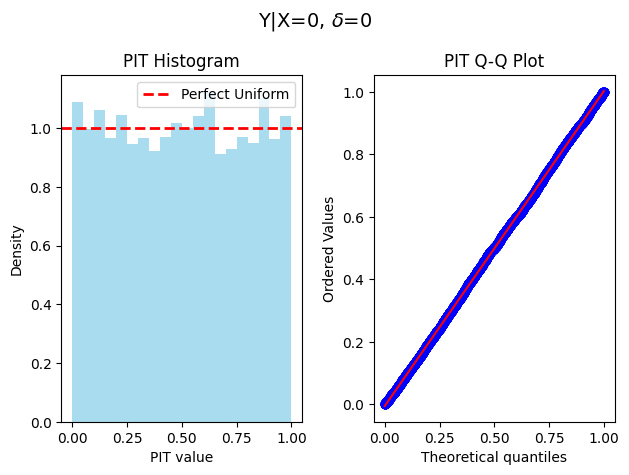}
\includegraphics[scale=0.22]{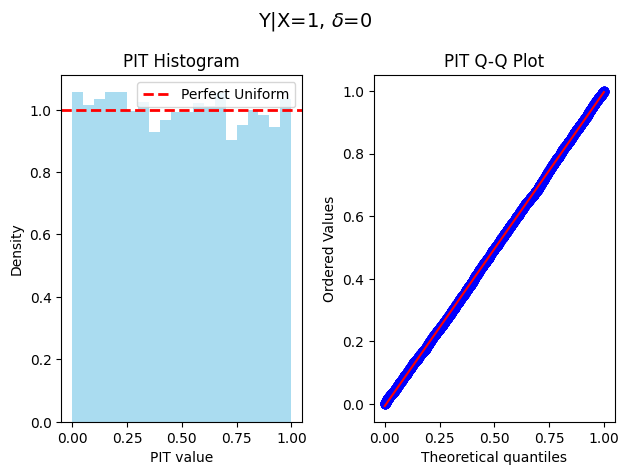}
\includegraphics[scale=0.22]{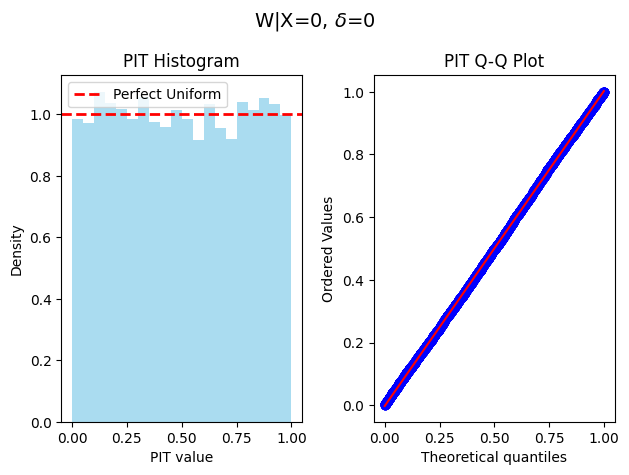}
\includegraphics[scale=0.22]{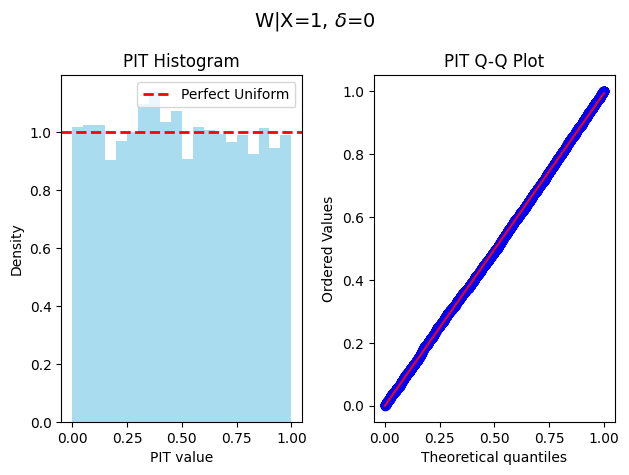}
\includegraphics[scale=0.22]{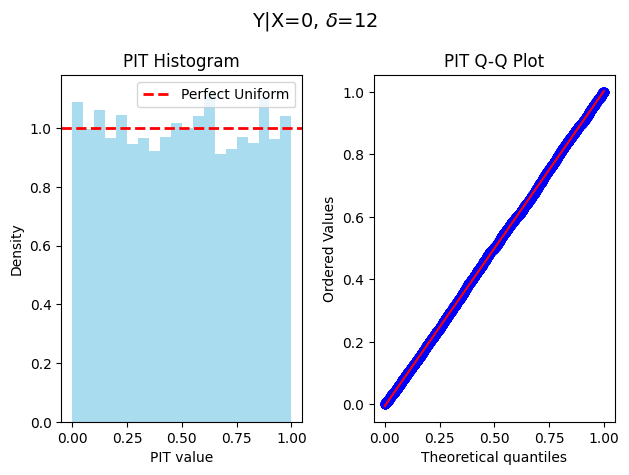}
\includegraphics[scale=0.22]{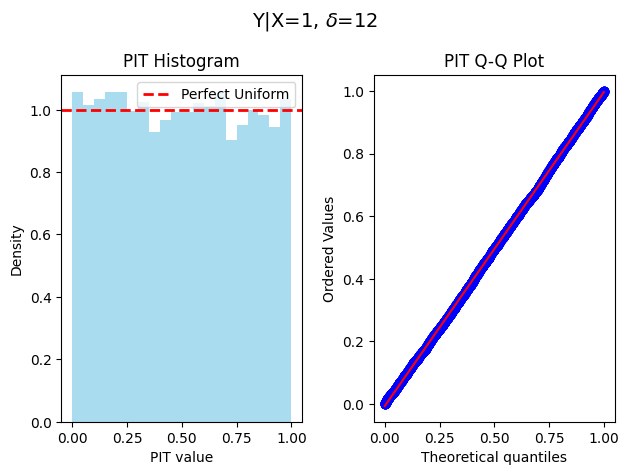}
\includegraphics[scale=0.22]{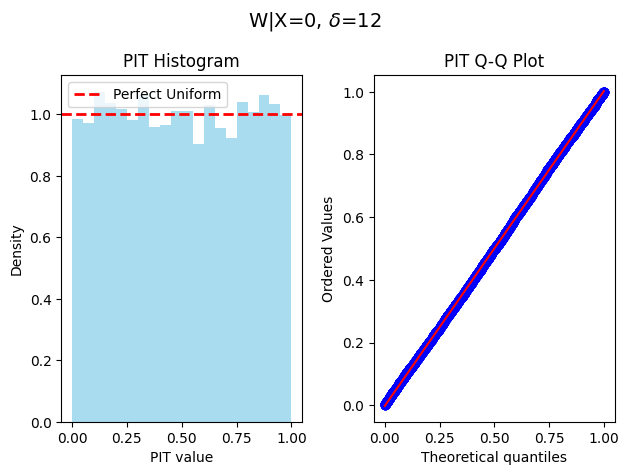}
\includegraphics[scale=0.22]{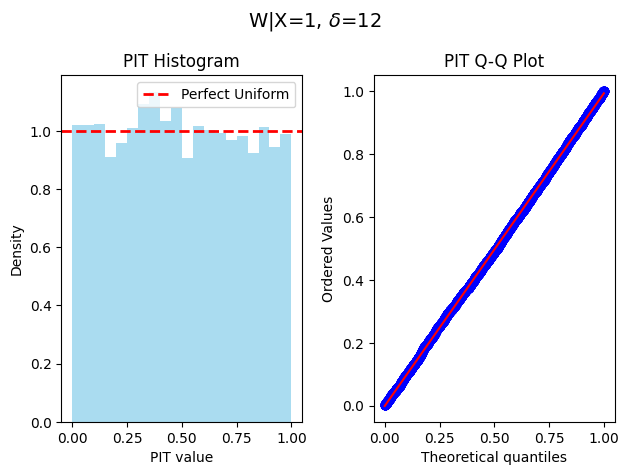}
\caption{PIT diagnostics for DR-based conditional CDF estimation.}
\label{fig:pit_dr}
\end{figure}
\begin{figure}[htbp]
\centering
\includegraphics[scale=0.22]{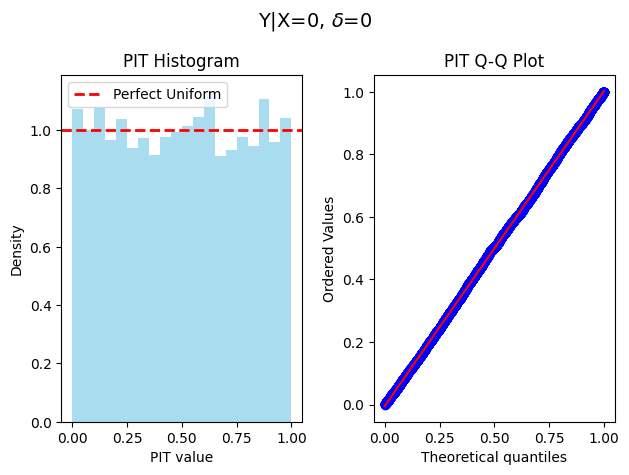}
\includegraphics[scale=0.22]{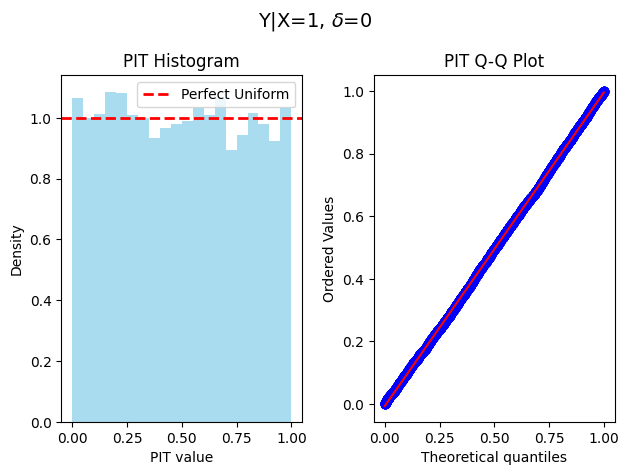}
\includegraphics[scale=0.22]{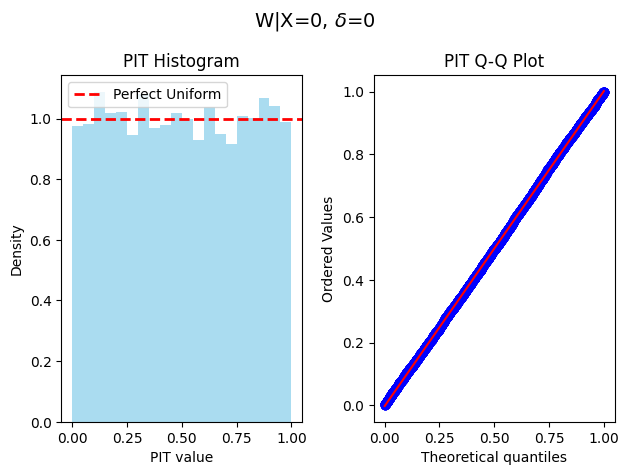}
\includegraphics[scale=0.22]{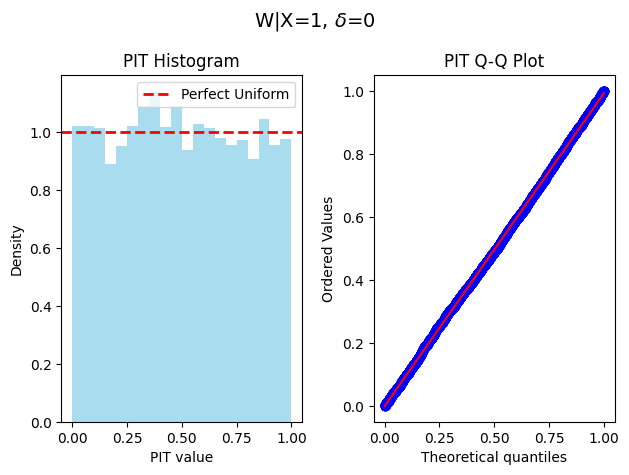}
\includegraphics[scale=0.22]{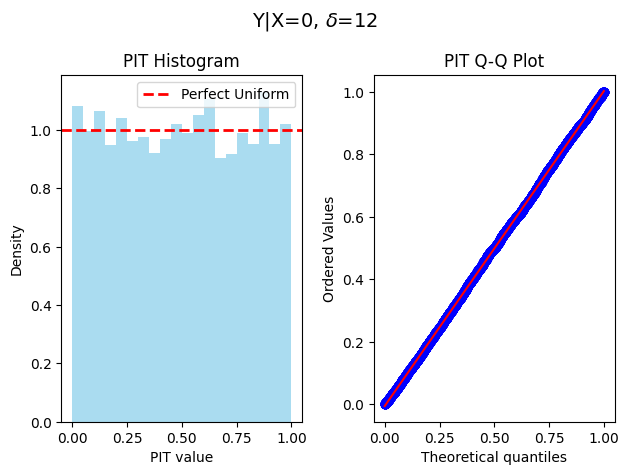}
\includegraphics[scale=0.22]{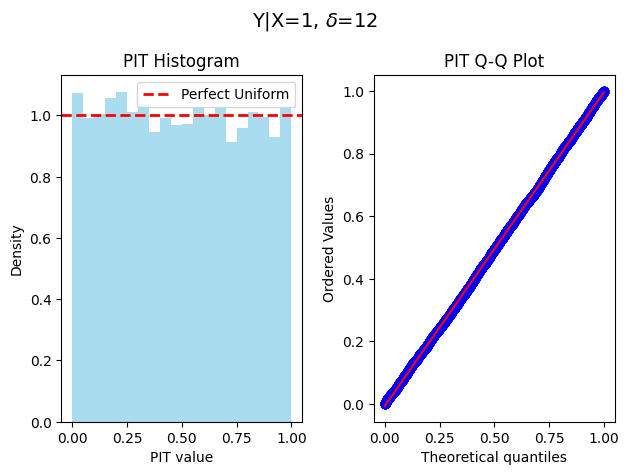}
\includegraphics[scale=0.22]{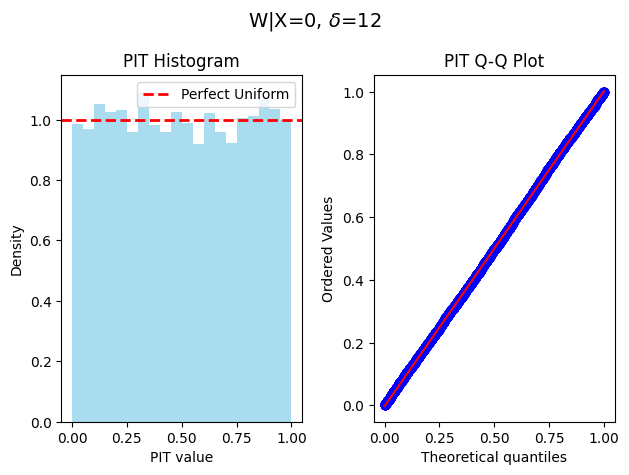}
\includegraphics[scale=0.22]{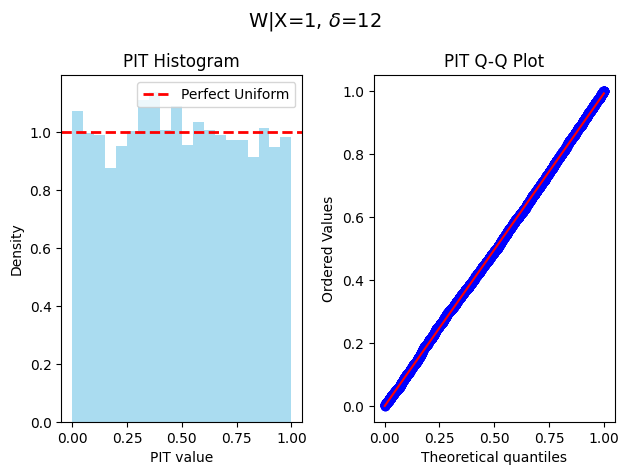}
\caption{PIT diagnostics for DCTM-based conditional CDF estimation.}
\label{fig:pit_dctm}
\end{figure}

\subsubsection{Quantitative PIT checks}
To quantify the graphical diagnostics, we compute (i) the Kolmogorov--Smirnov (KS) test $p$-value,
(ii) the chi-squared ($\chi^2$) goodness-of-fit test $p$-value, and (iii) the sample mean/variance of PIT for each
case. The results are reported in Tables \ref{tab:delta0-ks-chi-table} and \ref{tab:delta12-ks-chi-table}.
Across all cases, the KS and $\chi^2$ test $p$-values are well above $\alpha=0.05$, so we do not reject the null
that PIT is uniform. The PIT means are close to $1/2$ and the PIT variances are close to $1/12$, consistent with
the figures and confirming good calibration in this simple normal setting.

\begin{table}[htbp]
\centering
\caption{Simple continuous setting with $\delta=0$: KS $p$-values, $\chi^2$ $p$-values, and PIT mean/variance.}
\label{tab:delta0-ks-chi-table}
\begingroup
\setlength{\tabcolsep}{6pt}
\begin{tabular}{lccccc}
\toprule
 & \textbf{Method} & \textbf{KS \textit{$p$-value}} & \textbf{$\chi^2$ \textit{$p$-value}} & \textbf{PIT Mean} & \textbf{PIT Var} \\
\midrule
\multirow{2}{*}{$Y\mid X=0$}  & DR   & 0.9024 & 0.4708 & 0.4986 & 0.0845\\
                             & DCTM & 0.8850 & 0.4024 & 0.4988 & 0.0845\\
\cline{2-6}
\multirow{2}{*}{$Y\mid X=1$}  & DR   & 0.3220 & 0.9870 & 0.4948 & 0.0839\\
                             & DCTM & 0.2174 & 0.8696 & 0.4941 & 0.0838\\
\midrule
\multirow{2}{*}{$W\mid X=0$}  & DR   & 0.9180 & 0.9710 & 0.4999 & 0.0841\\
                             & DCTM & 0.9327 & 0.9117 & 0.4997 & 0.0840\\
\cline{2-6}
\multirow{2}{*}{$W\mid X=1$}  & DR   & 0.3118 & 0.6169 & 0.4958 & 0.0824\\
                             & DCTM & 0.3444 & 0.5482 & 0.4961 & 0.0823\\
\bottomrule
\end{tabular}
\endgroup
\end{table}

\begin{table}[htbp]
\centering
\caption{Simple continuous setting with $\delta=12$: KS $p$-values, $\chi^2$ $p$-values, and PIT mean/variance.}
\label{tab:delta12-ks-chi-table}
\begingroup
\setlength{\tabcolsep}{6pt}
\begin{tabular}{lccccc}
\toprule
 & \textbf{Method} & \textbf{KS \textit{$p$-value}} & \textbf{$\chi^2$ \textit{$p$-value}} & \textbf{PIT Mean} & \textbf{PIT Var} \\
\midrule
\multirow{2}{*}{$Y\mid X=0$}  & DR   & 0.9024 & 0.4708 & 0.4986 & 0.0845\\
                             & DCTM & 0.9217 & 0.3770 & 0.4984 & 0.0845\\
\cline{2-6}
\multirow{2}{*}{$Y\mid X=1$}  & DR   & 0.3220 & 0.9870 & 0.4948 & 0.0839\\
                             & DCTM & 0.3127 & 0.9800 & 0.4946 & 0.0838\\
\midrule
\multirow{2}{*}{$W\mid X=0$}  & DR   & 0.9199 & 0.9540 & 0.5000 & 0.0842\\
                             & DCTM & 0.9080 & 0.9522 & 0.5000 & 0.0840\\
\cline{2-6}
\multirow{2}{*}{$W\mid X=1$}  & DR   & 0.3075 & 0.6745 & 0.4958 & 0.0824\\
                             & DCTM & 0.3609 & 0.4394 & 0.4960 & 0.0823\\
\bottomrule
\end{tabular}
\endgroup
\end{table}

\subsubsection{Slope estimation}
Finally, we compare $\rho_C$ estimation accuracy. Table~\ref{tab:normal-crrr-slope-table} reports Monte Carlo
performance for DR and DCTM, based on 30 repetitions. In each repetition we use $n=10^5$ observations and report
the mean, standard deviation (SD), mean absolute deviation (MAD), and root mean squared error (RMSE) of
$\widehat\rho_C$.

\begin{table}[htbp]
\centering
\caption{Simple continuous setting: estimation results for the CRRR slope $\rho_C$.}
\label{tab:normal-crrr-slope-table}
\begingroup
\setlength{\tabcolsep}{6pt}
\begin{tabular}{lccccc}
\toprule
 & \textbf{Method} & \textbf{Mean} & \textbf{SD} & \textbf{MAD} & \textbf{RMSE}\\
\midrule
\multirow{2}{*}{$\delta=0$}   & DR   & 0.58205 & 0.00080 & 0.00066 & 0.00079\\
                             & DCTM & 0.58216 & 0.00091 & 0.00073 & 0.00092\\
\midrule
\multirow{2}{*}{$\delta=12$}  & DR   & 0.58190 & 0.00085 & 0.00068 & 0.00084\\
                             & DCTM & 0.58206 & 0.00112 & 0.00084 & 0.00111\\
\bottomrule
\end{tabular}
\endgroup

\vspace{0.5ex}
\footnotesize Note: errors are computed relative to the theoretical truth $\rho_C\approx 0.58192$.
\end{table}

In summary, under this simple continuous (normal) setting, both DR and DCTM perform similarly well for both
conditional CDF estimation and the slope $\rho_C$, with DR slightly better in some cases. This is expected:
when DR is correctly specified, it can be highly competitive. The natural question is whether differences emerge
under more complex and misspecified settings, which we investigate next.

\subsection{Complex Continuous Setting}\label{sec:continuous}
In socio-economic applications, data are often far more complex than the simple normal setting above, involving
nonlinearities, strong feature interactions, and heteroskedasticity. To assess the advantage of our approach in
such settings, we consider a complex continuous DGP with nonlinearity and high-order interactions.

Let latent variables $(Z_1,Z_2)$ follow
\[
(Z_1,Z_2)\sim \mathcal{N}\!\left(
\begin{pmatrix}0\\0\end{pmatrix},
\begin{pmatrix}
1 & \rho_0\\
\rho_0 & 1
\end{pmatrix}
\right),
\qquad \rho_0\in(0,1).
\]
Let $X=(X_1,\ldots,X_p)$ be $p$-dimensional covariates with i.i.d.\ components
$X_j\sim \mathrm{Unif}[-1,1]$. We generate $Y$ and $W$ by nonlinear transformations:
\begin{align}
Y &= m_Y(X)+s_Y(X)\,g_Y(Z_1), \label{eq:dgp_Y}\\
W &= m_W(X)+s_W(X)\,g_W(Z_2), \label{eq:dgp_W}
\end{align}
where $m_Y(\cdot)$ and $m_W(\cdot)$ are nonlinear mean functions, $s_Y(\cdot)$ and $s_W(\cdot)$ are positive
heteroskedastic scale functions capturing conditional standard deviations, and $g_Y(\cdot)$ and $g_W(\cdot)$ are
strictly monotone nonlinear transforms (e.g., softplus-based). This construction induces both skewness and
heteroskedasticity in $Y\mid X$ and $W\mid X$: the monotone nonlinear transforms distort symmetry, while the
scale functions make dispersion depend on $X$. Such a DGP is closer to many real-world datasets and provides a
stringent test of conditional distribution modeling.

Importantly, since $g_Y$ and $g_W$ are strictly monotone, the Spearman's rank correlation between $Y$ and $W$
conditional on $X$ is primarily determined by the rank dependence of $(Z_1,Z_2)$. For a Gaussian copula, the
approximate relation \eqref{eq:rhoS_rhoP} yields the theoretical CRRR slope (used as the benchmark truth):
\[
\rho_C=\rho_S=\frac{6}{\pi}\arcsin\!\left(\frac{\rho_0}{2}\right).
\]

\subsubsection{DGP specification}
We generate a sample $\{(Y_i,W_i,X_i)\}_{i=1}^n$ with $n=5\times 10^5$ and $p=8$. In
\eqref{eq:dgp_Y}--\eqref{eq:dgp_W}, we set
\[
\begin{aligned}
m_Y(X) &= 6\sin(\pi X_1X_2)+2(X_3^2-X_4^2),\\
m_W(X) &= 4\cos(\pi X_1)+3X_2X_3,\\
s_Y(X) &= \exp\!\left(0.5+0.6X_5+0.5X_6X_7-0.3X_8\right),\\
s_W(X) &= \exp\!\left(0.3+0.5X_4-0.5X_5X_6\right),\\
g_Y(z) &= z+\alpha_Y\,\mathrm{softplus}(\beta_Y z),\\
g_W(z) &= z+\alpha_W\,\mathrm{softplus}(\beta_W z),
\end{aligned}
\]
with $\alpha_Y=0.8$, $\beta_Y=1.2$, $\alpha_W=0.6$, and $\beta_W=1.0$.

We set $\rho_0=0.6$. Plugging into \eqref{eq:rhoS_rhoP} yields the theoretical truth
\[
\rho_C=\frac{6}{\pi}\arcsin(0.6/2)\approx 0.58192.
\]
That is, regressing ranks of $Y$ on ranks of $W$ should yield a slope around $0.58$.

\subsubsection{Conditional CDF comparisons at representative covariates}
In the DGP design described above, distributional variation with respect to \(X\) is jointly driven by interaction terms, curvature terms, and nonlinear components. Accordingly, we construct the following one-dimensional index \(s(X)\) to characterize the complexity of covariate-group heterogeneity:
\[
s(X)=0.8 X_1 X_2+0.6\left(X_3^2-X_4^2\right)+0.8 \sin \left(\pi X_5\right)+0.6 X_6 X_7,
\]
where \(X_1X_2\) and \(X_6X_7\) capture high-dimensional interaction effects, \(X_3^2-X_4^2\) represents polynomial-type nonlinear curvature, and \(\sin(\pi X_5)\) captures a smooth non-polynomial nonlinearity. The coefficients are chosen to balance the relative contribution of each component to the variation in the simulated sample. We then select covariate groups around the $10\%$ and $90\%$ percentiles of \(s(X)\) as representative groups and plot their estimated conditional CDF curves, as shown in Figure~\ref{fig:complex_cdf_compare}. Importantly, \(s(X)\) is used solely to select representative groups for visualization and does not enter the estimation or inference of \(\rho_C\).
\begin{figure}[htbp]
\centering
\includegraphics[scale=0.43]{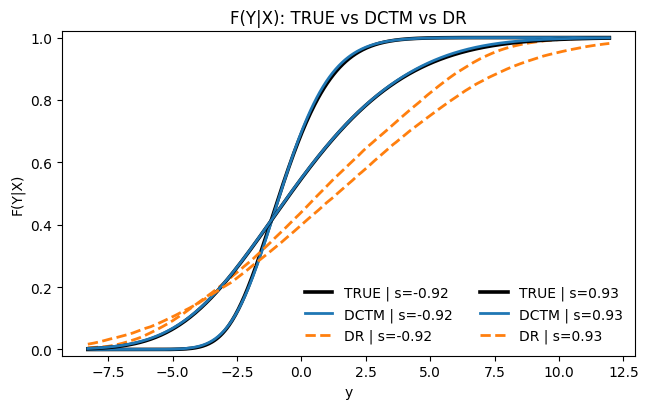}
\includegraphics[scale=0.43]{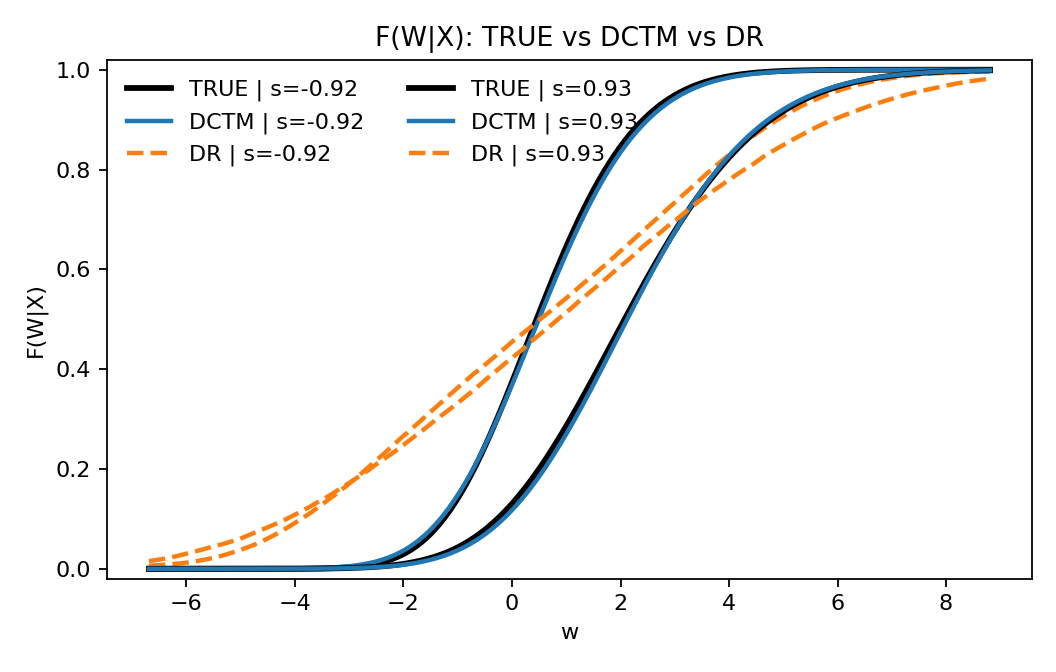}
\caption{Complex continuous setting: DR, DCTM, and true conditional CDFs for two representative covariate groups.}
\label{fig:complex_cdf_compare}
\end{figure}

Figure~\ref{fig:complex_cdf_compare} shows that DCTM tracks the true distributional shape much more closely:
across both central and tail regions, the DCTM curves almost overlap with the truth, yielding small global
shape error. In contrast, DR exhibits clear systematic deviations, indicating limited expressive power under
this nonlinear, interaction-rich DGP. Therefore, DCTM produces more accurate first-stage ranks, which in turn
supports more accurate slope estimation.

\subsubsection{Slope estimation}
Table~\ref{tab:continuous-slope-table} reports the Monte Carlo performance for estimating $\rho_C$, based on 30
repetitions with $n=10^5$ observations per repetition (Mean, SD, MAD, RMSE).
\setlength{\heavyrulewidth}{1.2pt}
\setlength{\lightrulewidth}{0.6pt}
\begin{table}[htbp]
\centering
\caption{Complex continuous setting: estimation results for the CRRR slope $\rho_C$.}
\label{tab:continuous-slope-table}
\begingroup
\setlength{\tabcolsep}{8pt}
\begin{tabular}{lcccc}
\toprule
\textbf{Method} & \textbf{Mean} & \textbf{SD} & \textbf{MAD} & \textbf{RMSE}\\
\midrule
DR   & 0.14761 & 0.00158 & 0.43432 & 0.43432\\
DCTM & 0.57646 & 0.00119 & 0.00547 & 0.00559\\
\bottomrule
\end{tabular}
\endgroup

\vspace{0.5ex}
\footnotesize Note: 30 repetitions; errors are computed relative to the theoretical truth $\rho_C\approx 0.58192$.
\end{table}

The DR-based estimator exhibits severe bias: its mean is about $0.148$, far from the truth, with MAD and RMSE
around $0.43$. This indicates substantial model misspecification in DR for conditional CDF estimation under this
DGP; the resulting structural rank errors propagate directly into $\widehat\rho_C$, causing serious
underestimation. In contrast, DCTM achieves near-oracle first-stage rank estimation and fits the complex
conditional distributions well, yielding a slope mean close to the truth (about $0.576$) and dramatically
smaller MAD/RMSE (around $0.005$).

Overall, this nonlinear continuous simulation demonstrates that DCTM substantially outperforms DR in both
conditional CDF fitting and $\rho_C$ estimation, expanding the applicability and robustness of CRRR for complex
continuous economic outcomes.

\subsection{Simple Discrete Ordinal Setting}
In research on intergenerational mobility, discrete ordinal outcomes (e.g., years of schooling, occupational class) arise frequently. To the best of our knowledge, existing studies of CRRR focus on continuous outcomes. We therefore consider CRRR modeling with discrete ordinal outcomes and develop an associated estimation strategy, which is verified through numerical simulation experiments below.

First, consider the simple discrete ordinal case. We start from the continuous DGP in Section~\ref{sec:simple_continuous} to generate continuous $Y$ and $W$, and
then discretize them into 8 ordered categories $Y_d$ and $W_d$ using equally spaced quantile cutoffs
$\{c_k\}$. The cutoffs are computed under the baseline distribution at $\delta=0$
($Y\sim \mathcal{N}(165,4^2)$ and $W\sim \mathcal{N}(180,4^2)$) and are then reused for $\delta=12$.

In the discrete case, the rank must be explicitly defined as in \eqref{eq:rankY_omega_en}--\eqref{eq:rankW_omega_en}.
Here we illustrate three representative choices $\omega\in\{0.0,0.5,1.0\}$.

\subsubsection{Model specification}
For DR, for each category $k$ we fit a logit model for the binary outcome $\mathbb{I}\{Y_d\le k\}$ to estimate
$F_k(x)=\mathbb P(Y_d\le k\mid X=x)$, and then impose monotonicity across $k$ by rearrangement to ensure
nondecreasing cumulative probabilities. For dDCTM, the network outputs all category probabilities end-to-end; see
subsection~\ref{sec:dctm} (discrete case). Hyperparameters are tuned as needed; we use
Adam with early stopping based on validation NLL.

\subsubsection{Conditional CDF fitting}
Figures \ref{fig:simple_discrete_cdf_delta0} and \ref{fig:simple_discrete_cdf_delta12} show the estimated and
true step-function CDFs for $\delta=0$ and $\delta=12$, with $X\in\{0,1\}$. In this simple discrete setting, DR
and dDCTM achieve very similar and accurate CDF estimates, almost overlapping with the truth.

\begin{figure}[htbp]
\centering
\includegraphics[scale=0.4]{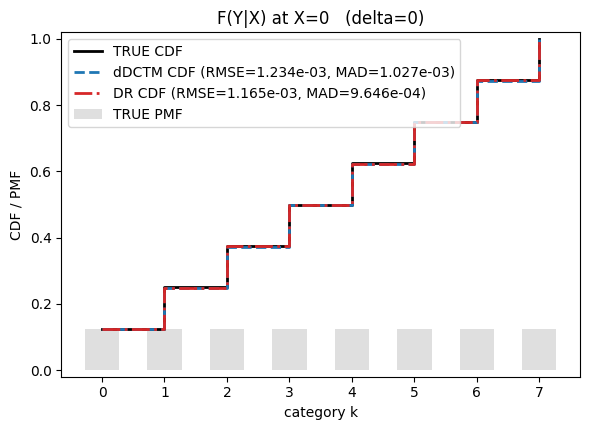}
\includegraphics[scale=0.4]{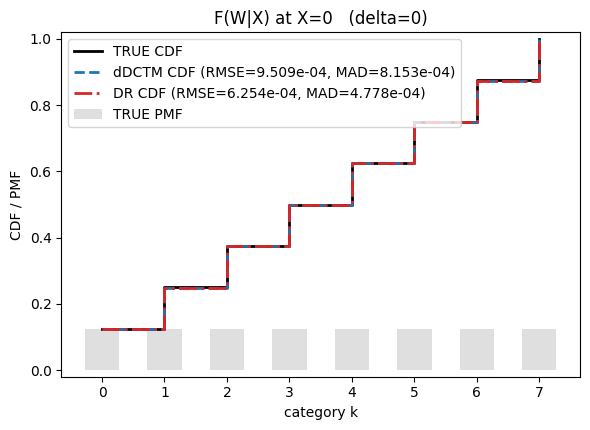}
\includegraphics[scale=0.4]{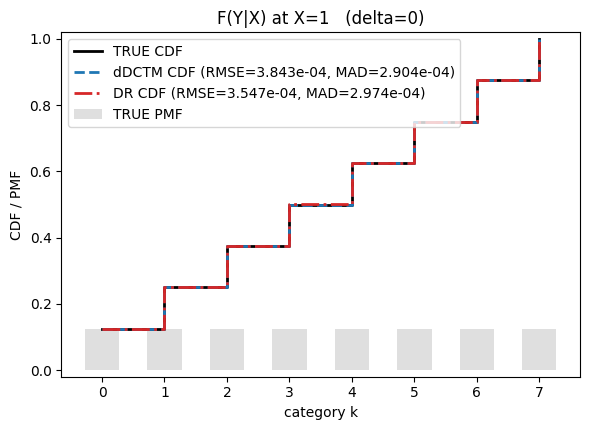}
\includegraphics[scale=0.4]{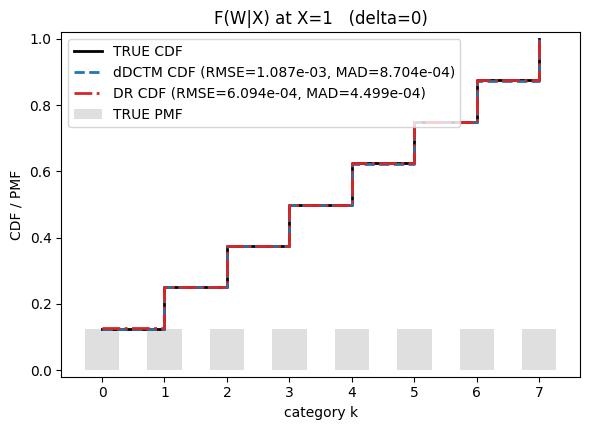}
\caption{Simple discrete ordinal setting with $\delta=0$ and $X\in\{0,1\}$: estimated vs.\ true step CDFs for $F(Y_d\mid X)$ and $F(W_d\mid X)$.}
\label{fig:simple_discrete_cdf_delta0}
\end{figure}

\begin{figure}[htbp]
\centering
\includegraphics[scale=0.4]{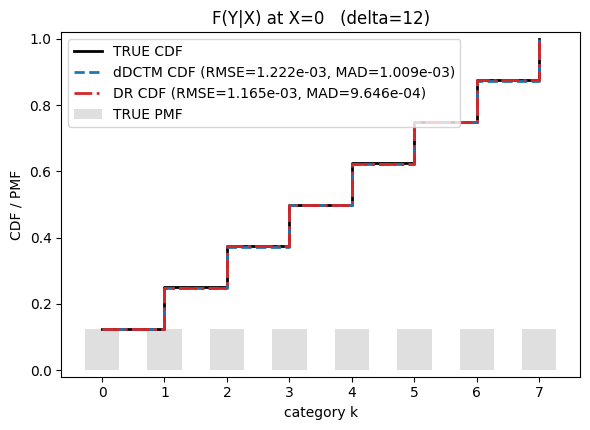}
\includegraphics[scale=0.4]{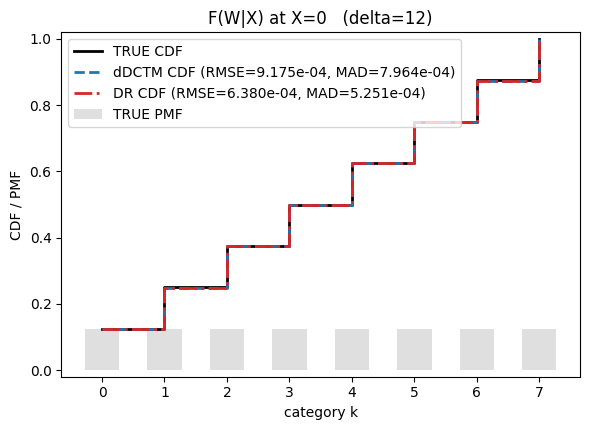}
\includegraphics[scale=0.4]{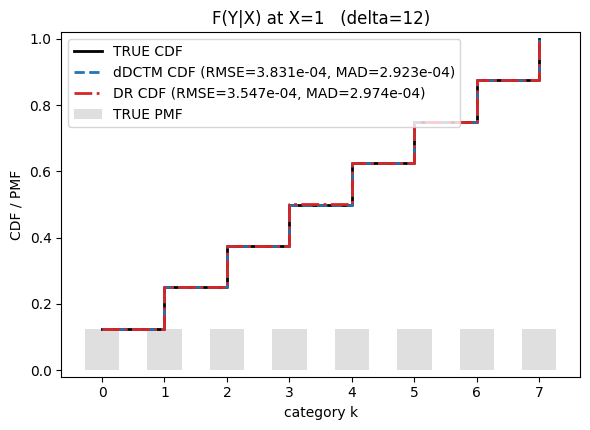}
\includegraphics[scale=0.4]{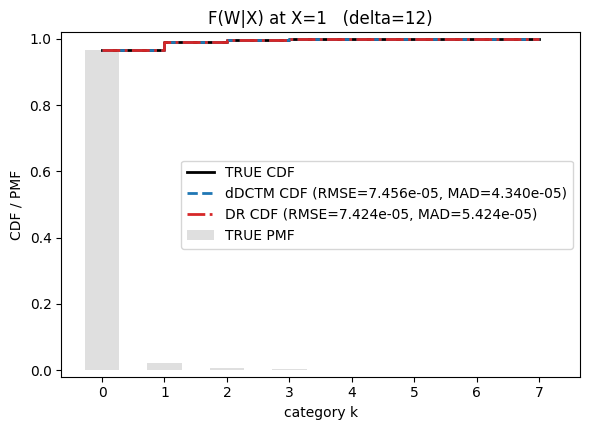}
\caption{Simple discrete ordinal setting with $\delta=12$ and $X\in\{0,1\}$: estimated vs.\ true step CDFs for $F(Y_d\mid X)$ and $F(W_d\mid X)$.}
\label{fig:simple_discrete_cdf_delta12}
\end{figure}

\subsubsection{Slope estimation under different $\omega$}
Table~\ref{tab:simple-omega-slope-table} reports Monte Carlo performance for $\omega\in\{0.0,0.5,1.0\}$, based on
30 repetitions with $n=100{,}000$ per repetition. We report the Mean, SD, MAD, and RMSE of $\widehat\rho_C$.
The ``true'' slope $\rho_C^{\mathrm{True}}$ for each $\omega$ is computed as follows: for each category $k$ we
know its interval $(c_k,c_{k+1}]$, and using the normal CDF we can obtain the population values of
$F_{Y_d\mid X}^{-}(k\mid x)$ and $F_{Y_d\mid X}(k\mid x)$. Applying the rank definition for the given $\omega$
then yields $\rho_C^{\mathrm{True}}$.

\begin{table}[htbp]
\centering
\renewcommand{\arraystretch}{1.5}
\caption{Simple discrete ordinal setting: CRRR slope $\rho_C$ estimation for $\omega\in\{0.0,0.5,1.0\}$.}
\label{tab:simple-omega-slope-table}
\begingroup
\setlength{\tabcolsep}{4pt}
\begin{tabular}{lccccccc}
\toprule
 &  & \textbf{Method} & \textbf{$\rho_C^{\mathrm{True}}$} & \textbf{Mean$(\widehat\rho_C)$} & \textbf{SD$(\widehat\rho_C)$} & \textbf{MAD} & \textbf{RMSE}\\
\midrule
\multirow{4}{*}{$\omega=0.0$}
& \multirow{2}{*}{$\delta=0$}  & DR    & 0.56920 & 0.56767 & 0.00255 & 0.00234 & 0.00293\\
&                              & dDCTM & 0.56920 & 0.56762 & 0.00252 & 0.00234 & 0.00293\\
\cline{3-8}
& \multirow{2}{*}{$\delta=12$} & DR    & 0.29507 & 0.29415 & 0.00179 & 0.00166 & 0.00199\\
&                              & dDCTM & 0.29507 & 0.29413 & 0.00210 & 0.00195 & 0.00227\\
\midrule
\multirow{4}{*}{$\omega=0.5$}
& \multirow{2}{*}{$\delta=0$}  & DR    & 0.56920 & 0.56760 & 0.00250 & 0.00241 & 0.00293\\
&                              & dDCTM & 0.56920 & 0.56759 & 0.00251 & 0.00242 & 0.00294\\
\cline{3-8}
& \multirow{2}{*}{$\delta=12$} & DR    & 0.58077 & 0.57855 & 0.00339 & 0.00330 & 0.00401\\
&                              & dDCTM & 0.58077 & 0.57848 & 0.00348 & 0.00341 & 0.00412\\
\midrule
\multirow{4}{*}{$\omega=1.0$}
& \multirow{2}{*}{$\delta=0$}  & DR    & 0.56920 & 0.56752 & 0.00250 & 0.00249 & 0.00297\\
&                              & dDCTM & 0.56920 & 0.56756 & 0.00254 & 0.00251 & 0.00299\\
\cline{3-8}
& \multirow{2}{*}{$\delta=12$} & DR    & 0.28573 & 0.28446 & 0.00180 & 0.00181 & 0.00218\\
&                              & dDCTM & 0.28573 & 0.28430 & 0.00198 & 0.00197 & 0.00241\\
\bottomrule
\end{tabular}
\endgroup

\vspace{0.5ex}
\footnotesize Note: 30 repetitions; sample size per repetition is $n=100{,}000$.
\end{table}

Table~\ref{tab:simple-omega-slope-table} highlights two points. First, $\rho_C$ differs substantially across
$\omega$, confirming that $\rho_C$ is sensitive to the tie-handling rule in discrete outcomes; empirical studies
must pre-specify and report the rank definition for meaningful interpretation. Second, in this simple discrete
ordinal setting, both DR and dDCTM achieve good estimation accuracy for $\rho_C$.

\subsection{Complex Discrete Ordinal Setting}
The previous section considered a simple discrete scenario. To better emulate real-world complexity, we now
consider a complex discrete ordinal DGP.

We begin with the complex continuous DGP in Section~\ref{sec:continuous} (see \eqref{eq:dgp_Y}--\eqref{eq:dgp_W})
to generate continuous $Y$ and $W$. We then discretize them into 8 ordered categories $Y_d$ and $W_d$ using
pre-specified cut points $\{c_k\}$, keeping all other DGP components unchanged. In the discrete case, ties are
prevalent, and the regression slope depends on the tie parameter $\omega$ in \eqref{eq:rankY_omega_en} and
\eqref{eq:rankW_omega_en}; see Figure~\ref{fig:rho_omega}. We again report results for
$\omega\in\{0.0,0.5,1.0\}$, with hyperparameters tuned as needed.

To visualize conditional distribution fitting in this complex discrete ordinal setting, we randomly select a
fixed covariate value $X=x$ and plot estimated versus true step-function CDFs for $F(Y_d\mid X=x)$ and
$F(W_d\mid X=x)$. As shown in Figure~\ref{fig:discrete_cdf_compare}, dDCTM nearly overlaps with the truth,
indicating high accuracy, while DR exhibits noticeable deviations.

\begin{figure}[htbp]
\centering
\includegraphics[scale=0.48]{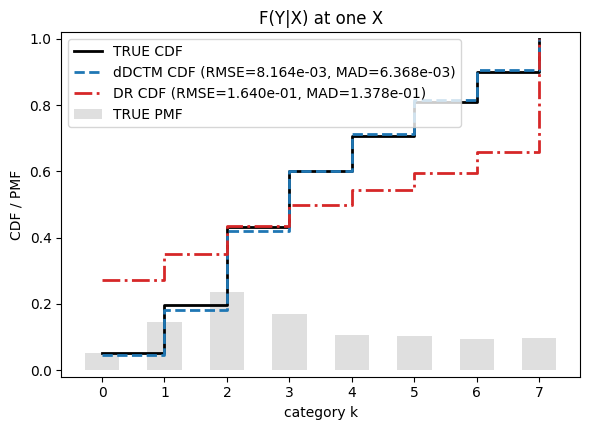}
\includegraphics[scale=0.48]{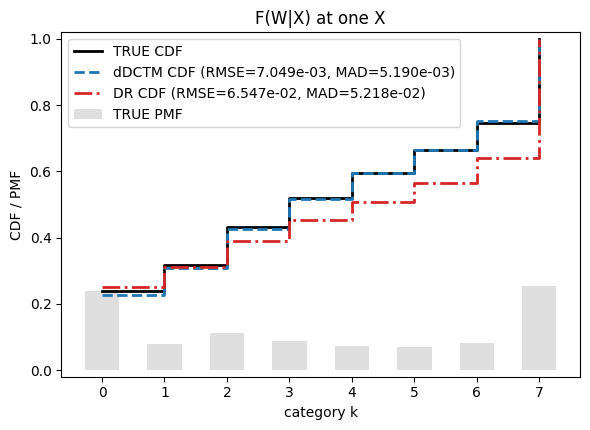}
\caption{Complex discrete ordinal setting: estimated vs.\ true step CDFs for $F(Y_d\mid X=x)$ and $F(W_d\mid X=x)$ at a randomly selected covariate value $X=x$.}
\label{fig:discrete_cdf_compare}
\end{figure}

Next, we examine $\widehat\rho_C$ for DR and dDCTM across $\omega\in[0,1]$, as shown in
Figure~\ref{fig:discrete_omega_all}. We use 11 equally spaced grid points in $[0,1]$. For each $\omega$, the
``true'' slope $\rho_C^{\mathrm{True}}$ is obtained via an auxiliary large-sample Monte Carlo approximation that
is separate from the main simulation loop. Specifically, for each category $k$ we know its cut points and
interval boundaries $(c_k,c_{k+1}]$, so
\[
\begin{aligned}
F_{Y_d\mid X}^{-}(k\mid x)
&=\mathbb{P}(Y_d<k\mid X=x)
=\mathbb{P}(Y\le c_k\mid X=x),\\
F_{Y_d\mid X}(k\mid x)
&=\mathbb{P}(Y_d\le k\mid X=x)
=\mathbb{P}(Y\le c_{k+1}\mid X=x).
\end{aligned}
\]
Using the DGP in \eqref{eq:dgp_Y}--\eqref{eq:dgp_W}, one can compute these probabilities (and similarly for $W_d$)
and then apply the rank definition at the given $\omega$ to compute the regression slope, yielding the Monte
Carlo approximation of $\rho_C^{\mathrm{True}}$.

\begin{figure}[t]
\centering
\includegraphics[scale=0.5]{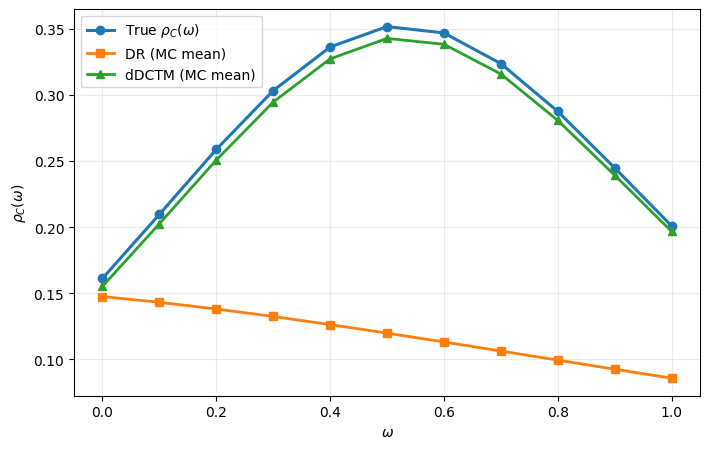}
\caption{Complex discrete ordinal setting: estimated $\rho_C(\omega)$ by traditional CRRR (DR) and our method (dDCTM) versus the truth.}
\label{fig:discrete_omega_all}
\end{figure}

Figure~\ref{fig:discrete_omega_all} shows that $\rho_C$ indeed varies with $\omega$, and its sampling behavior
also changes accordingly. Moreover, DR exhibits substantial bias relative to the truth. This is driven by
systematic underestimation in DR-based conditional CDFs; as $\omega$ increases, rank bias becomes more
pronounced. In contrast, dDCTM remains close to the truth for all $\omega$, illustrating its robustness to the
rank definition in this complex discrete setting.

Finally, Table~\ref{tab:omega-crrr-slope-table} reports detailed Monte Carlo results for the three common rank
definitions $\omega\in\{0.0,0.5,1.0\}$, based on 30 repetitions with $n=10^5$ per repetition.

\begin{table}[t]
\centering
\caption{Complex discrete ordinal setting: CRRR slope $\rho_C$ estimation for $\omega\in\{0.0,0.5,1.0\}$.}
\label{tab:omega-crrr-slope-table}
\begingroup
\setlength{\tabcolsep}{6pt}
\begin{tabular}{lcccccc}
\toprule
 & \textbf{Method} & \textbf{$\rho_C^{\mathrm{True}}$} & \textbf{Mean} & \textbf{SD} & \textbf{MAD} & \textbf{RMSE}\\
\midrule
\multirow{2}{*}{$\omega=0.0$}  & DR    & 0.15556 & 0.14742 & 0.00312 & 0.00814 & 0.00870\\
                              & dDCTM & 0.15556 & 0.15517 & 0.00324 & 0.00039 & 0.00321\\
\midrule
\multirow{2}{*}{$\omega=0.5$}  & DR    & 0.35002 & 0.11976 & 0.00336 & 0.23026 & 0.23028\\
                              & dDCTM & 0.35002 & 0.35066 & 0.00360 & 0.00169 & 0.00398\\
\midrule
\multirow{2}{*}{$\omega=1.0$}  & DR    & 0.18949 & 0.08571 & 0.00351 & 0.10378 & 0.10384\\
                              & dDCTM & 0.18949 & 0.19670 & 0.00294 & 0.00722 & 0.00777\\
\bottomrule
\end{tabular}
\endgroup

\vspace{0.5ex}
\footnotesize Note: 30 repetitions; sample size per repetition is $n=10^5$.
\end{table}

Table~\ref{tab:omega-crrr-slope-table} confirms that, under all three rank definitions, the substantial bias in
DR-based conditional CDF estimation propagates into $\widehat\rho_C$, producing much larger MAD and RMSE than
dDCTM. In contrast, dDCTM maintains good accuracy across $\omega\in\{0.0,0.5,1.0\}$. Therefore, in this complex
discrete ordinal DGP, dDCTM clearly outperforms DR both in conditional distribution fitting and in $\rho_C$
estimation.


\section{Empirical Studies}\label{sec:empirical}

\subsection{Intergenerational Income Mobility in the PSID}\label{subsec:psid_income}
In this section, we study intergenerational income mobility using the PSID-SHELF dataset from the U.S.
Panel Study of Income Dynamics (PSID).

\subsubsection{The PSID-SHELF Dataset}\label{subsubsec:psid_shelf}
The Panel Study of Income Dynamics (PSID)\footnote{PSID website: \url{https://psidonline.isr.umich.edu/GettingStarted.aspx}}
is the longest-running nationally representative household panel survey in the world. It was originally
launched to study the dynamics of income and poverty and to evaluate President Lyndon Johnson's War on Poverty.
Over the past five decades, the scope of PSID has expanded substantially and now covers a wide range of topics,
including health, wealth, consumption, charitable giving, child development, and the transition to adulthood.
The PSID has been collecting data for more than 55 years, has surveyed over 85,000 individuals, and traces
family lineages for as many as seven generations. Its scientific value has long surpassed the original focus on
poverty dynamics and has attracted researchers across disciplines \citep{solon1992intergenerational,corcoran1992association,lee2009trends,mazumder2015estimating,pfeffer2018generations,fisher2023intergenerational}.

Despite PSID's efforts to maintain longitudinal consistency, many variables change across waves. The intricate
coding (e.g., ER30000) and shifting cross-year indices make it time-consuming to construct an intergenerational
mobility dataset directly from the raw PSID. The Panel Study of Income Dynamics-Social, Health, and Economic
Longitudinal File (PSID-SHELF)\footnote{Daumler, Davis, Esther Friedman, and Fabian T. Pfeffer. 2025.
\textit{PSID-SHELF User Guide and Codebook, 1968--2021, Beta Release}. PSID-SHELF Data Documentation 2025-01.
Ann Arbor, MI: Survey Research Center, Institute for Social Research, University of Michigan.
DOI:10.7302/25205.}
provides a unified and user-friendly longitudinal file.

A key advantage of PSID-SHELF is that it offers a longitudinal file containing the complete multi-generational
PSID sample. The current version includes 42 survey waves spanning 1968 -- 2021 and integrates all variables for
each sample member across the entire period. All individuals ever interviewed in the PSID core study are
included. The dataset covers more than 8,000 families and contains over 900,000 observations from about 53,000
individuals. In addition, PSID-SHELF replaces the raw PSID variable codes with meaningful names, greatly
reducing data-processing costs.   A second advantage is that PSID-SHELF provides a unified set of key indicators
covering broad substantive themes: (i) social characteristics (e.g., demographics, family type, education, race
and ethnicity); (ii) health characteristics (e.g., chronic conditions, COVID-19, dementia, disability); and (iii)
economic characteristics (e.g., earnings, family income, occupation, wealth).

In this empirical analysis, we use the individual-level longitudinal PSID-SHELF data to study intergenerational
income mobility. We first obtain the FIMS genealogical identifier (GID) file from PSID to construct a
father-child ID mapping, and then link it to PSID-SHELF to build a one-to-one father-child matched dataset.
To ensure comparability between generations, we use total labor income measured at age 30 for both the child and
the father. 

We therefore restrict the sample to individuals born between 1939 and 1990, avoiding cases where
income is outside the PSID-SHELF coverage. We drop pairs with missing income for either generation, yielding a
final matched sample size of 2,360. The covariates include: whether the father has higher education; whether the
family has at least two children; total household expenditure; and the father's age at the child's birth.
Missing covariate values are imputed using the median. Because child income, father income, and household
expenditure are highly right-skewed, we apply the transformation $\log(1+x)$ to these three variables for
robustness.

Table~\ref{tab:pisd-desc-table} reports descriptive statistics. By gender, sons' average income is higher than daughters', while fathers' average income is not significantly different between the son and daughter samples. The proportion of families with higher education is approximately 53\%, slightly lower in the father-daughter sample than in the father-son sample. 75\% of families have at least two children. Overall, the covariate differences between the father-son and father-daughter samples are not significant.
\begin{table}[htbp]
\centering
\caption{Descriptive statistics}
\label{tab:pisd-desc-table}
\setlength{\tabcolsep}{6pt}
\begin{threeparttable}
\begin{tabular}{lcccccc}
\toprule
& \multicolumn{2}{c}{Father-Child} & \multicolumn{2}{c}{Father-Son}
& \multicolumn{2}{c}{Father-Daughter} \\
\cmidrule(lr){2-3}\cmidrule(lr){4-5}\cmidrule(lr){6-7}
& Mean & SD & Mean & SD & Mean & SD \\
\midrule
log (Child Income)   & 8.90   & 3.24 & 9.58 & 2.59 & 8.30  & 3.63 \\
log (Father Income)  & 5.36   & 4.09 & 5.42 & 4.09 & 5.30  & 4.10 \\
Higher Education     & 0.53 & 0.50 & 0.55 & 0.50  & 0.51 & 0.50 \\
At least 2 children  & 0.75 & 0.43 & 0.75 & 0.44  & 0.76 & 0.43 \\
log (Expand)         & 9.94  & 0.07 & 9.94 & 0.06 & 9.94  & 0.08 \\
Age at birth         & 23.79 & 4.82 & 24.00 & 4.85 & 23.60 & 4.78 \\
\bottomrule
\end{tabular}

\begin{tablenotes}[flushleft]
\footnotesize
\item \textit{Notes:} The sample size is 1,107 for father-son pairs and 1,253 for father-daughter pairs.
\end{tablenotes}
\end{threeparttable}
\end{table}

Figures~\ref{fig:income_trans_metric_fc}--\ref{fig:income_trans_metric_fc_daughter} display the heatmaps of the marginal-rank transition matrices for the Father-Child, Father-Son, and Father-Daughter samples, respectively. In the left-hand matrix, the entry in row $i$ and column $j$ represents the conditional probability that the child falls into the $j$th income decile given that the father is in the $i$th income decile. Under complete intergenerational independence, each cell probability should be close to 10\%, while the right-hand matrix reports the corresponding percentage deviations from this benchmark. Overall, all three matrices depart markedly from the benchmark of perfect mobility, indicating that the father's income position exerts a substantial influence on the child's income position. The most salient common feature is top-end persistence: when the father belongs to the 10th decile, the probability that the child also remains in the 10th decile reaches 22.0\%, 23.4\%, and 24.6\% in the father-child, father-son, and father-daughter samples, respectively, all of which are substantially above the benchmark value of 10\%. This pattern indicates strong intergenerational persistence among high-income families. In addition, the father-son matrix exhibits a more pronounced concentration around adjacent deciles, suggesting that income transmission between fathers and sons is characterized by a more continuous inheritance of relative rank positions. By contrast, the father-daughter matrix is relatively more diffuse in the middle deciles, but displays stronger persistence at the top of the distribution.

\begin{figure}[htbp]
  \centering
  \includegraphics[scale=0.4]{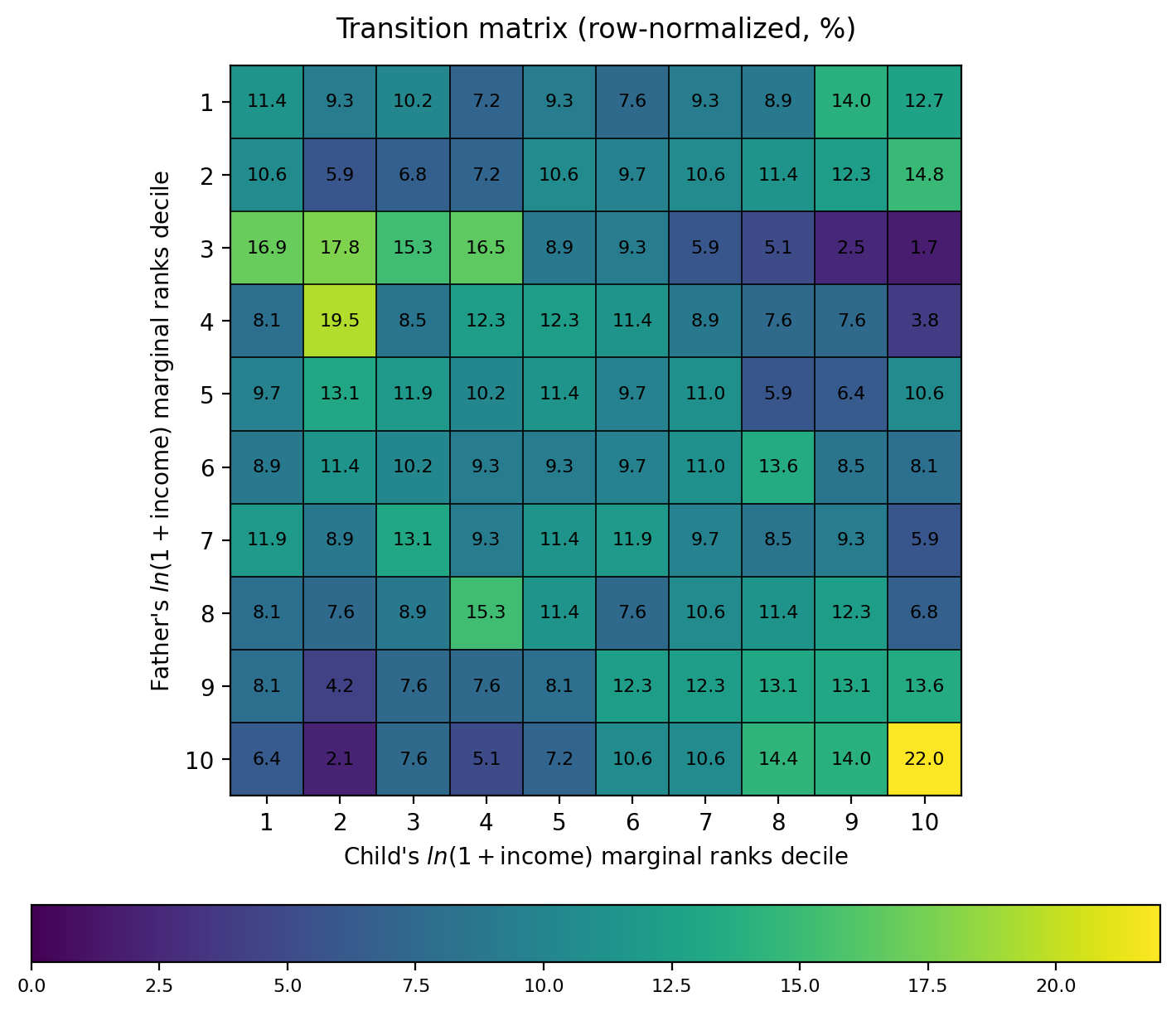}
  \includegraphics[scale=0.4]{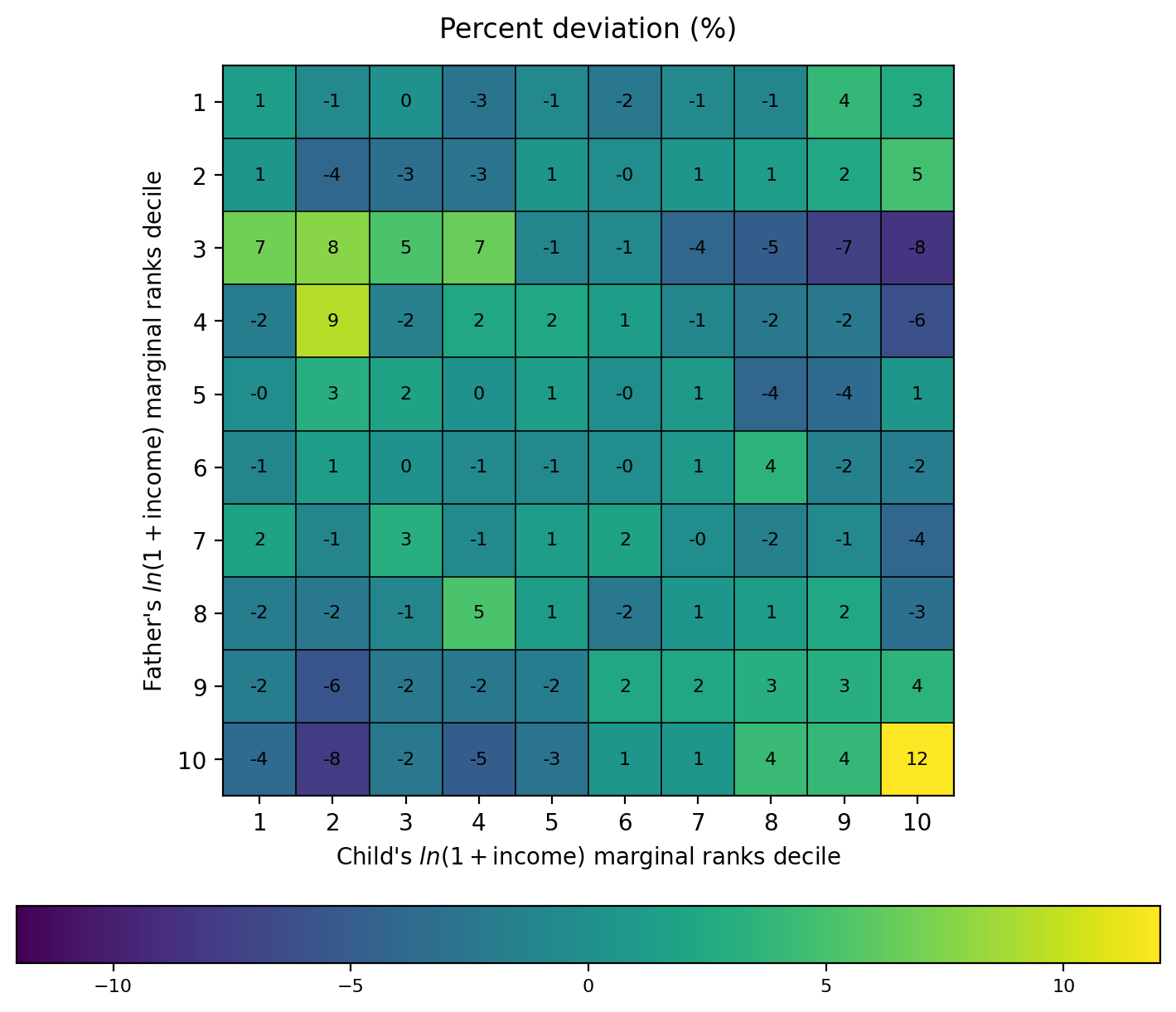}
  \caption{Transition matrices: Father-Child}
  \label{fig:income_trans_metric_fc}
\end{figure}

\begin{figure}[htbp]
  \centering
  \includegraphics[scale=0.4]{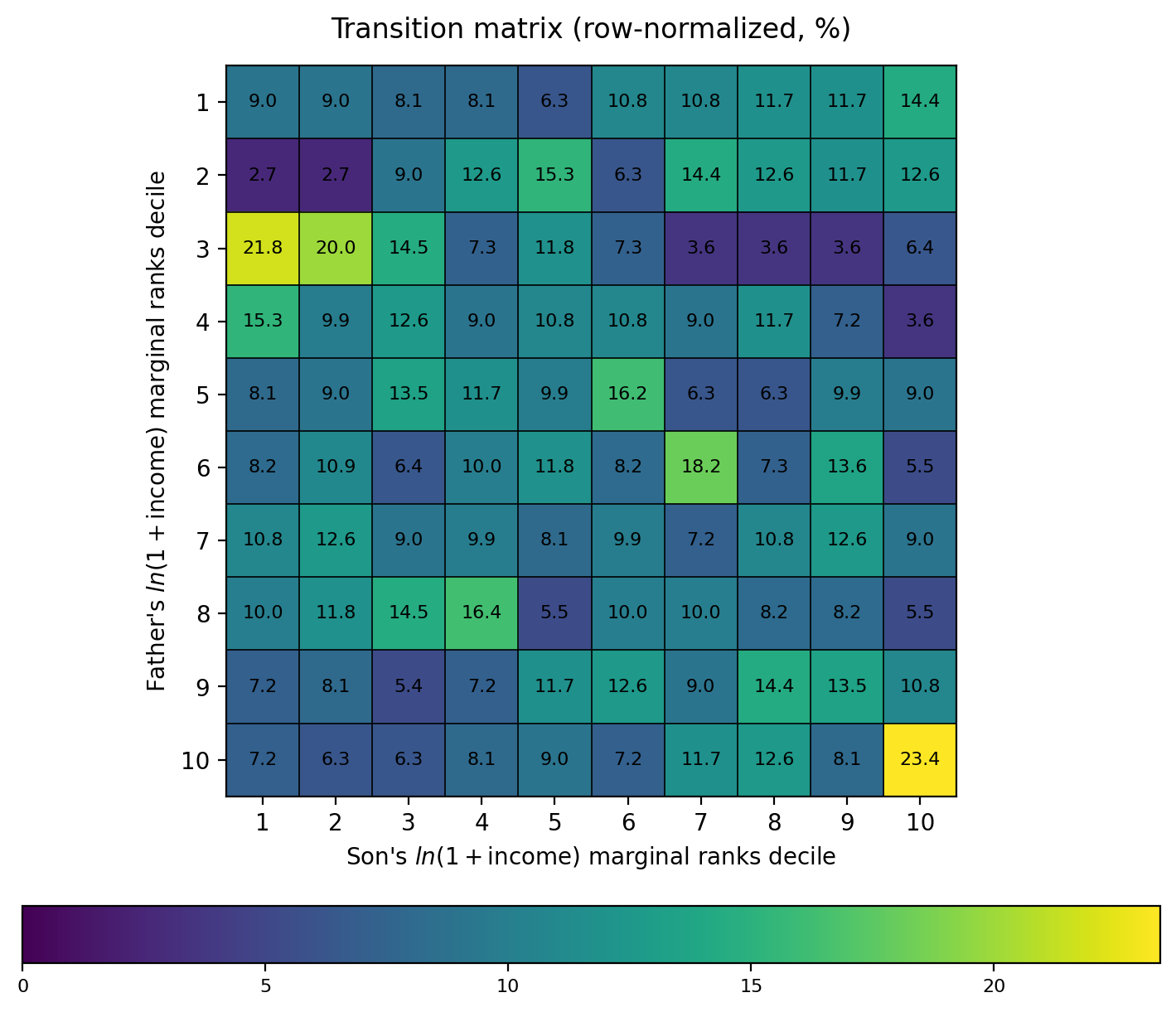}
  \includegraphics[scale=0.4]{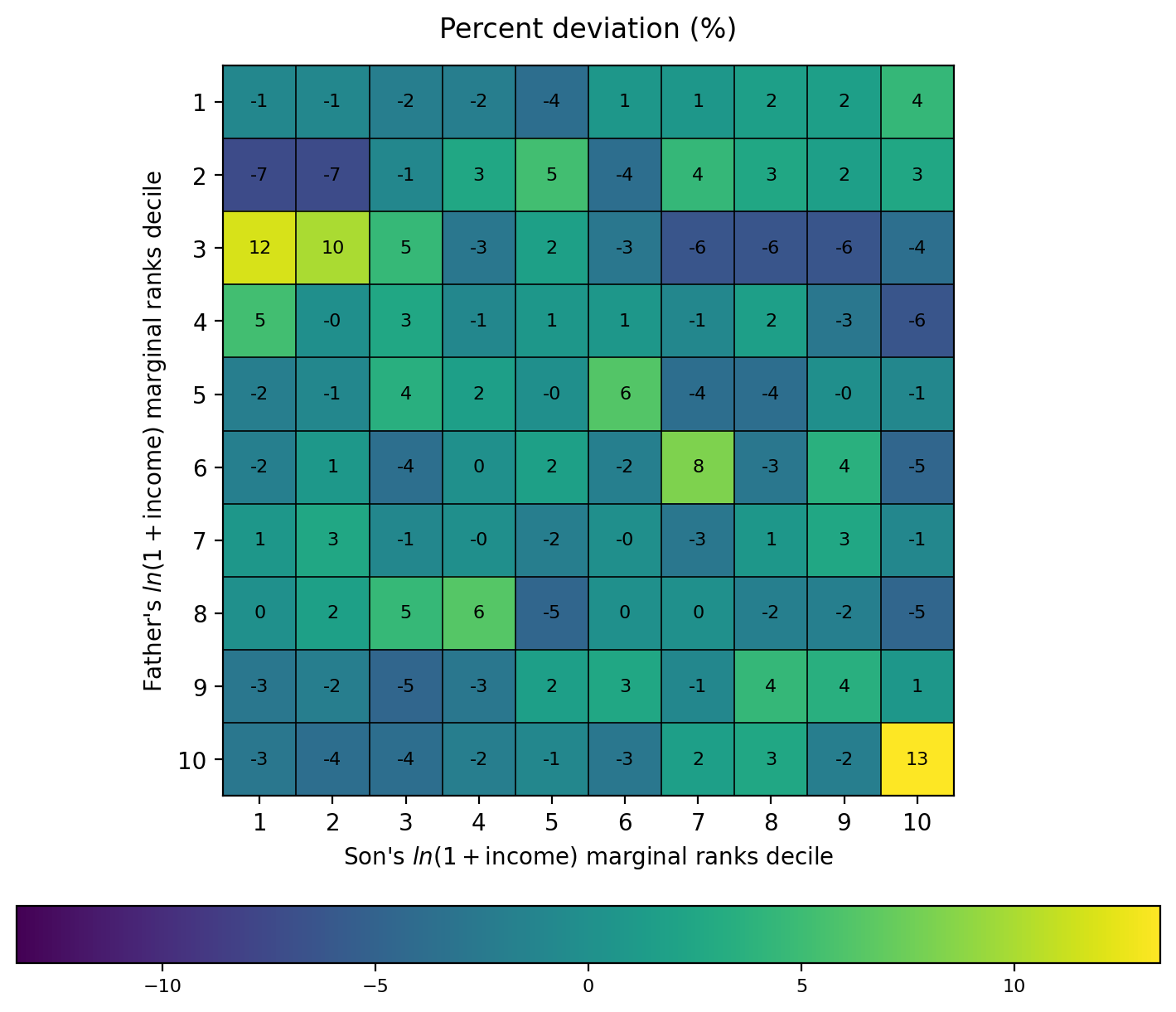}
  \caption{Transition matrices: Father-Son}
  \label{fig:income_trans_metric_fc_son}
\end{figure}

\begin{figure}[htbp]
  \centering
  \includegraphics[scale=0.4]{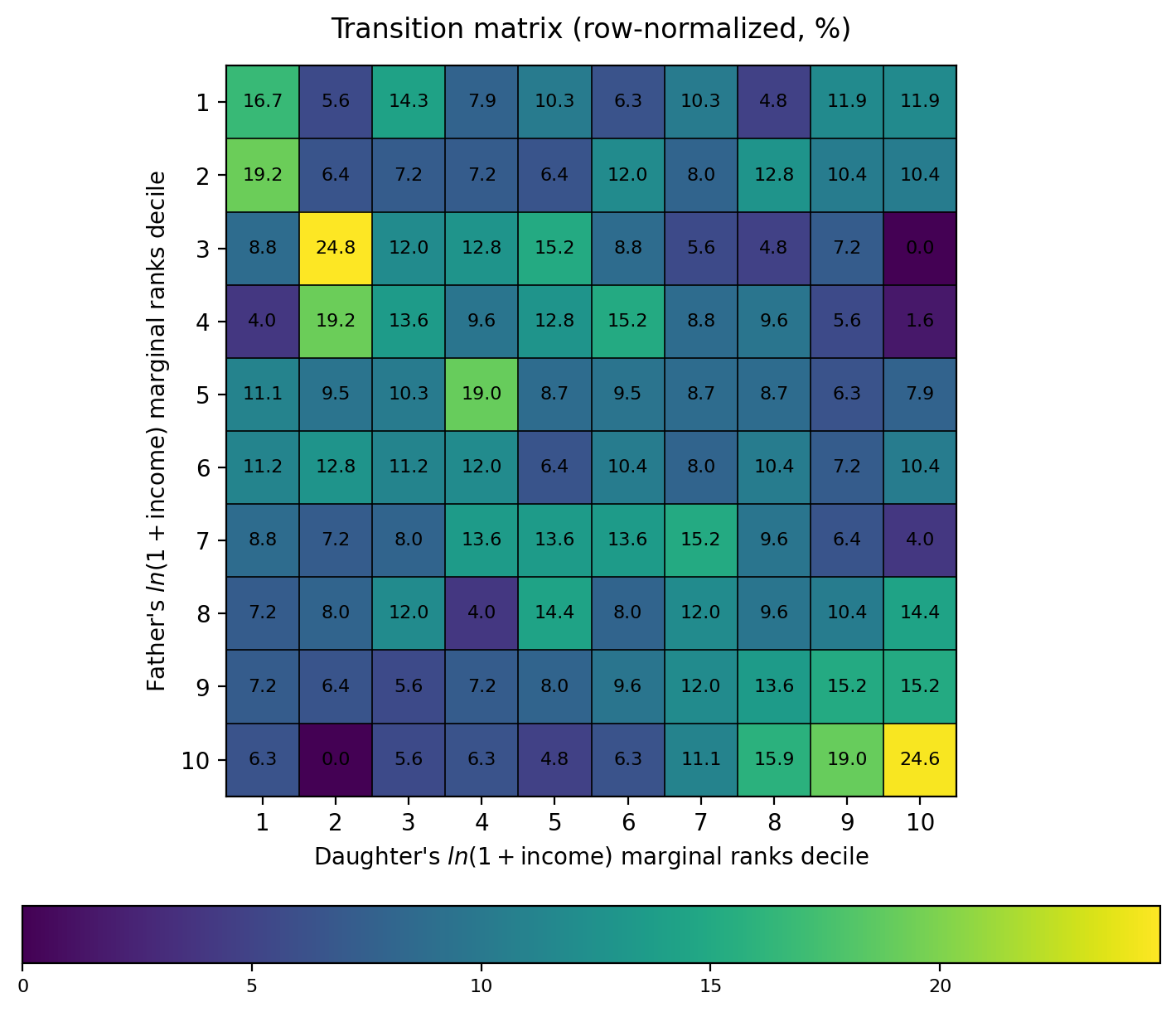}
  \includegraphics[scale=0.4]{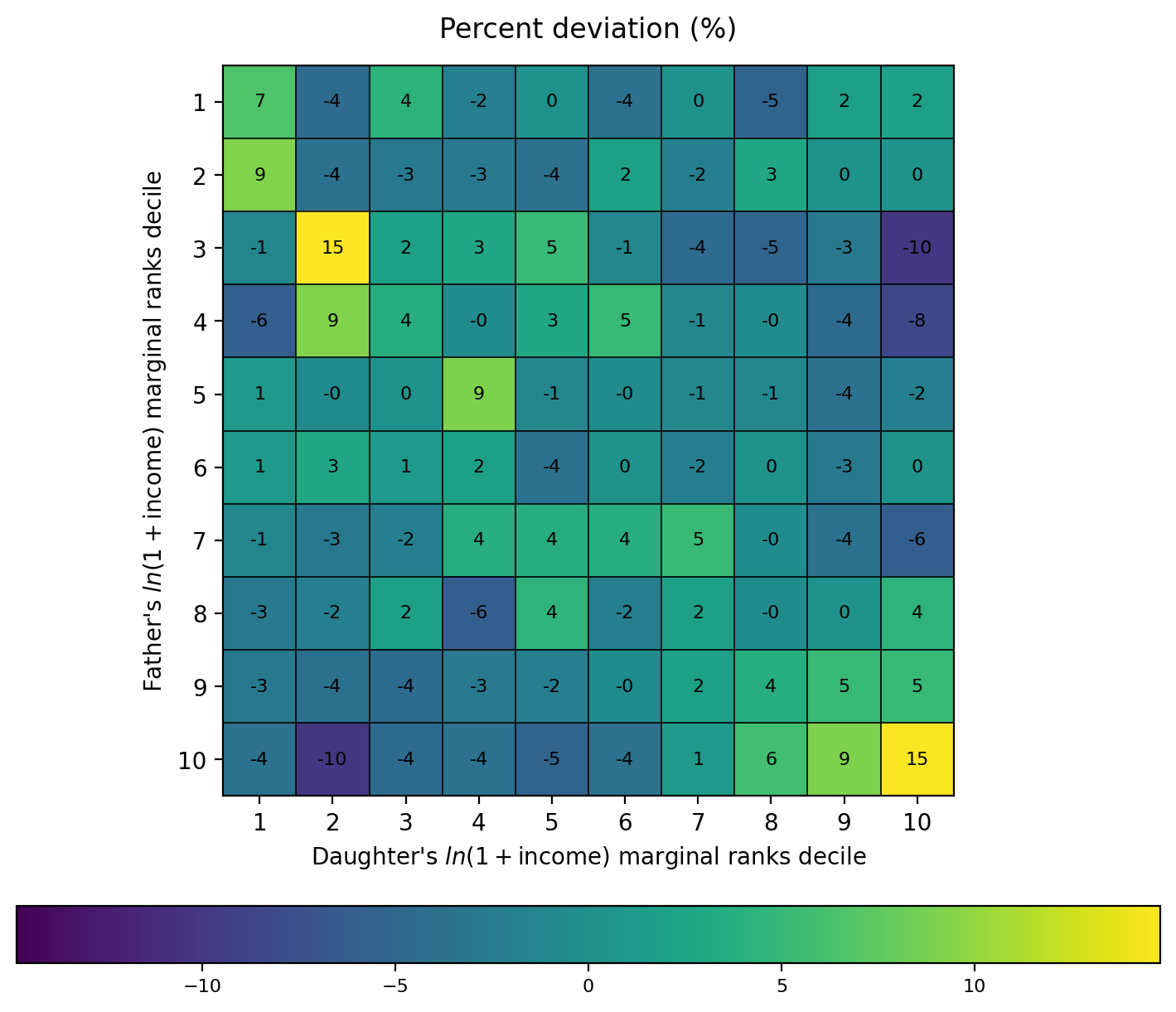}
  \caption{Transition matrices: Father-Daughter}
  \label{fig:income_trans_metric_fc_daughter}
\end{figure}

\subsubsection{Income mobility analysis}\label{subsubsec:psid_income_analysis}
Income is continuous. We estimate conditional ranks using DCTM and then compute the OLS-slope estimate
$\widehat\rho_C$. A larger $\widehat\rho_C$ indicates stronger intergenerational persistence and hence lower
mobility.

Table~\ref{tab:pisd-income-mobility-table} reports RRR, CRRR, and subgroup analyses by father characteristics
for father-child pairs. Standard errors and confidence intervals are obtained via Algorithm~2 using the
empirical bootstrap with 200 resamples. The results show a significant positive persistence effect in the
father-child income relationship both with and without covariate adjustment. The RRR estimate is 0.180, while
CRRR decreases to 0.121 after controlling for covariates. This suggests that the within-group average
intergenerational persistence is about 0.12, while the remaining persistence can be attributed to between-group
differences in covariates ($\rho_{RRR} - \rho_{C}\approx 0.059$). Subgroup results indicate higher income mobility in
families with higher paternal education, paternal age at birth above 26, and larger number of children.

\begin{table}[htbp]
\centering
\caption{Father-Child intergenerational income mobility}
\label{tab:pisd-income-mobility-table}
\setlength{\tabcolsep}{6pt}
\begin{threeparttable}
\begin{tabular}{lccc}
\toprule
& \multicolumn{3}{c}{Father-Child} \\
\cmidrule(lr){2-4}
& $\widehat\rho_C$ & SE & 95\% CI \\
\midrule
RRR   & 0.180 & 0.022 & [0.142, 0.218] \\
CRRR  & 0.121 & 0.019 & [0.081, 0.160] \\
CRRR, by Father:  &  &  & \\
\quad Higher Education     & 0.073 & 0.030 & [0.018, 0.129] \\
\quad Age $>26$ at birth   & 0.108 & 0.034 & [0.039, 0.178] \\
\quad At least 2 children  & 0.104 & 0.021 & [0.059, 0.148] \\
\bottomrule
\end{tabular}
\end{threeparttable}
\end{table}

To study gender differences, we analyze father-son and father-daughter pairs separately; see
Table~\ref{tab:pisd-income-mobility-by-sex}. The RRR estimate is 0.099 and the CRRR estimate is 0.062 for
father-son pairs, implying that within-group persistence accounts for about 62\% of total persistence. For
father-daughter pairs, the RRR estimate is 0.248 and the CRRR estimate is 0.180, implying that within-group
persistence accounts for about 73\% of total persistence. Thus, even after controlling for covariates, the
intergenerational persistence in daughters' income remains substantially higher than that in sons' income,
suggesting that daughters' income ranks are more strongly tied to family background. Subgroup analyses show that
father-daughter persistence exceeds father-son persistence across all subgroups, indicating a pronounced
gender gap: daughters' income appears more constrained by fathers' experiences and backgrounds.

\begin{table}[htbp]
\centering
\caption{Father-Son and Father-Daughter intergenerational income mobility}
\label{tab:pisd-income-mobility-by-sex}
\setlength{\tabcolsep}{6pt}
\begin{threeparttable}
\begin{tabular}{lcccccc}
\toprule
& \multicolumn{3}{c}{Father-Son} & \multicolumn{3}{c}{Father-Daughter} \\
\cmidrule(lr){2-4}\cmidrule(lr){5-7}
& $\widehat\rho_C$ & SE & 95\% CI & $\widehat\rho_C$ & SE & 95\% CI \\
\midrule
RRR   & 0.099 & 0.031 & [0.033, 0.165] & 0.248 & 0.026 & [0.188, 0.308] \\
CRRR  & 0.062 & 0.029 & [0.007, 0.117] & 0.180 & 0.028 & [0.124, 0.237] \\
CRRR, by Father: &  &  &  &  &  & \\
\quad Higher Education     & -0.027 & 0.037 & [-0.092, 0.039] & 0.179 & 0.045 & [0.101, 0.258] \\
\quad Age $>26$ at birth   & 0.036  & 0.054 & [-0.078, 0.151] & 0.171 & 0.047 & [0.059, 0.283] \\
\quad At least 2 children  & 0.047  & 0.033 & [-0.012, 0.107] & 0.161 & 0.034 & [0.093, 0.230] \\
\bottomrule
\end{tabular}
\end{threeparttable}
\end{table}

\subsection{Intergenerational Educational Mobility in the IHDS}\label{subsec:ihds_edu}
In this section, we analyze intergenerational educational mobility between fathers and children in India using
  the 2012 India Human Development Survey (IHDS).

\subsubsection{The IHDS Dataset}\label{subsubsec:ihds_data}
The IHDS is a nationally representative, multi-topic survey in India with two waves:
2004 -- 2005 (IHDS-I\footnote{Desai, Sonalde, Vanneman, Reeve, and National Council of Applied Economic Research,
New Delhi. India Human Development Survey (IHDS), 2005. Inter-university Consortium for Political and Social
Research [distributor], 2018-08-08. \url{https://doi.org/10.3886/ICPSR22626.v12}}) and
2011 -- 2012 (IHDS-II\footnote{Desai, Sonalde, and Vanneman, Reeve. India Human Development Survey-II (IHDS-II),
2011--12. Inter-university Consortium for Political and Social Research [distributor], 2018-08-08.
\url{https://doi.org/10.3886/ICPSR36151.v6}}). Our empirical analysis uses IHDS-II.

IHDS-I interviewed 41,554 households in 2004 -- 2005, including 26,734 rural and 14,820 urban households. In
2011 -- 2012, IHDS-II re-interviewed about 83\% of the original IHDS-I households (some households were replaced
when interviewers could not locate them), yielding a final sample of 42,152 households (27,579 rural and 14,573
urban). These households cover 33 states and union territories, 384 districts, 1,420 villages, and 1,042 urban
neighborhoods. The survey includes modules on health, education, employment, economic status, marriage,
fertility, caste, and household assets, among others, and can be linked to IHDS-I, providing unique value for
studying patterns of social and economic change in India \citep{asher2017estimating,chetverikov2023inference,asher2024intergenerational,emran2025gender,ahsan2025ranks}.

IHDS-II provides questionnaires organized at the individual, household, and eligible-women levels. To construct
father-child linkages, we follow the data-processing procedure in \citet{asher2024intergenerational}. We merge the three data sources with $\mathrm{householdID}$ and $\mathrm{personID}$ as the primary keys, and infer father-child relationships.
Our outcome of interest is the \emph{years of education}. We drop observations with missing outcome value for either fathers or
children, and then aggregate the original education variable into the most widely used education definition in
the India literature, yielding 7 ordered discrete levels (Edu Level), from illiterate to higher education,
coded as 1 -- 7.

To explain structural heterogeneity in intergenerational educational association, we include household-level
covariates: social group (Group coded as Scheduled Castes (SC) / Scheduled Tribes (ST) = 1; Muslim = 2; Others =
3), household size (Nperson), per-capita monthly household consumption expenditure (Mpce), and urban residence (Urban).
Mpce is highly right-skewed, so we use $\log(\text{Mpce}+1)$. Table~\ref{tab:ihds-desc-table} reports descriptive
statistics.

\begin{table}[htbp]
\centering
\caption{Descriptive statistics}
\label{tab:ihds-desc-table}
\setlength{\tabcolsep}{6pt}
\begin{threeparttable}
\begin{tabular}{lcccccc}
\toprule
& \multicolumn{2}{c}{Father-Child} & \multicolumn{2}{c}{Father-Son}
& \multicolumn{2}{c}{Father-Daughter} \\
\cmidrule(lr){2-3}\cmidrule(lr){4-5}\cmidrule(lr){6-7}
& Mean & SD & Mean & SD & Mean & SD \\
\midrule
Child Edu Level    & 3.69 & 1.88 & 4.30 & 1.66 & 3.32 & 1.92 \\
Father Edu Level   & 2.72 & 1.58 & 3.19 & 1.93 & 2.43 & 1.25 \\
Group              & 2.28 & 0.89 & 2.30 & 0.88 & 2.27 & 0.89 \\
Nperson            & 5.88 & 2.50 & 6.29 & 2.63 & 5.64 & 2.38 \\
log (Mpce)         & 9.89 & 0.65 & 9.90 & 0.65 & 9.89 & 0.65 \\
Urban              & 0.36 & 0.48 & 0.37 & 0.48 & 0.36 & 0.48 \\
\bottomrule
\end{tabular}

\begin{tablenotes}[flushleft]
\footnotesize
\item \textit{Notes:} The sample size is 32,099 for father-son pairs and 52,984 for father-daughter pairs.
\end{tablenotes}
\end{threeparttable}
\end{table}

Figures~\ref{fig:trans_metric_fc}--\ref{fig:trans_metric_fc_daughter} display the heatmaps of the transition matrices for the Father-Child, Father-Son, and Father-Daughter samples, respectively. The left panel shows the transition probability matrix, while the right panel reports the percentage deviation from the child's marginal distribution. The results show that the transition matrices in all three samples depart markedly from the independence benchmark, indicating that the father's educational status exerts a substantial influence on the child's educational outcome. Overall, the probabilities are concentrated along the main diagonal and its neighboring cells, suggesting pronounced intergenerational persistence in educational attainment. Top-end persistence is particularly strong: when the father belongs to the highest education grade, the probability that the child also remains in the highest grade is \(36.0\%\), \(39.8\%\), and \(29.0\%\) in the father-child, father-son, and father-daughter samples, respectively. This pattern indicates strong intergenerational persistence among highly educated families, with the effect being most pronounced in the father-son sample. These conclusions are consistent with the RRR coefficient results presented later.

\begin{figure}[htbp]
  \centering
  \includegraphics[scale=0.4]{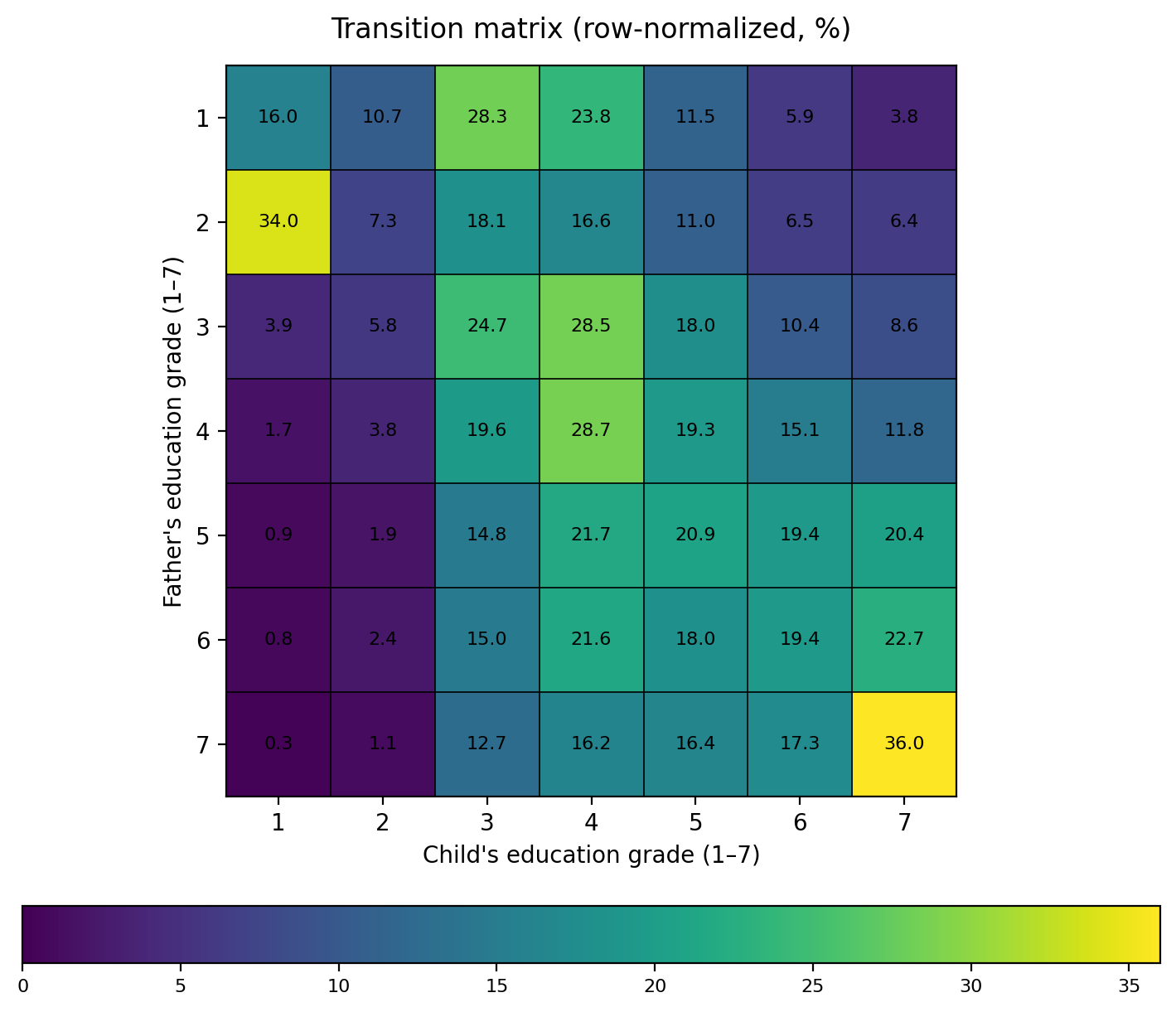}
  \includegraphics[scale=0.4]{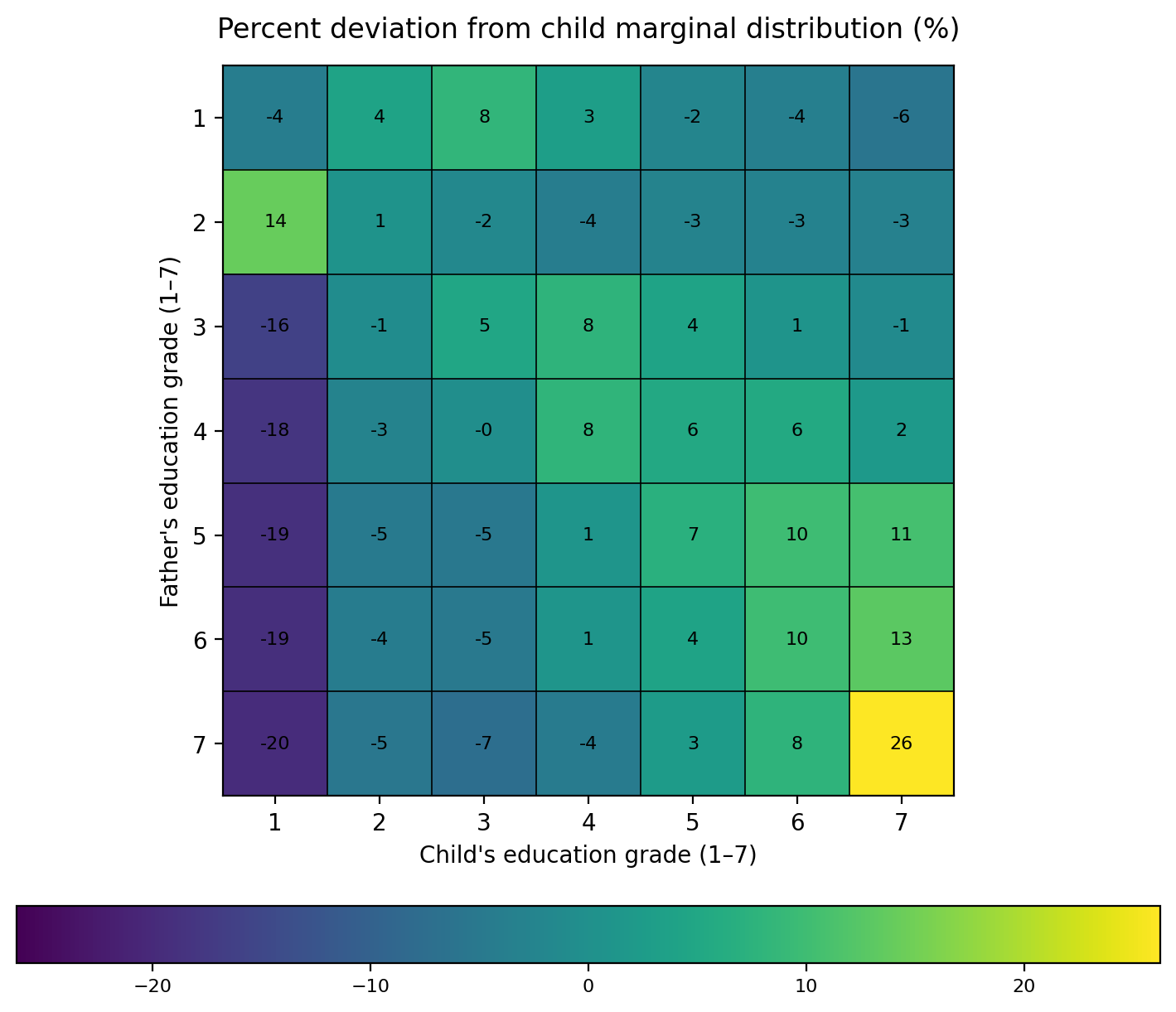}
  \caption{Transition matrices: Father-Child}
  \label{fig:trans_metric_fc}
\end{figure}

\begin{figure}[t]
  \centering
  \includegraphics[scale=0.4]{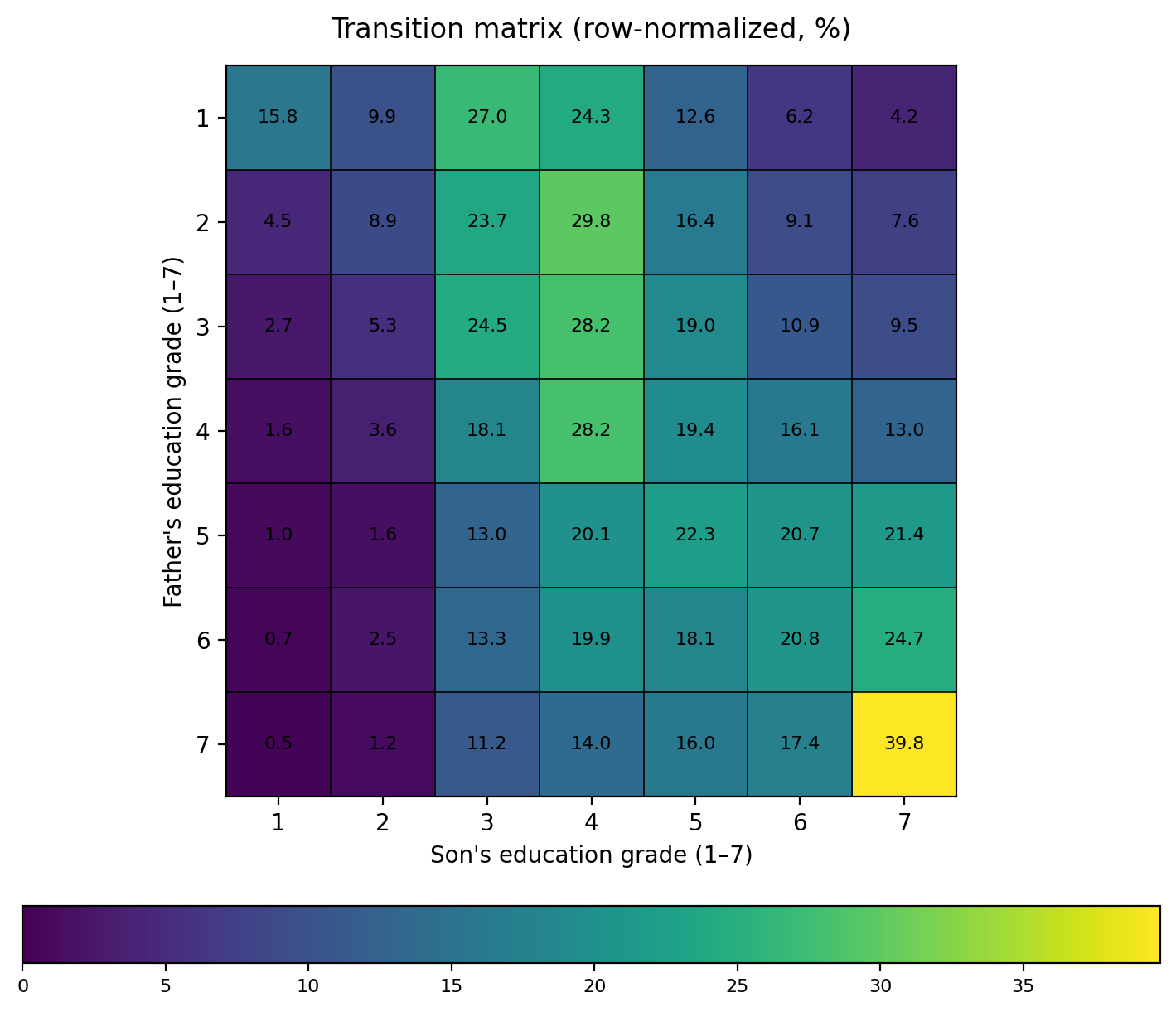}
  \includegraphics[scale=0.4]{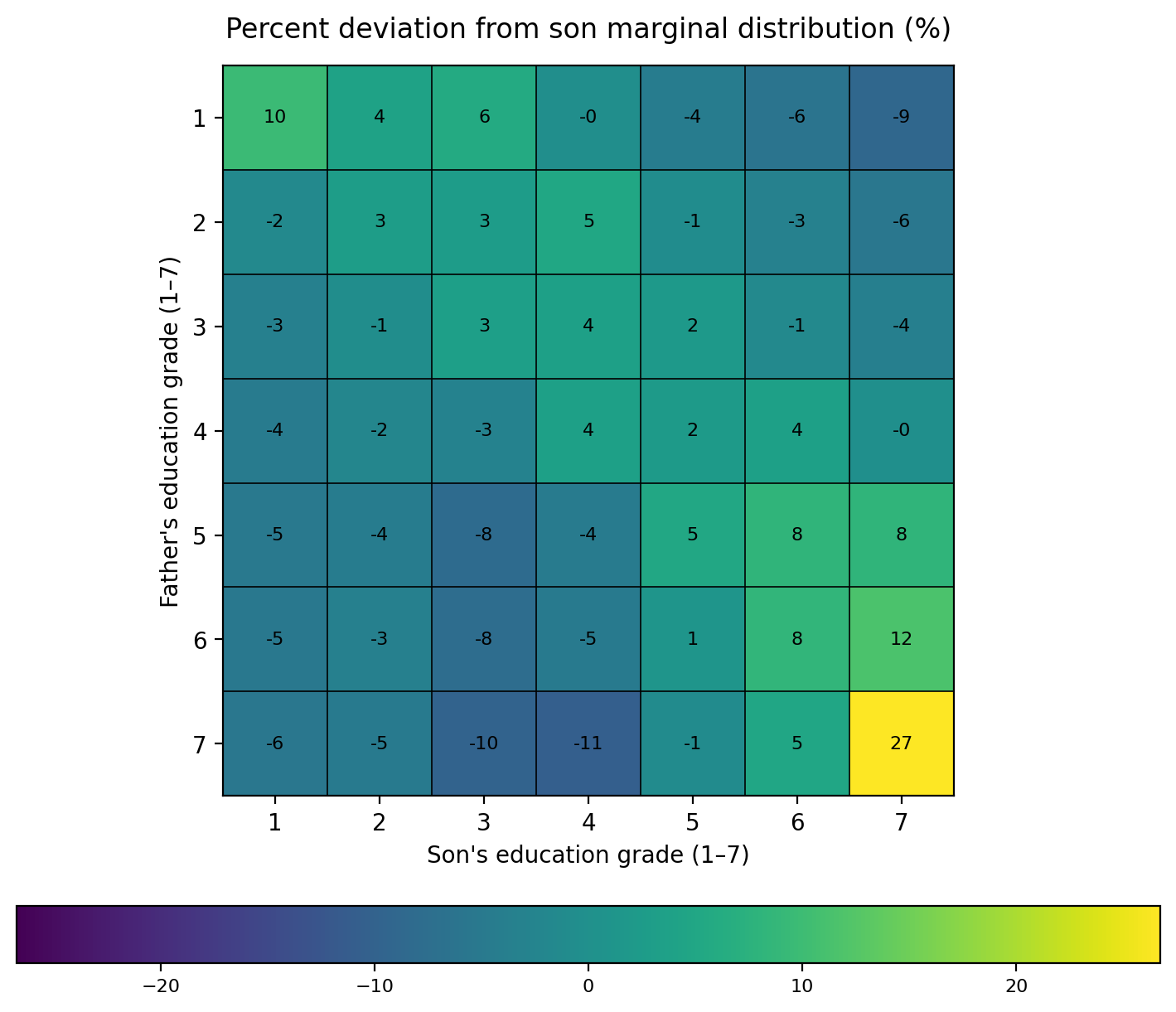}
  \caption{Transition matrices: Father-Son}
  \label{fig:trans_metric_fc_son}
\end{figure}

\begin{figure}[t]
  \centering
  \includegraphics[scale=0.4]{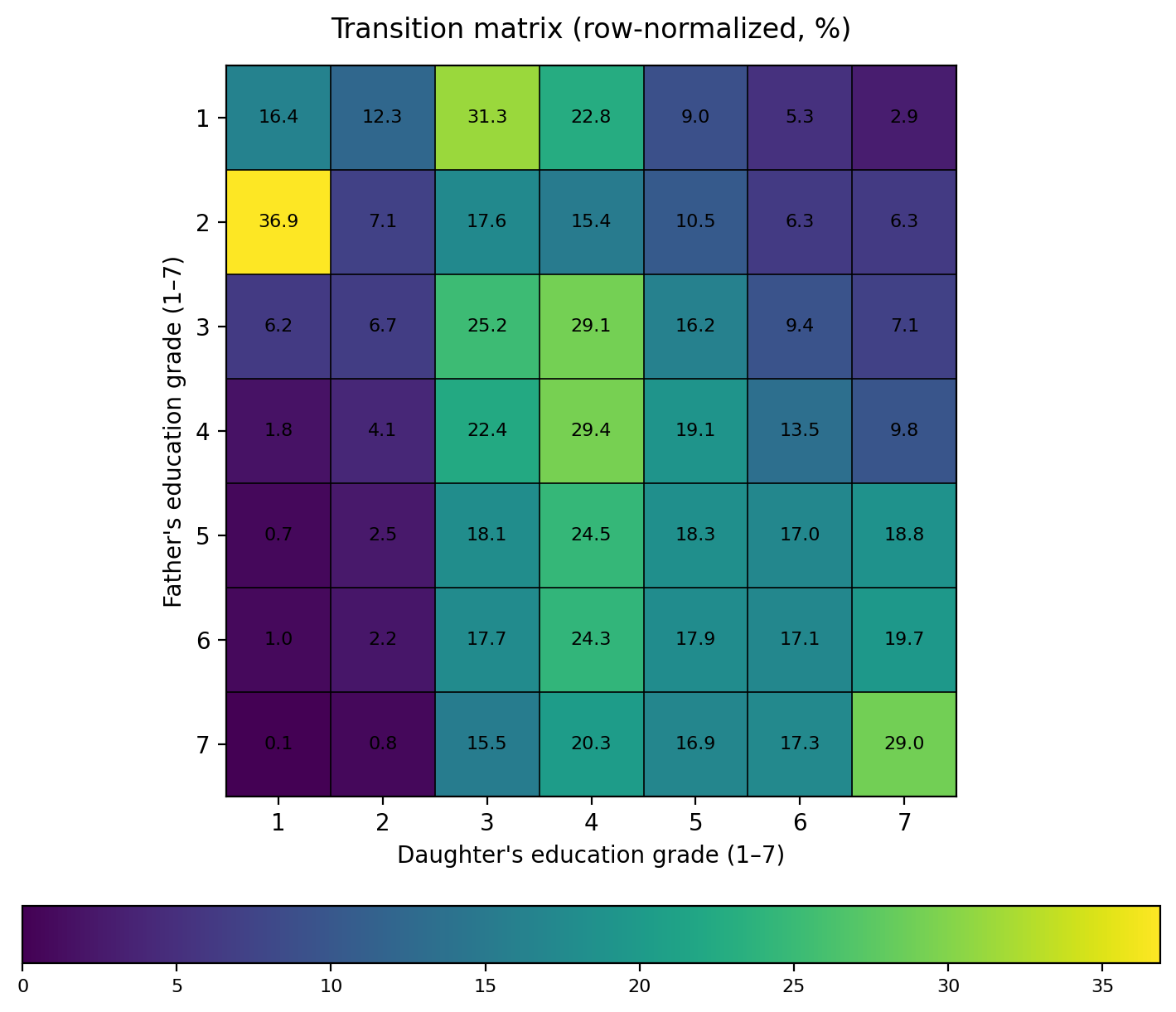}
  \includegraphics[scale=0.4]{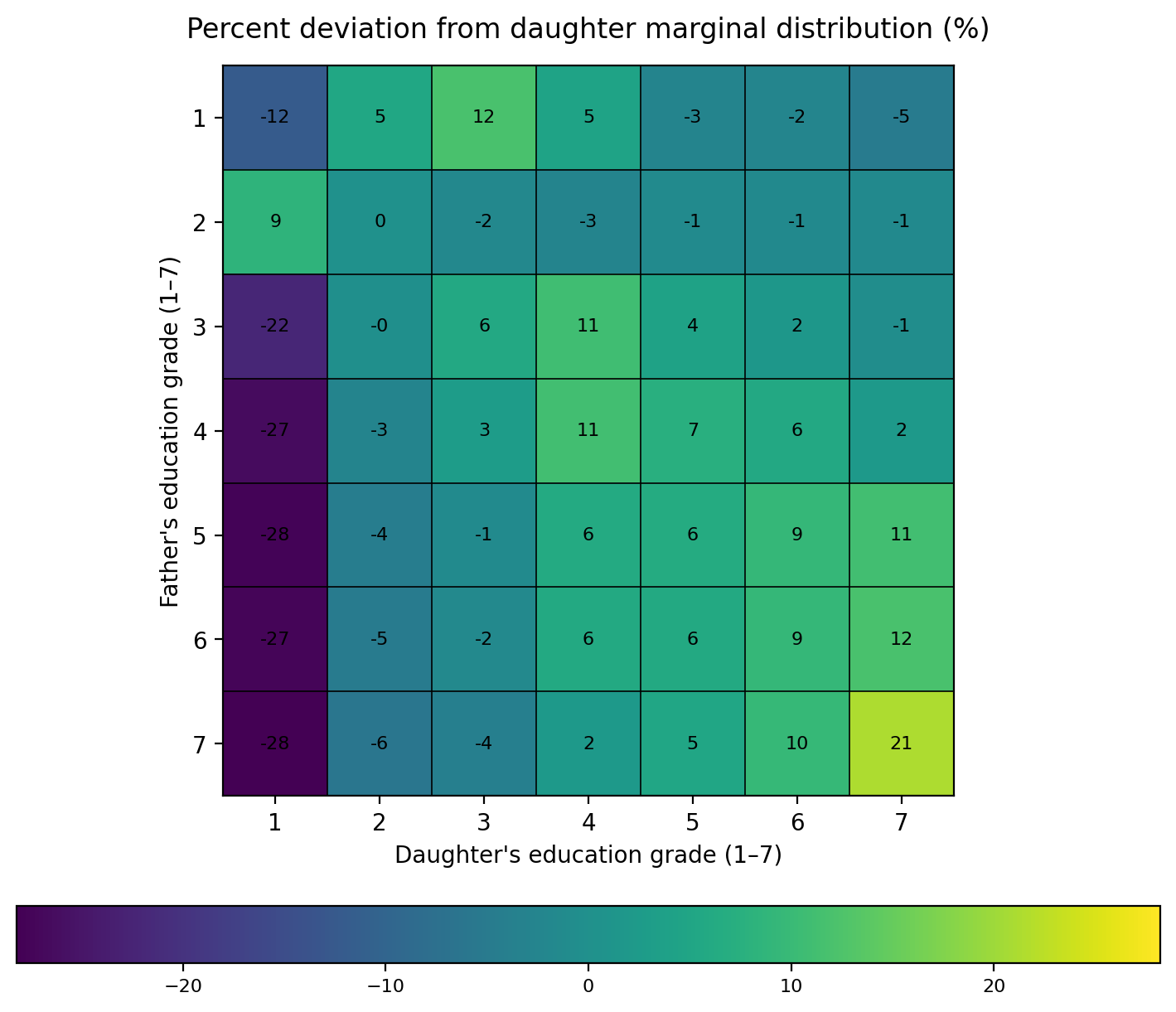}
  \caption{Transition matrices: Father-Daughter}
  \label{fig:trans_metric_fc_daughter}
\end{figure}

\subsubsection{Educational mobility analysis}\label{subsubsec:ihds_edu_analysis}
Education level is a discrete ordinal variable. We therefore use dDCTM to estimate the conditional CDF/PMF,
following the network design described in the discrete-ordinal simulation section. We then plug the estimated
conditional CDF (including the left limit) into the $\omega$-parameterized rank definitions
\eqref{eq:rankY_omega_en} and \eqref{eq:rankW_omega_en} to compute conditional ranks, and finally compute the OLS-slope
estimate $\widehat\rho_C$. 

Figure~\ref{fig:fc_son_dau_rho_omega} plots $\widehat\rho_C(\omega)$ for father-child, father-son, and father-daughter pairs as a function of $\omega$. The figure shows that $\widehat\rho_C$ is highly sensitive to tie handling. For example, under the
smallest-rank convention ($\omega=0.0$), educational mobility appears lower for father-daughter pairs than for
father-son pairs. However, for $\omega\in[0.5,1.0]$, the conclusion reverses: father-daughter pairs exhibit
higher persistence (lower mobility) than father-son pairs, and the gap widens as $\omega$ increases.

Table~\ref{tab:edu-mobility-analysis} reports RRR, CRRR, and subgroup analyses under three rank
definitions $\omega\in\{0.0,0.5,1.0\}$. The results show significant positive persistence in education between fathers and children, both with and without covariate adjustment. From subgroup analyses under the mid-rank definition ($\omega=0.5$), persistence is stronger (mobility lower) in Muslim households and in urban residence groups, while mobility is higher in larger households. To examine gender heterogeneity, we compare father-son and father-daughter pairs. Under $\omega=0.5$, sons exhibit lower mobility in Muslim households, whereas sons in larger households and in urban residence groups show higher mobility. For daughters, mobility is higher in Muslim households and in larger households, while mobility is lower in urban residence groups. Overall, these patterns indicate pronounced gender differences in intergenerational educational mobility in India.
\begin{figure}[htbp]
  \centering
  \includegraphics[scale=0.5]{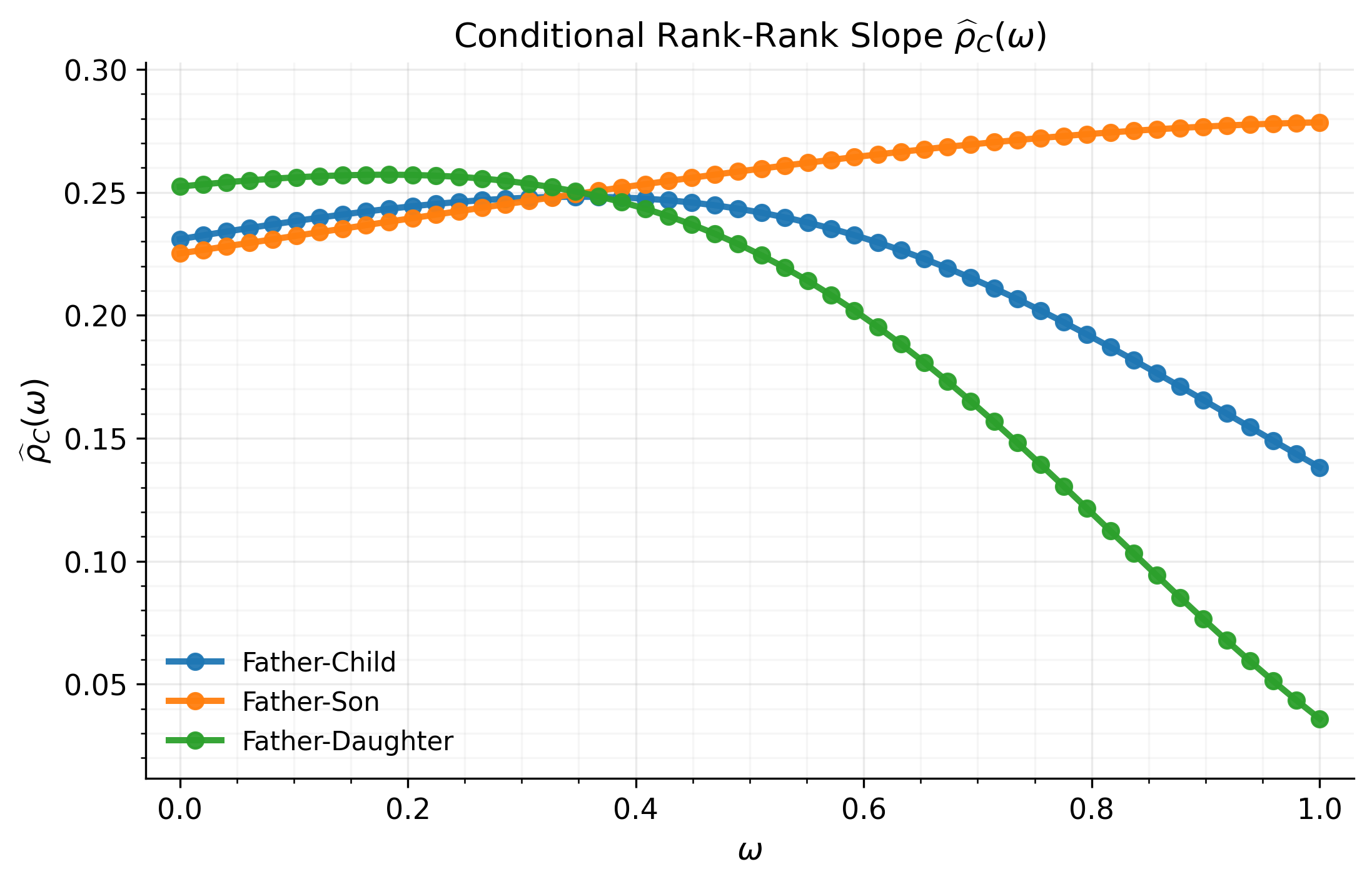}
  \caption{Estimated CRRR slope $\widehat\rho_C(\omega)$}
  \label{fig:fc_son_dau_rho_omega}
\end{figure}

\begin{table}[htbp]
\centering
\caption{Educational mobility analysis for $\omega\in\{0.0,0.5,1.0\}$}
\label{tab:edu-mobility-analysis}
\begingroup
\setlength{\tabcolsep}{6pt}
\begin{tabular}{llccc}
\toprule
& & Father-Child & Father-Son & Father--Daughter \\
\midrule
\multirow{6}{*}{$\omega=0.0$}  & RRR & 0.353 & 0.365 & 0.312 \\
& CRRR & 0.246 & 0.223 & 0.251 \\
& CRRR by:  & & & \\
& Muslim & 0.270 & 0.286 & 0.279 \\
& Nperson $\geq 4$  & 0.241 & 0.221 & 0.252 \\
& Urban & 0.213 & 0.246 & 0.178 \\
\midrule
\multirow{6}{*}{$\omega=0.5$}  & RRR & 0.366 & 0.412 & 0.335 \\
& CRRR & 0.236 & 0.259 & 0.227 \\
& CRRR by:  & & & \\
& Muslim & 0.251 & 0.347 & 0.203 \\
& Nperson $\geq 4$  & 0.231 & 0.257 & 0.224 \\
& Urban & 0.257 & 0.254 & 0.233 \\
\midrule
\multirow{6}{*}{$\omega=1.0$}  & RRR & 0.281 & 0.466 & 0.115 \\
& CRRR & 0.149 & 0.281 & 0.037 \\
& CRRR by:  & & & \\
& Muslim & 0.173 & 0.390 & 0.037 \\
& Nperson $\geq 4$  & 0.158 & 0.280 & 0.034 \\
& Urban & 0.237 & 0.257 & 0.131 \\
\bottomrule
\end{tabular}
\endgroup
\end{table}

\section{Conclusion}
In conclusion, this paper develops a more flexible and robust CRRR estimation framework by replacing DR with a DCTM-based conditional distribution learner and using cross-fitting to reduce overfitting bias. The proposed approach accommodates nonlinear, interaction-rich, heteroskedastic, and discrete ordered outcomes, delivers consistent and asymptotically normal estimation under fixed model complexity, and supports valid inference via an exchangeable bootstrap. Simulations and applications to PSID income and IHDS education data show that, while DR-based CRRR can work well in simple settings, the DCTM-based CRRR yields markedly more stable and accurate estimates of conditional ranks and $\widehat{\rho}_C$ in complex continuous and ordered-discrete scenarios.

Next, we outline several limitations of our work and suggest directions for future research. First, the current asymptotic theory is established under a fixed model-complexity regime. In more general nonparametric settings, eliminating approximation bias typically requires model complexity to grow with sample size, yet neural networks may fail to satisfy weak-convergence conditions for the input process in this case, making the classical CRRR functional inference framework harder to apply. A promising direction is to incorporate ideas from double/debiased machine learning (DDML) by constructing Neyman-orthogonal score functions to reduce the impact of DCTM estimation errors on the target parameter, thereby substantially relaxing rate requirements (e.g., \cite{chernozhukov2018double}). Second, for the discrete extension of CRRR, this paper is among the first to consider parametric rank definitions (via the tie-handling parameter $\omega$) and to study the sensitivity of $\widehat{\rho}_C$ to alternative rank definitions through simulations and empirical work. However, due to the discontinuity and non-differentiability induced by discrete ranks, the associated inference theory is considerably more involved; a formal discrete-asymptotic theory for CRRR is left for future research.

From a practical perspective, DCTM is a neural-network-based algorithm and typically requires sufficiently large sample sizes; when subgroups are defined very finely, estimation quality may deteriorate. Moreover, bootstrap-based inference is computationally expensive because the model must be retrained for each resample. Future work may explore influence-function-based resampling or lightweight retraining schemes (e.g., sharing a backbone network) to improve computational efficiency.

\section{Appendix}
\renewcommand{\thesection}{\Alph{section}.}                 
\renewcommand{\theequation}{\Alph{section}\arabic{equation}} 
\renewcommand{\thefigure}{\Alph{section}\arabic{figure}}     
\renewcommand{\thetable}{\Alph{section}\arabic{table}}       

\setcounter{equation}{0}
\setcounter{figure}{0}
\setcounter{table}{0}

\addtocontents{toc}{\protect\setcounter{tocdepth}{-1}}

\addtocontents{toc}{\protect\setcounter{tocdepth}{2}}

\begin{proof}[Proof of Lemma~\ref{lem:dctm_consistency_fun}]
By Assumption~\ref{ass:emp_process} (i) and the uniqueness of the population-risk minimizer in
Assumption~\ref{ass:spec_fun}, for each fixed $k$, the estimator $\widehat F_{R\mid X}^{(-k)}$ converges to
$F_{R\mid X}$ (almost everywhere). By compactness of $[0,1]\times\mathcal X$ and the continuity structure of
functions in $\mathcal F_R$, the convergence can be strengthened to uniform convergence.
\end{proof}

\begin{proof}[Proof of Theorem~\ref{thm:consistency_rho_fun}]
Define the full-sample means
\[
\overline{\widehat U}:=\mathbb{P}_n(\widehat U)=\frac1n\sum_{i=1}^n\widehat U_i,
\qquad
\overline{\widehat V}:=\mathbb{P}_n(\widehat V)=\frac1n\sum_{i=1}^n\widehat V_i,
\]
and write the three estimators as
\begin{align*}
\widehat\rho_C
:=&\frac{\mathbb{P}_n\!\left[(\widehat U-\overline{\widehat U})(\widehat V-\overline{\widehat V})\right]}
{\mathbb{P}_n\!\left[(\widehat V-\overline{\widehat V})^2\right]},\\
\widetilde\rho_C
:=&\frac{\mathbb{P}_n\!\left[(\widehat U-\overline{\widehat U})(\widehat V-\overline{\widehat V})\right]}
{\sqrt{\mathbb{P}_n\!\left[(\widehat U-\overline{\widehat U})^2\right]\,
\mathbb{P}_n\!\left[(\widehat V-\overline{\widehat V})^2\right]}},\\
\breve\rho_C
:=&12\,\mathbb{P}_n\!\left[(\widehat U-\tfrac12)(\widehat V-\tfrac12)\right].
\end{align*}

Let the true conditional ranks be $U:=F_{Y\mid X}(Y\mid X)$ and $V:=F_{W\mid X}(W\mid X)$, and define
$\overline U:=\mathbb{P}_n(U)$ and $\overline V:=\mathbb{P}_n(V)$. By Assumption~\ref{ass:data_fun},
$\mathbb{E}[U]=\mathbb{E}[V]=1/2$ and $\mathrm{Var}(U)=\mathrm{Var}(V)=1/12$.

For each fold $k$, define the uniform errors
\[
\Delta_{Y,k}
:=
\sup_{(y,x)}
\Big|\widehat F_{Y\mid X}^{(-k)}(y\mid x)-F_{Y\mid X}(y\mid x)\Big|,
\]
and similarly define $\Delta_{W,k}$. By Lemma~\ref{lem:dctm_consistency_fun}, for each fixed $k$,
$\Delta_{Y,k}\xrightarrow{p}0$ and $\Delta_{W,k}\xrightarrow{p}0$. Since $K$ is fixed,
\[
\Delta_Y:=\max_{1\le k\le K}\Delta_{Y,k}\xrightarrow{p}0,
\qquad
\Delta_W:=\max_{1\le k\le K}\Delta_{W,k}\xrightarrow{p}0.
\]
Hence for any $i\in\mathcal I_k$,
\[
|\widehat U_i-U_i|
=
\Big|\widehat F_{Y\mid X}^{(-k)}(y_i\mid x_i)-F_{Y\mid X}(y_i\mid x_i)\Big|
\le \Delta_{Y,k}\le \Delta_Y,
\]
and therefore
\[
\mathbb{P}_n|\widehat U-U|
=
\frac1n\sum_{i=1}^n|\widehat U_i-U_i|
\le \Delta_Y\xrightarrow{p}0.
\]
Similarly, $\mathbb{P}_n|\widehat V-V|\xrightarrow{p}0$. Consequently,
\[
|\overline{\widehat U}-\overline U|
\le \mathbb{P}_n|\widehat U-U|\xrightarrow{p}0,
\qquad
|\overline{\widehat V}-\overline V|
\le \mathbb{P}_n|\widehat V-V|\xrightarrow{p}0.
\]
By the law of large numbers, $\overline U\xrightarrow{p} \mathbb{E}[U]=1/2$ and $\overline V\xrightarrow{p} \mathbb{E}[V]=1/2$,
which implies
\[
\overline{\widehat U}\xrightarrow{p}\tfrac12,
\qquad
\overline{\widehat V}\xrightarrow{p}\tfrac12.
\]

Using boundedness $\widehat U,U,\widehat V,V\in[0,1]$ and the triangle inequality, we obtain the pointwise bound
\[
\big|(\widehat U_i-\overline{\widehat U})(\widehat V_i-\overline{\widehat V})
-(U_i-\overline U)(V_i-\overline V)\big|
\le
C\Big(|\widehat U_i-U_i|+|\widehat V_i-V_i|
+|\overline{\widehat U}-\overline U|+|\overline{\widehat V}-\overline V|\Big),
\]
for some constant $C>0$. Averaging over $i$ yields
\[
\Big|
\mathbb{P}_n\!\left[(\widehat U-\overline{\widehat U})(\widehat V-\overline{\widehat V})\right]
-
\mathbb{P}_n\!\left[(U-\overline U)(V-\overline V)\right]
\Big|
\xrightarrow{p}0.
\]
Moreover, by the law of large numbers,
\[
\mathbb{P}_n\!\left[(U-\overline U)(V-\overline V)\right]\xrightarrow{p}\mathrm{Cov}(U,V),
\qquad
\mathbb{P}_n\!\left[(\widehat U-\overline{\widehat U})(\widehat V-\overline{\widehat V})\right]
\xrightarrow{p}\mathrm{Cov}(U,V).
\]
Similarly,
\[
\mathbb{P}_n\!\left[(\widehat V-\overline{\widehat V})^2\right]\xrightarrow{p}\mathrm{Var}(V)=\frac{1}{12}>0,
\qquad
\mathbb{P}_n\!\left[(\widehat U-\overline{\widehat U})^2\right]\xrightarrow{p}\mathrm{Var}(U)=\frac{1}{12}>0.
\]

Since the denominators converge to positive limits, the continuous mapping theorem implies
\begin{align*}
\widehat\rho_C
&\xrightarrow{p}\frac{\mathrm{Cov}(U,V)}{\mathrm{Var}(V)}=\rho_C,\\
\widetilde\rho_C
&\xrightarrow{p}\frac{\mathrm{Cov}(U,V)}{\sqrt{\mathrm{Var}(U)\mathrm{Var}(V)}}=\rho_C,\\
\breve\rho_C
&\xrightarrow{p}12\,\mathbb{E} \left[(U-\tfrac12)(V-\tfrac12)\right]=\rho_C.
\end{align*}
This completes the proof.
\end{proof}

\begin{proof}[Proof of Lemma~\ref{lem:ABlin_fun}]
We first handle $A_n$. For each observation $i$, let $k(i)$ denote the fold containing $i$. Under
Assumption~\ref{ass:emp_process}(ii), define
\[
\eta_{R,k}
:=
\sup_{(r,x)}
\Bigg|
\sqrt{n_{ck}}\big(\widehat F_{R\mid X}^{(-k)}(r\mid x)-F_{R\mid X}(r\mid x)\big)
-
\frac{1}{\sqrt{n_{ck}}}\sum_{j\in\mathcal I_k^c}\varphi_R(r,x;Z_j)
\Bigg|.
\]
Then for $i\in\mathcal I_k$,
\begin{align}
\widehat U_i-U_i
&=\frac{1}{n_{ck}}\sum_{j\in\mathcal I_k^c}\varphi_Y(Y_i,X_i;Z_j)+r_{Yi}, \label{eq:rank-lin-fun}\\
\widehat V_i-V_i
&=\frac{1}{n_{ck}}\sum_{j\in\mathcal I_k^c}\varphi_W(W_i,X_i;Z_j)+r_{Wi}, \label{eq:rank-lin-fun-1}
\end{align}
where the remainders satisfy $|r_{Yi}|\le \eta_{Y,k}/\sqrt{n_{ck}}$ and $|r_{Wi}|\le \eta_{W,k}/\sqrt{n_{ck}}$.

Averaging over the full sample yields
\[
\mathbb{P}_n(|r_Y|)
=
\frac{1}{n}\sum_{k=1}^K\sum_{i\in\mathcal I_k}|r_{Yi}|
\le
\frac{1}{n}\sum_{k=1}^K\sum_{i\in\mathcal I_k}\frac{\eta_{Y,k}}{\sqrt{n_{ck}}}
=
\sum_{k=1}^K\frac{n_k}{n}\cdot\frac{\eta_{Y,k}}{\sqrt{n_{ck}}}.
\]
Since $K$ is fixed and $n_k/n\to 1/K$ while $n_{ck}\asymp n$, we have
\[
\mathbb{P}_n(|r_Y|)
\le
C\cdot\frac{\max_k \eta_{Y,k}}{\sqrt{n}},
\]
for some constant $C>0$. Multiplying by $\sqrt n$ gives
\[
\sqrt n\,\mathbb{P}_n(|r_Y|)
\le
C\cdot\max_k \eta_{Y,k}.
\]
For each fixed $k$, $\eta_{Y,k}=o_p(1)$, and since $K$ is fixed, $\max_k \eta_{Y,k}=o_p(1)$. Hence
\[
\sqrt n\,\mathbb{P}_n(|r_Y|)=o_p(1).
\]
Similarly, $\sqrt n\,\mathbb{P}_n(|r_W|)=o_p(1)$.

Expand the product for $A_n$:
\begin{equation}\label{eq:An}
\begin{aligned}
(\widehat U_i-\tfrac12)(\widehat V_i-\tfrac12)
=&(U_i-\tfrac12)(V_i-\tfrac12)
+(V_i-\tfrac12)(\widehat U_i-U_i)\\
&+(U_i-\tfrac12)(\widehat V_i-V_i)
+R_{ni},
\end{aligned}
\end{equation}
where $R_{ni}:=(\widehat U_i-U_i)(\widehat V_i-V_i)$.

Using \eqref{eq:rank-lin-fun}--\eqref{eq:rank-lin-fun-1} and $n_{ck}\asymp n$, together with
$\sqrt n\,\mathbb{P}_n(|r_Y|)=o_p(1)$ and $\sqrt n\,\mathbb{P}_n(|r_W|)=o_p(1)$, Assumption~\ref{ass:emp_process}(ii)
implies
\[
\mathbb{P}_n\big[(\widehat U-U)^2\big]=O_p(n^{-1}),
\qquad
\mathbb{P}_n\big[(\widehat V-V)^2\big]=O_p(n^{-1}).
\]
By Cauchy--Schwarz,
\[
\mathbb{P}_n|R_n|
=
\mathbb{P}_n\big(|\widehat U-U|\,|\widehat V-V|\big)
\le
\sqrt{\mathbb{P}_n\big[(\widehat U-U)^2\big]}\,
\sqrt{\mathbb{P}_n\big[(\widehat V-V)^2\big]}
=
O_p(n^{-1}),
\]
and thus
\[
\sqrt n\,\mathbb{P}_n|R_n|=o_p(1).
\]

Apply $\mathbb{P}_n$ to \eqref{eq:An}, subtract $A$, and multiply by $\sqrt n$:
\begin{equation}\label{eq:A_term2_fun}
\begin{aligned}
\sqrt n(A_n-A)
=&\frac1{\sqrt n}\sum_{i=1}^n\Big((U_i-\tfrac12)(V_i-\tfrac12)-A\Big)\\
&+\sqrt n\,\mathbb{P}_n\Big[(V-\tfrac12)(\widehat U-U)\Big]\\
&+\sqrt n\,\mathbb{P}_n\Big[(U-\tfrac12)(\widehat V-V)\Big]
+o_p(1)\\
=&\frac1{\sqrt n}\sum_{i=1}^n\Big((U_i-\tfrac12)(V_i-\tfrac12)-A\Big)+Q_n+P_n+o_p(1).
\end{aligned}
\end{equation}
We only handle $Q_n$; the term $P_n$ is symmetric.

Write
\[
Q_n
=
\sqrt n\sum_{k=1}^K\frac{1}{n}\sum_{i\in\mathcal I_k}(V_i-\tfrac12)(\widehat U_i-U_i).
\]
By \eqref{eq:rank-lin-fun},
\[
Q_n
=
\sqrt n\sum_{k=1}^K\frac{1}{n}\sum_{i\in\mathcal I_k}(V_i-\tfrac12)
\left\{\frac{1}{n_{ck}}\sum_{j\in\mathcal I_k^c}\varphi_Y(Y_i,X_i;Z_j)\right\}
+o_p(1).
\]
Swap sums and rearrange:
\[
Q_n
=
\sqrt n\sum_{k=1}^K \frac{n_k}{n}\cdot\frac{1}{n_{ck}}
\sum_{j\in\mathcal I_k^c}
\left\{\frac{1}{n_k}\sum_{i\in\mathcal I_k}(V_i-\tfrac12)\varphi_Y(Y_i,X_i;Z_j)\right\}
+o_p(1).
\]
For each fixed $k$ and conditional on $Z_j$, the random variables
$\{(V_i-\tfrac12)\varphi_Y(Y_i,X_i;Z_j)\}_{i\in\mathcal I_k}$ are i.i.d.\ so the conditional law of large numbers
yields
\[
\frac{1}{n_k}\sum_{i\in\mathcal I_k}(V_i-\tfrac12)\varphi_Y(Y_i,X_i;Z_j)
\xrightarrow{p}
\mathbb{E} \big[(\widetilde V-\tfrac12)\varphi_Y(\widetilde Y,\widetilde X;Z_j)\mid Z_j\big]
=\gamma_Y(Z_j).
\]
Using $n_k/n\to 1/K$ and $n_{ck}/n\to 1-1/K$,
\[
Q_n
=
\frac{1}{K-1}\cdot\frac{1}{\sqrt n}\sum_{k=1}^K\sum_{j\in\mathcal I_k^c}\gamma_Y(Z_j)+o_p(1).
\]
Each observation $j$ belongs to exactly one fold $\mathcal I_m$ and is included in the double sum
$\sum_{k=1}^K\sum_{j\in\mathcal I_k^c}$ for all $k\neq m$, i.e., exactly $K-1$ times. Hence
\[
\sum_{k=1}^K\sum_{j\in\mathcal I_k^c}\gamma_Y(Z_j)=(K-1)\sum_{j=1}^n\gamma_Y(Z_j),
\]
and therefore
\[
Q_n=\frac{1}{\sqrt n}\sum_{j=1}^n\gamma_Y(Z_j)+o_p(1).
\]
The term $\sqrt n\,\mathbb{P}_n[(U-\tfrac12)(\widehat V-V)]$ is handled similarly, yielding
$\frac{1}{\sqrt n}\sum_{j=1}^n\gamma_W(Z_j)+o_p(1)$.

Substituting back into \eqref{eq:A_term2_fun} gives
\[
\sqrt n(A_n-A)
=
\frac{1}{\sqrt n}\sum_{i=1}^n\Big((U_i-\tfrac12)(V_i-\tfrac12)-A+\gamma_Y(Z_i)+\gamma_W(Z_i)\Big)+o_p(1),
\]
which is the asymptotic linear representation for $A_n$.

For $B_{n,U}$, expanding \eqref{eq:rank-lin-fun} yields
\[
(\widehat U_i-\tfrac12)^2
=
(U_i-\tfrac12)^2+2(U_i-\tfrac12)(\widehat U_i-U_i)+T_{ni},
\]
where $\sqrt n\,\mathbb{P}_n|T_n|=o_p(1)$. By the same argument as for $A_n$,
\[
\sqrt n(B_{n,U}-B_U)
=
\frac{1}{\sqrt n}\sum_{i=1}^n\Big((U_i-\tfrac12)^2-B_U+\delta_Y(Z_i)\Big)+o_p(1).
\]
Similarly,
\[
\sqrt n(B_{n,V}-B_V)
=
\frac{1}{\sqrt n}\sum_{i=1}^n\Big((V_i-\tfrac12)^2-B_V+\delta_W(Z_i)\Big)+o_p(1).
\]
Combining the three expansions yields the joint asymptotic linearity statement, completing the proof.
\end{proof}

\begin{proof}[Proof of Theorem~\ref{thm:asy_normal_fun}]
Note that $\widehat\rho_C$ and $\widetilde\rho_C$ are centered by the sample means, whereas $\breve\rho_C$ is
centered by $1/2$. For convenience, we first derive asymptotic linearity and normality under the $1/2$-centered
moment form and then show that the sample-mean-centered forms are asymptotically equivalent.

By Lemma~\ref{lem:ABlin_fun},
\[
\sqrt n\Big((A_n,B_{n,U},B_{n,V})-(A,B_U,B_V)\Big)
=
\frac1{\sqrt n}\sum_{i=1}^n\Psi(Z_i)+o_p(1),
\]
where $A=\rho_C/12$, $B_U=B_V=1/12$, and $\Psi(Z)=(\psi_A(Z),\psi_U(Z),\psi_V(Z))^\top$ has finite second moments.

The three estimators (in $1/2$-centered moment form) are
\[
\breve\rho_C=12A_n,
\qquad
\widehat\rho_C=\frac{A_n}{B_{n,V}},
\qquad
\widetilde\rho_C=\frac{A_n}{\sqrt{B_{n,U}B_{n,V}}}.
\]
Since $B_U=B_V=1/12>0$ and $(B_{n,U},B_{n,V})\xrightarrow{p}(B_U,B_V)$, the above mappings are differentiable in
a neighborhood of the truth and the Delta method applies.

For $\breve\rho_C$,
\[
\sqrt n(\breve\rho_C-\rho_C)
=
12\cdot \frac1{\sqrt n}\sum_{i=1}^n\psi_A(Z_i)+o_p(1),
\]
so the influence function is $\psi_{\mathrm{cov}}(Z)=12\psi_A(Z)$.

For $\widehat\rho_C=A_n/B_{n,V}$, let $g(a,b)=a/b$. At $(A,B_V)$,
\[
\partial_a g=\frac{1}{b}=12,
\qquad
\partial_b g=-\frac{a}{b^2}=-12\rho_C.
\]
Hence
\[
\sqrt n(\widehat\rho_C-\rho_C)
=
12\cdot \frac1{\sqrt n}\sum_{i=1}^n\Big(\psi_A(Z_i)-\rho_C\,\psi_V(Z_i)\Big)+o_p(1),
\]
so $\psi_{\mathrm{ols}}(Z)=12\{\psi_A(Z)-\rho_C\psi_V(Z)\}$.

For $\widetilde\rho_C=A_n/\sqrt{B_{n,U}B_{n,V}}$, let $h(a,b,c)=a(bc)^{-1/2}$. At $(A,B_U,B_V)$,
\[
\partial_a h=(B_UB_V)^{-1/2}=12,
\qquad
\partial_b h=-6\rho_C,
\qquad
\partial_c h=-6\rho_C.
\]
Therefore,
\[
\sqrt n(\widetilde\rho_C-\rho_C)
=
12\cdot \frac1{\sqrt n}\sum_{i=1}^n
\Big(\psi_A(Z_i)-\tfrac{\rho_C}{2}\big(\psi_U(Z_i)+\psi_V(Z_i)\big)\Big)+o_p(1),
\]
so $\psi_{\mathrm{corr}}(Z)=12\{\psi_A(Z)-\frac{\rho_C}{2}(\psi_U(Z)+\psi_V(Z))\}$.

Finally, the central limit theorem yields asymptotic normality and the variance expressions.

In the main text, $\widehat\rho_C$ and $\widetilde\rho_C$ are centered using sample means
$\overline{\widehat U}$ and $\overline{\widehat V}$. Note that
\[
\mathbb{P}_n\!\left[(\widehat U-\overline{\widehat U})(\widehat V-\overline{\widehat V})\right]
=
\mathbb{P}_n\!\left[(\widehat U-\tfrac12)(\widehat V-\tfrac12)\right]
-(\overline{\widehat U}-\tfrac12)(\overline{\widehat V}-\tfrac12).
\]
Since $\overline{\widehat U}-\tfrac12=O_p(n^{-1/2})$ and $\overline{\widehat V}-\tfrac12=O_p(n^{-1/2})$,
\[
\sqrt n\Big|
\mathbb{P}_n\!\left[(\widehat U-\overline{\widehat U})(\widehat V-\overline{\widehat V})\right]-A_n
\Big|
=
\sqrt n\,|\overline{\widehat U}-\tfrac12|\,|\overline{\widehat V}-\tfrac12|
=o_p(1).
\]
Similarly,
\[
\sqrt n\Big|\mathbb{P}_n\!\left[(\widehat V-\overline{\widehat V})^2\right]-B_{n,V}\Big|=o_p(1),
\qquad
\sqrt n\Big|\mathbb{P}_n\!\left[(\widehat U-\overline{\widehat U})^2\right]-B_{n,U}\Big|=o_p(1).
\]
Hence the sample-mean-centered and $1/2$-centered forms differ by $o_p(1)$ at the $\sqrt n$ scale, and therefore
share the same asymptotic distribution and variance expressions. This proves the theorem.
\end{proof}

\begin{proof}[Proof of Theorem~\ref{thm:boot_valid}]
By Theorem~\ref{thm:asy_normal_fun}, any estimator form $\ddot\rho_C$ admits an asymptotic linear representation
\[
\sqrt n(\ddot\rho_C-\rho_C)
=
\frac{1}{\sqrt n}\sum_{i=1}^n\psi_\rho(Z_i)+o_p(1),
\]
where $\psi_\rho(\cdot)$ is the corresponding influence function, satisfying $\mathbb{E}[\psi_\rho(Z)]=0$ and
$\mathbb{E}|\psi_\rho(Z)|^{2+\delta}<\infty$ for some $\delta>0$.

The bootstrap statistic $\ddot\rho_C^{*}$ is computed using a weighted-MLE fit of DCTM and weighted sample
moments. Combining Lemma~\ref{lem:ABlin_fun} and the derivation in Theorem~\ref{thm:asy_normal_fun} (replacing
the empirical average $\mathbb{P}_n$ with the weighted empirical average
$\mathbb{P}_n^{*}f:=n^{-1}\sum_{i=1}^n\widetilde\omega_{ni}f(Z_i)$, and noting that the within-fold/out-of-fold
structure remains valid under reweighting), we obtain the conditional asymptotic linear representation
\begin{equation}\label{eq:boot_ALR}
\sqrt n(\ddot\rho_C^{*}-\ddot\rho_C)
=
\frac1{\sqrt n}\sum_{i=1}^n(\widetilde\omega_{ni}-1)\psi_\rho(Z_i)+o_{P^*}(1),
\end{equation}
where $o_{P^*}(1)$ denotes convergence in probability under the weight distribution conditional on the data.

By \eqref{eq:boot_ALR}, it suffices to study
\[
S_n^*
:=
\frac1{\sqrt n}\sum_{i=1}^n(\widetilde\omega_{ni}-1)\psi_\rho(Z_i).
\]
Under Assumption~\ref{ass:emp_process}(ii) and the moment conditions implied by Theorem~\ref{thm:asy_normal_fun},
there exists $\delta>0$ such that $\mathbb{E}|\psi_\rho(Z)|^{2+\delta}<\infty$, and
$\mathbb{E}[\psi_\rho(Z)]=0$, $\mathrm{Var}(\psi_\rho(Z))=\sigma_\rho^2\in(0,\infty)$. Therefore, applying the
exchangeable bootstrap CLT in Lemma~\ref{lem:bootclt} with $\xi_i:=\psi_\rho(Z_i)$ yields
\[
S_n^*\ \rightsquigarrow_P\ \mathcal{N} (0,\sigma_\rho^2).
\]
Combining this with \eqref{eq:boot_ALR} and Slutsky's theorem gives
\[
\sqrt n(\ddot\rho_C^{*}-\ddot\rho_C)\ \rightsquigarrow_P\ \mathcal{N}(0,\sigma_\rho^2),
\]
so the bootstrap distribution consistently estimates the limiting distribution and the bootstrap-based standard
errors and confidence intervals are asymptotically valid. This completes the proof.
\end{proof}

\bibliographystyle{apa}

\bibliography{nkthesis}

@article{beller2006intergenerational,
  title={Intergenerational social mobility: The United States in comparative perspective},
  author={Beller, Emily and Hout, Michael},
  journal={The future of children},
  pages={19--36},
  year={2006},
  publisher={JSTOR}
}

@article{chetty2014land,
  title={Where is the land of opportunity? The geography of intergenerational mobility in the United States},
  author={Chetty, Raj and Hendren, Nathaniel and Kline, Patrick and Saez, Emmanuel},
  journal={The quarterly journal of economics},
  volume={129},
  number={4},
  pages={1553--1623},
  year={2014},
  publisher={MIT Press}
}

@incollection{spearman1961proof,
  title     = {The Proof and Measurement of Association between Two Things},
  author    = {Spearman, Charles},
  booktitle = {Studies in Individual Differences: The Search for Intelligence},
  editor    = {Jenkins, James J. and Paterson, Donald Gildersleeve},
  year      = {1961},
  pages     = {45--58},
  publisher = {Appleton-Century-Crofts},
  address   = {New York},
  doi       = {10.1037/11491-005}
}

@article{chetverikov2023inference,
  title={Inference for rank-rank regressions},
  author={Chetverikov, Denis and Wilhelm, Daniel},
  journal={arXiv preprint arXiv:2310.15512},
  year={2023}
}

@article{chernozhukov2024conditional,
  title={Conditional rank-rank regression},
  author={Chernozhukov, Victor and Fern{\'a}ndez-Val, Iv{\'a}n and Meier, Jonas and Van Vuuren, Aico and Vella, Francis},
  journal={arXiv preprint arXiv:2407.06387},
  year={2024}
}

@article{becker1986human,
  title={Human capital and the rise and fall of families},
  author={Becker, Gary S and Tomes, Nigel},
  journal={Journal of labor economics},
  volume={4},
  number={3, Part 2},
  pages={S1--S39},
  year={1986},
  publisher={University of Chicago Press}
}

@article{atkinson1980intergenerational,
  title={On intergenerational income mobility in Britain},
  author={Atkinson, Anthony B},
  journal={Journal of Post Keynesian Economics},
  volume={3},
  number={2},
  pages={194--218},
  year={1980},
  publisher={Taylor \& Francis}
}

@book{cramer1999mathematical,
  title={Mathematical methods of statistics},
  author={Cram{\'e}r, Harald},
  volume={9},
  year={1999},
  publisher={Princeton university press}
}

@article{kendall1948rank,
  title={Rank and product-moment correlation},
  author={Kendall, Maurice G},
  journal={Biometrika},
  pages={177--193},
  year={1949},
  publisher={JSTOR}
}

@inproceedings{baumann2021deep,
  title={Deep conditional transformation models},
  author={Baumann, Philipp FM and Hothorn, Torsten and R{\"u}gamer, David},
  booktitle={Joint European Conference on Machine Learning and Knowledge Discovery in Databases},
  pages={3--18},
  year={2021},
  organization={Springer}
}

@article{hothorn2014conditional,
  title={Conditional transformation models},
  author={Hothorn, Torsten and Kneib, Thomas and B{\"u}hlmann, Peter},
  journal={Journal of the Royal Statistical Society Series B: Statistical Methodology},
  volume={76},
  number={1},
  pages={3--27},
  year={2014},
  publisher={Oxford University Press}
}

@article{hothorn2018most,
  title={Most likely transformations},
  author={Hothorn, Torsten and M{\"o}st, Lisa and B{\"u}hlmann, Peter},
  journal={Scandinavian Journal of Statistics},
  volume={45},
  number={1},
  pages={110--134},
  year={2018},
  publisher={Wiley Online Library}
}

@article{kook2024estimating,
  title={Estimating conditional distributions with neural networks using r package deeptrafo},
  author={Kook, Lucas and Baumann, Philipp FM and D{\"u}rr, Oliver and Sick, Beate and R{\"u}gamer, David},
  journal={Journal of Statistical Software},
  volume={111},
  pages={1--36},
  year={2024}
}

@inproceedings{sick2021deep,
  title={Deep transformation models: Tackling complex regression problems with neural network based transformation models},
  author={Sick, Beate and Hathorn, Torsten and D{\"u}rr, Oliver},
  booktitle={2020 25th International Conference on Pattern Recognition (ICPR)},
  pages={2476--2481},
  year={2021},
  organization={IEEE}
}

@article{kook2022deep,
  title={Deep interpretable ensembles},
  author={Kook, Lucas and G{\"o}tschi, Andrea and Baumann, Philipp FM and Hothorn, Torsten and Sick, Beate},
  journal={arXiv preprint arXiv:2205.12729},
  year={2022}
}

@article{chernozhukov2018double,
  title={Double/debiased machine learning for treatment and structural parameters},
  author={Chernozhukov, Victor and Chetverikov, Denis and Demirer, Mert and Duflo, Esther and Hansen, Christian and Newey, Whitney and Robins, James},
  year={2018},
  journal = {The Econometrics Journal},
  pages = {C1–C68},
  volume = {21},
  number = {1}
}

@article{van1996weak,
  title={Weak convergence and empirical processes with applications to statistics},
  author={Vaart, AW van der and Wellner, Jon A},
  journal={Journal of the Royal Statistical Society-Series A Statistics in Society},
  volume={160},
  number={3},
  pages={596--608},
  year={1997},
  publisher={London: Royal Statistical Society, 1988-}
}

@article{ward2022internal,
  title={Internal migration, education, and intergenerational mobility: Evidence from american history},
  author={Ward, Zachary},
  journal={Journal of Human Resources},
  volume={57},
  number={6},
  pages={1981--2011},
  year={2022},
  publisher={University of Wisconsin Press}
}

@article{asher2024intergenerational,
  title={Intergenerational mobility in India: New measures and estimates across time and social groups},
  author={Asher, Sam and Novosad, Paul and Rafkin, Charlie},
  journal={American Economic Journal: Applied Economics},
  volume={16},
  number={2},
  pages={66--98},
  year={2024},
  publisher={American Economic Association 2014 Broadway, Suite 305, Nashville, TN 37203-2425}
}

@incollection{hoeffding1992class,
  title={A class of statistics with asymptotically normal distribution},
  author={Hoeffding, Wassily},
  booktitle={Breakthroughs in statistics: Foundations and basic theory},
  pages={308--334},
  year={1992},
  publisher={Springer}
}

@article{mogstad2023family,
  title={Family background, neighborhoods, and intergenerational mobility},
  author={Mogstad, Magne and Torsvik, Gaute},
  journal={Handbook of the Economics of the Family},
  volume={1},
  number={1},
  pages={327--387},
  year={2023},
  publisher={Elsevier}
}

@incollection{bottou2012stochastic,
  title={Stochastic gradient descent tricks},
  author={Bottou, L{\'e}on},
  booktitle={Neural networks: tricks of the trade: second edition},
  pages={421--436},
  year={2012},
  publisher={Springer}
}

@article{kingma2014adam,
  title={Adam: A method for stochastic optimization},
  author={Kingma, Diederik P and Ba, Jimmy},
  journal={arXiv preprint arXiv:1412.6980},
  year={2014}
}

@article{loshchilov2017decoupled,
  title={Decoupled weight decay regularization},
  author={Loshchilov, Ilya and Hutter, Frank},
  journal={arXiv preprint arXiv:1711.05101},
  year={2017}
}

@techreport{dahl2008association,
  title       = {The Association between Children's Earnings and Fathers' Lifetime Earnings: Estimates Using Administrative Data},
  author      = {Dahl, Molly and DeLeire, Thomas},
  institution = {Institute for Research on Poverty, University of Wisconsin--Madison},
  type        = {Discussion Paper},
  number      = {1342-08},
  year        = {2008},
  month       = aug,
  url         = {https://www.irp.wisc.edu/publications/dps/pdfs/dp134208.pdf}
}

@article{olivetti2015name,
  title={In the name of the son (and the daughter): Intergenerational mobility in the United States, 1850--1940},
  author={Olivetti, Claudia and Paserman, M Daniele},
  journal={American Economic Review},
  volume={105},
  number={8},
  pages={2695--2724},
  year={2015},
  publisher={American Economic Association 2014 Broadway, Suite 305, Nashville, TN 37203}
}

@article{fagereng2021wealthy,
  title={Why do wealthy parents have wealthy children?},
  author={Fagereng, Andreas and Mogstad, Magne and R{\o}nning, Marte},
  journal={Journal of Political Economy},
  volume={129},
  number={3},
  pages={703--756},
  year={2021},
  publisher={The University of Chicago Press Chicago, IL}
}

@article{halliday2021intergenerational,
  title={Intergenerational mobility in self-reported health status in the US},
  author={Halliday, Timothy and Mazumder, Bhashkar and Wong, Ashley},
  journal={Journal of Public Economics},
  volume={193},
  pages={104307},
  year={2021},
  publisher={Elsevier}
}

@article{song2020long,
  title={Long-term decline in intergenerational mobility in the United States since the 1850s},
  author={Song, Xi and Massey, Catherine G and Rolf, Karen A and Ferrie, Joseph P and Rothbaum, Jonathan L and Xie, Yu},
  journal={Proceedings of the National Academy of Sciences},
  volume={117},
  number={1},
  pages={251--258},
  year={2020},
  publisher={National Academy of Sciences}
}

@article{chetty2014united,
  title={Is the United States still a land of opportunity? Recent trends in intergenerational mobility},
  author={Chetty, Raj and Hendren, Nathaniel and Kline, Patrick and Saez, Emmanuel and Turner, Nicholas},
  journal={American economic review},
  volume={104},
  number={5},
  pages={141--147},
  year={2014},
  publisher={American Economic Association 2014 Broadway, Suite 305, Nashville, TN 37203}
}

@article{kotera2017educational,
  title={Educational policy and intergenerational mobility},
  author={Kotera, Tomoaki and Seshadri, Ananth},
  journal={Review of economic dynamics},
  volume={25},
  pages={187--207},
  year={2017},
  publisher={Elsevier}
}

@techreport{chetty2017mobility,
  title={Mobility report cards: The role of colleges in intergenerational mobility},
  author={Chetty, Raj and Friedman, John N and Saez, Emmanuel and Turner, Nicholas and Yagan, Danny},
  year={2017},
  institution={national bureau of economic research}
}

@article{fletcher2021intergenerational,
  title={Intergenerational health mobility: Magnitudes and importance of schools and place},
  author={Fletcher, Jason and Jajtner, Katie M},
  journal={Health economics},
  volume={30},
  number={7},
  pages={1648--1667},
  year={2021},
  publisher={Wiley Online Library}
}

@techreport{boar2021occupational,
  title={Occupational choice and the intergenerational mobility of welfare},
  author={Boar, Corina and Lashkari, Danial},
  year={2021},
  institution={National Bureau of Economic Research}
}

@article{chetty2016effects,
  title={The effects of exposure to better neighborhoods on children: New evidence from the moving to opportunity experiment},
  author={Chetty, Raj and Hendren, Nathaniel and Katz, Lawrence F},
  journal={American Economic Review},
  volume={106},
  number={4},
  pages={855--902},
  year={2016},
  publisher={American Economic Association 2014 Broadway, Suite 305, Nashville, TN 37203}
}

@article{mazumder2018intergenerational,
  title={Intergenerational mobility in the United States: What we have learned from the PSID},
  author={Mazumder, Bhashkar},
  journal={The Annals of the American Academy of Political and Social Science},
  volume={680},
  number={1},
  pages={213--234},
  year={2018},
  publisher={SAGE Publications Sage CA: Los Angeles, CA}
}

@article{foresi1995conditional,
  title={The conditional distribution of excess returns: An empirical analysis},
  author={Foresi, Silverio and Peracchi, Franco},
  journal={Journal of the American Statistical Association},
  volume={90},
  number={430},
  pages={451--466},
  year={1995},
  publisher={Taylor \& Francis}
}

@article{chernozhukov2013inference,
  title={Inference on counterfactual distributions},
  author={Chernozhukov, Victor and Fern{\'a}ndez-Val, Iv{\'a}n and Melly, Blaise},
  journal={Econometrica},
  volume={81},
  number={6},
  pages={2205--2268},
  year={2013},
  publisher={Wiley Online Library}
}

@article{box1964analysis,
  title={An analysis of transformations},
  author={Box, George EP and Cox, David R},
  journal={Journal of the Royal Statistical Society Series B: Statistical Methodology},
  volume={26},
  number={2},
  pages={211--243},
  year={1964},
  publisher={Oxford University Press}
}

@article{croix2024nepotism,
  title={Nepotism vs. intergenerational transmission of human capital in Academia (1088--1800)},
  author={Croix, David de la and Go{\~n}i, Marc},
  journal={Journal of Economic Growth},
  volume={29},
  number={4},
  pages={469--514},
  year={2024},
  publisher={Springer}
}

@article{solon1992intergenerational,
  title={Intergenerational income mobility in the United States},
  author={Solon, Gary},
  journal={The American Economic Review},
  pages={393--408},
  year={1992},
  publisher={JSTOR}
}

@article{corcoran1992association,
  title={The association between men's economic status and their family and community origins},
  author={Corcoran, Mary and Gordon, Roger and Laren, Deborah and Solon, Gary},
  journal={Journal of human resources},
  pages={575--601},
  year={1992},
  publisher={JSTOR}
}

@article{lee2009trends,
  title={Trends in intergenerational income mobility},
  author={Lee, Chul-In and Solon, Gary},
  journal={The review of economics and statistics},
  volume={91},
  number={4},
  pages={766--772},
  year={2009},
  publisher={The MIT Press}
}

@incollection{mazumder2015estimating,
  author    = {Mazumder, Bhashkar},
  title     = {Estimating the Intergenerational Elasticity and Rank Association in the United States: Overcoming the Current Limitations of Tax Data},
  booktitle = {Inequality: Causes and Consequences},
  editor    = {Cappellari, Lorenzo and Polachek, Solomon W. and Tatsiramos, Konstantinos},
  series    = {Research in Labor Economics},
  volume    = {43},
  pages     = {83--129},
  year      = {2016},
  publisher = {Emerald Group Publishing Limited},
  doi       = {10.1108/S0147-912120160000043012},
  url       = {https://doi.org/10.1108/S0147-912120160000043012}
}

@article{fisher2023intergenerational,
  title={Intergenerational Mobility using Income, Consumption, and Wealth},
  author={Fisher, Jonathan and Johnson, David},
  year={2023},
 journal = {SSRN},
 pages = {4435534}
}

@article{pfeffer2018generations,
  title={Generations of advantage. Multigenerational correlations in family wealth},
  author={Pfeffer, Fabian T and Killewald, Alexandra},
  journal={Social Forces},
  volume={96},
  number={4},
  pages={1411--1442},
  year={2018},
  publisher={Oxford University Press}
}

@article{emran2025gender,
  title={Is gender destiny? Gender bias and intergenerational educational mobility in India},
  author={Emran, M Shahe and Jiang, Hanchen and Shilpi, Forhad},
  journal={Journal of Economic Behavior \& Organization},
  volume={238},
  pages={107217},
  year={2025},
  publisher={Elsevier}
}

@article{ahsan2025ranks,
  title={When Ranks Fail: New Evidence on Intergenerational Educational Mobility},
  author={Ahsan, Md Nazmul and Emran, M Shahe and Mohammed, Shariq and Murphy, Orla and Shilpi, Forhad},
  year={2025},
  journal = {SSRN},
  pages = {5881282}
}

@article{asher2017estimating,
  title={Estimating intergenerational mobility with coarse data: A nonparametric approach},
  author={Asher, Sam and Novosad, Paul and Rafkin, Charlie},
  journal={World Bank. Technical report},
  year={2017}
}

@book{van2000asymptotic,
  title={Asymptotic statistics},
  author={Van der Vaart, Aad W},
  volume={3},
  year={2000},
  publisher={Cambridge university press}
}

@article{newey1994large,
  title={Large sample estimation and hypothesis testing},
  author={Newey, Whitney K and McFadden, Daniel},
  journal={Handbook of econometrics},
  volume={4},
  pages={2111--2245},
  year={1994},
  publisher={Elsevier}
}

@inproceedings{Ma2018MMoE,
  author    = {Jiaqi Ma and Zhe Zhao and Xinyang Yi and Jilin Chen and Lichan Hong and Ed H. Chi},
  title     = {Modeling Task Relationships in Multi-task Learning with Multi-gate Mixture-of-Experts},
  booktitle = {Proceedings of the 24th ACM SIGKDD International Conference on Knowledge Discovery and Data Mining},
  pages     = {1930--1939},
  year      = {2018},
  publisher = {ACM},
  doi       = {10.1145/3219819.3220007}
}

\end{document}